\documentclass{LMCS}

\def\dOi{11(1:14)2015}
\lmcsheading%
{\dOi}
{1--47}
{}
{}
{Feb.~28, 2014}
{Mar.~30, 2015}
{}

\ACMCCS{[{\bf Theory of computation}]: Semantics and reasoning---Program semantics}
\subjclass{F.3.2}

\usepackage{mathpartir}
\usepackage{url}

\usepackage{longtable}
\usepackage{wrapfig}
\usepackage{comment}

\usepackage{color}

\usepackage{graphicx,epsfig,color}

\usepackage{xypic,mathtools}
\input xy
\xyoption{all}
\usepackage{verbatim}
\usepackage{amsmath}
\usepackage{amssymb}
\usepackage{eucal}
\usepackage{hyperref}
\usepackage{bm}
\mathcode`:="003A  
\mathcode`;="003B  
\mathcode`?="003F  
\mathcode`|="026A  
\mathcode`<="4268  
\mathcode`>="5269  

\mathchardef\ls="213C    
\mathchardef\gr="213E    
\mathchardef\uparrow="0222  
\mathchardef\downarrow="0223  

\newcommand{\fF}{\mathcal{F}}    

\newcommand{\lbb}{\mathopen{[\![}}
\newcommand{\rbb}{\mathclose{]\!]}}
\newcommand{\bb}[1]{\lbb #1 \rbb}

\def\B{\mathcal B}

\newcommand\pto\rightharpoonup
\newcommand\E\varepsilon

\newenvironment{Iff-RL}{\textbf{($\Rightarrow$)} }{\bigskip}
\newenvironment{Iff-LR}{\textbf{($\Leftarrow$)} }{}

\newenvironment{remark}{\begin{rem}}{\end{rem}}
\newenvironment{theorem}{\begin{thm}}{\end{thm}}
\newenvironment{example}{\begin{exa}}{\end{exa}}
\newenvironment{definition}{\begin{defi}}{\end{defi}}
\newenvironment{proposition}{\begin{prop}}{\end{prop}}
\newenvironment{corollary}{\begin{cor}}{\end{cor}}

\theoremstyle{definition}\newtheorem{convention}[thm]{Convention}
\theoremstyle{definition}\newtheorem{construction}[thm]{Construction}

\newcommand{\m}{\mathit}
\renewcommand \: {\colon}
\newcommand{\mb}{\mathbb}
\newcommand{\mc}{\mathcal}
\newcommand{\mf}{\mathbf}
\def \p {\mathcal P}
\newcommand{\tx}{\texttt}

\newcommand\F{\mathcal{F}}
\newcommand\T{\mathcal{T}}

\newcommand\G{\mathcal{G}}

\newcommand\FS{\mathcal{S}}
\newcommand\FH{\mathcal{H}}
\newcommand\J{\mathcal{J}}
\newcommand\M{\mathcal{M}}

\def \incl {\iota} 

\def \catC {\mf{C}}
\def \catD {\mf{D}}
\def \Set {\mf{Set}}

\def \Poset { \mf{Poset}}

\newcommand{\coalg}[1]{{\bf CoAlg}(#1)}
\newcommand{\bialg}[1]{{\bf BiAlg}(#1)}

\def \Lw { \mf{L^{\m{op}}_{\Sigma}}}

\newcommand\U{\mathcal{U}}
\newcommand\K{\mathcal{K}}
\newcommand{\set}{\mf{Set}}
\newcommand{\prsh}[1]{\set^{#1}}

\def \At {\m{At}}
\newcommand{\seq}{\leftarrow} 

\def \predp {\texttt{p}} 
\def \predq {\texttt{q}} 

\def \NatList {\tx{NatList}}
\def \pList {\tx{List}}
\def \pNat {\tx{Nat}}
\def \tcons {cons}
\def \tzero {zero}
\def \tnil {nil}
\def \tsucc {succ}

\def \tconsZ {c}
\def \tzeroZ {z}
\def \tnilZ {n}
\def \tsuccZ {s}

\def\tr#1{\stackrel{#1}{\to}}          
\newcommand{\nodoor}[1]{*+[F-:<3pt>]{#1}}

\def \tuple#1{\dot{#1}}
\newcommand{\element}[2]{#2(#1)}

\def\QEQ{{%
    \setbox0\hbox{$\longrightarrow$}%
    \rlap{\hbox to \wd0{\hss $\bullet$\hss}}\box0
}}

\newcommand{\lconc}{\ensuremath{\,::\,}} 
\newcommand{\lconcc}{\ensuremath{\,\overline{::}\,}} 
\newcommand{\To}{\Rightarrow}
\newcommand{\li}{\mathcal{L}}
\newcommand{\cof}[1]{\mc{C}({#1})} 
\newcommand{\pg}[1]{{#1}_{\lambda}} 
\newcommand{\pgg}[1]{{#1}_{\lambda'}} 
\newcommand{\pa}[1]{{#1}_{\delta}} 
\newcommand{\paa}[1]{{#1}_{\delta'}} 

\newcommand{\liftf}[2]{{#1}^{{\scriptscriptstyle #2}}} 

\newcommand{\lift}[1]{\breve{#1}} 

\newcommand{\liftt}[1]{\widehat{#1}} 
\newcommand{\lifts}[1]{\ddot{#1}} 
\newcommand{\liftLax}[1]{\widetilde{#1}} 
\def \PP  {\liftt{\p _c}\liftt{\p _f}} 

\newcommand{\efuncat}[1]{{\bf End}({#1})} 

\newcommand{\EM}[1]{\mf{EM}(#1)} 

\newcommand{\id}{\m{id}}
\newcommand{\after}{\mathrel{\circ}}

\newcommand{\Kl}{\mathcal{K}{\kern-.2ex}\ell}
\newcommand{\cat}[1]{\ensuremath{\mf{#1}}}
\newcommand{\C}{\cat{C}}

\newcommand{\reprsat}{r_{\mathit{sat}}} 
\newcommand{\df}{\ensuremath{\,:\!\!=\,}}
\newcommand{\elist}{\ensuremath{[\,]}}
\newcommand{\KUP}{\mathcal{R}}

\newcommand{\desat}{\overline{d}} 

\newcommand{\tree}{T} 

\usepackage[layout=margin,final]{fixme}
\FXRegisterAuthor{fb}{afb}{Pip}
\FXRegisterAuthor{fz}{afz}{Fab}

\usepackage{semantic}
\usepackage{multirow}
\usepackage[table,xcdraw]{xcolor}

\newcommand{\acapo}[2]{ \parbox{#1}{\centering #2}}

\begin{document}

\title[Bialgebraic Semantics for Logic Programming]{Bialgebraic Semantics for Logic Programming}
\author[F.~Bonchi]{Filippo Bonchi}
\address{Ecole Normale Sup\'{e}rieure de Lyon, CNRS, Inria, UCBL, Universit\'{e} de Lyon --- Laboratoire de l'Informatique du Parall\'{e}lisme.}
\email{\{filippo.bonchi,fabio.zanasi\}@ens-lyon.fr}

\author[F.~Zanasi]{Fabio Zanasi}
\address{\vspace{-18 pt}}

\keywords{Logic Programming, Coalgebrae on presheaves, Bialgebrae, Compositionality}



\maketitle

\begin{abstract}
Bialgebrae provide an abstract framework encompassing the semantics of different kinds of computational models. 
In this paper we propose a bialgebraic approach to the semantics of logic programming. Our methodology is to study logic programs as reactive systems and exploit abstract techniques developed in that setting. First we use \emph{saturation} to model the operational semantics of logic programs as coalgebrae on presheaves. Then, we make explicit the underlying algebraic structure by using bialgebrae on presheaves. The resulting semantics turns out to be compositional with respect to conjunction and term substitution. Also, it encodes a parallel model of computation, whose soundness is guaranteed by a built-in notion of synchronisation between different threads.

\end{abstract}

\section{Introduction}

A fundamental tenet for the semantics of programming languages is \emph{compositionality}: the meaning of a program expression is reducible to the meaning of its subexpressions. This allows inductive reasoning on the structure of programs, and provides several techniques to prove
properties of these.

In the last decades, much research has been devoted to develop abstract frameworks for defining compositional semantics for different sorts of computational models. For instance, in the setting of concurrency theory, several rule formats \cite{aceto1999structural} have been introduced for ensuring certain behavioural equivalences to be congruences with respect to the syntactic operators of process algebrae. These works have inspired the Mathematical Operational Semantics by Turi and Plotkin \cite{TuriP97} where the semantics of a language is supposed to be a bialgebra for a given distributive law representing a so called \emph{abstract GSOS specification}.

A main drawback for this approach is its poor expressiveness, in the sense that many relevant computational models do not naturally fit into the bialgebraic framework. Prime examples of these, still from concurrency theory, are Petri nets and the calculus of Mobile Ambients \cite{CardelliG98}: defining the operational behaviour of these systems in terms of their components seems to be an intrinsically complex task that needs serious ingenuity of researchers (see, e.g, \cite{BruniMMS13} for a recent compositional semantics of Petri nets).

Motivated by these considerations, Milner initiated a research program devoted to systematically derive compositional semantics for more flexible semantics specification called \emph{reactive systems} \cite{LeiferMilnerCONCUR00}. As shown in \cite{DBLP:journals/tcs/BonchiM09}, these can be modeled as coalgebrae on presheaf categories and the resulting compositional semantics can be obtained by means of \emph{saturation}, a technique that we will detail later.

\medskip

In this paper we study \emph{logic programs as reactive systems}, applying the aforementioned techniques to obtain a compositional semantics for logic programming.
%
Our approach consists of two steps. First, we model the saturated semantics of a program by means of coalgebrae on presheaves. This allows us to achieve a first form of compositionality, with respect to the substitution of the logic signature. Then, we extend our approach to a bialgebraic setting, making the resulting semantics compositional also with respect to the internal structure of goals.


In the remainder of this introduction, we describe these two steps in more detail.

\subsection*{Coalgebrae on Presheaves and Saturated Semantics} Coalgebrae on presheaves have been successfully employed to provide semantics to \emph{nominal} calculi: sophisticated process calculi with complex mechanisms for variable binding, like the $\pi$-calculus \cite{FioreMS02,DBLP:conf/lics/FioreS06}. The idea is to have an index category $\catC$ of interfaces (or names), and encode as a presheaf $\F \: \catC \to \Set$ the mapping of any object $i$ of $\catC$ to the set of states having $i$ as interface, and any arrow $f \: i \to j$ to a function switching the interface of states from $i$ to $j$. The operational semantics of the calculus will arise as a notion of transition between states, that is, as a coalgebra $\alpha \: \F \to \B (\F)$, where $\B \colon \prsh{\catC} \to \prsh{\catC}$ is a functor on presheaves encoding the kind of behavior that we want to express.

As an arrow in a presheaf category, $\alpha$ has to be a natural transformation, i.e. it should commute with arrows $f \: i \to j$ in the index category $\catC$. Unfortunately, this naturality requirement may fail when the structure of $\catC$ is rich enough, as for instance when non-injective substitutions \cite{Marino05,DBLP:conf/calco/Staton09} or name fusions \cite{DBLP:journals/entcs/Miculan08,BonchiBCG_fusion2012} occur. 
As a concrete example, consider the $\pi$-calculus term $t=\bar{a}<x> | b(y) $ consisting of a process $\bar{a}<x>$ sending a message $x$ on a channel named $a$, in parallel with $b(y)$ receiving a message on a channel named $b$. Since the names $a$ and $b$ are different, the two processes cannot synchronize. Conversely the term $t\theta=\bar{a}<x> | a(y)$, that is obtained by applying the substitution $\theta$ mapping $b$ to $a$, can synchronize.
If $\theta$ is an arrow of the index category $\catC$, then the operational semantics $\alpha$ is not natural since $\alpha(t\theta)\neq \alpha(t)\overline{\theta}$, where $\overline{\theta}$ denotes the application of $\theta$ to the transitions of $t$. As a direct consequence, also the unique morphism to the terminal coalgebra is not natural: this means that the abstract semantics of $\pi$-calculus is not compositional, in the sense that bisimilarity is not a congruence w.r.t. name substitutions.
In order to make bisimilarity a congruence, Sangiorgi introduced in \cite{Sang96} \emph{open bisimilarity}, that is defined by considering the transitions of processes under \emph{all} possible name substitutions $\theta$.

The approach of \emph{saturated semantics} \cite{DBLP:journals/tcs/BonchiM09} can be seen as a generalization of open bisimilarity, relying on analogous principles: the operational semantics $\alpha$ is ``saturated'' w.r.t. the arrows of the index category $\catC$, resulting in a natural transformation $\alpha^{\sharp}$ in $\prsh{\catC}$.
In \cite{BonchiBCG_fusion2012,MontSamm_NetConsciousPi_CoalgSem}, this is achieved by first shifting the definition of $\alpha$ to the category $\prsh{|\catC|}$ of presheaves indexed by the discretization $|\catC|$ of $\catC$.
Since $|\catC|$ does not have other arrow than the identities, $\alpha$ is trivially a natural transformation in this setting. The source of $\alpha$ is $\U(\F) \in \prsh{|\catC|}$, where $\U \: \prsh{\catC} \to \prsh{|\catC|}$ is a forgetful functor defined by composition with the inclusion $\iota \: |\catC| \to \catC$. The functor $\U$ has a right adjoint $\K\: \prsh{|\catC|} \to \prsh{
\catC}$ sending a presheaf to its \emph{right Kan extension} along $\iota$.
The adjoint pair $\U \dashv \K$ induces an isomorphism $(\cdot)_{X,Y}^{\sharp}\colon \prsh{|\catC|}[\U (X),Y] \to \prsh{\catC}[X,\K (Y)]$ mapping $\alpha$ to $\alpha^{\sharp}$. The latter is a natural transformation in $\prsh{\catC}$ and, consequently, the abstract semantics results to be compositional.

\smallskip

In the first part of the paper, we show that the saturated approach can be fruitfully instantiated to \emph{coalgebraic logic programming} \cite{KomMcCuskerPowerAMAST10,KomPowCALCO11,KomPowerCSL11,KatyaJournal}, which consists of a novel semantics for logic programming and a parallel resolution algorithm based on \emph{coinductive trees}. These are a variant of $\wedge\vee$-trees \cite{Gupta94} modeling \emph{parallel} implementations of logic programming, where the soundness of the derivations represented by a tree is guaranteed by the restriction to \emph{term-matching} (whose algorithm, differently from unification, is parallelizable~\cite{MitchellSeqUnification}).

There are two analogies with the $\pi$-calculus: (a) the state space is modeled by a presheaf on the index category $\Lw$, that is the (opposite) \emph{Lawvere Theory} associated with some signature $\Sigma$; (b) the operational semantics given in \cite{KomPowCALCO11} fails to be a natural transformation in $\prsh{\Lw}$: Example \ref{Ex:non_compositional} provides a counter-example which is similar to the $\pi$-calculus term $t$ discussed above.

The authors of \cite{KomPowCALCO11} obviate to (b) by relaxing naturality to \emph{lax naturality}:
the operational semantics $p$ of a logic program is given as an arrow in the category $\m{Lax}(\Lw,\Poset)$ of locally ordered functors $\F \: \Lw \to \Poset$ and lax natural transformations between them. They show the existence of a cofree comonad
that induces a morphism $\bb{\cdot}_p$ mapping atoms (i.e., atomic formulae) to coinductive trees. Since $\bb{\cdot}_p$ is not natural but lax natural, the semantics provided by coinductive trees is not compositional, in the sense that, for some atoms $A$ and substitution $\theta$,
  $$\bb{A \theta}_p \neq \bb{A}_p\overline{\theta}$$
where $\bb{A \theta}_p$ is the coinductive tree associated with $A\theta$ and $\bb{A}_p\overline{\theta}$ denotes the result of applying $\theta$ to each atom occurring in the tree $\bb{A}_p$.

\smallskip

Instead of introducing laxness, we propose to tackle the non-naturality of $p$ with a saturated approach. It turns out that, in the context of logic programming, the saturation map $(\cdot)^{\sharp}$ has a neat description in terms of substitution mechanisms: while $p$ performs \emph{term-matching} between the atoms and the heads of clauses of a given logic program, its saturation $p^{\sharp}$ (given as a coalgebra in $\prsh{\Lw}$) performs \emph{unification}. It is worth to remark here that not only most general unifiers are considered but \emph{all} possible unifiers.

A cofree construction leading to a map $\bb{\cdot}_{p^{\sharp}}$ can be obtained by very standard categorical tools, such as terminal sequences \cite{ak:fixed-point-set-functor}. This is possible because, as $\Set$, both $\prsh{\Lw}$ and $\prsh{|\Lw|}$ are (co)complete categories, whereas in the lax approach, $\m{Lax}(\Lw,\Poset)$ not being (co)complete, more indirect and more sophisticated categorical constructions are needed \cite[Sec. 4]{KomPowerCSL11}. By naturality of $p^{\sharp}$, the semantics given by $\bb{\cdot}_{p^{\sharp}}$ turns out to be compositional, as in the desiderata. Analogously to $\bb{\cdot}_{p}$, also $\bb{\cdot}_{p^{\sharp}}$ maps atoms to tree structures, which we call \emph{saturated $\wedge\vee$-trees}. They generalize coinductive trees, in the sense that the latter can be seen as a ``desaturation'' of saturated $\wedge\vee$-trees, where all unifiers that are not term-matchers have been discarded. This observation leads to a \emph{translation} from saturated to coinductive
trees, based on the counit $\epsilon$ of the adjunction $\U \dashv \K$. It follows that our framework encompasses the semantics in \cite{KomPowCALCO11,KomPowerCSL11}.

   Analogously to what is done in \cite{KomPowerCSL11}, we propose a notion of \emph{refutation subtree} of a given saturated $\wedge\vee$-tree, intuitively corresponding to an SLD-refutation of an atomic goal in a program.\fznote{Parte per Katia} In our approach, not all the refutation subtrees represent sound derivations, because the same variable may be substituted for different terms in the various branches. We thus study the class of \emph{synched} refutation subtrees: they are the ones in which, at each step of the represented derivation, \emph{the same} substitution is applied on all the atoms considered on different branches. Refutation subtrees with this property do represent sound derivations and are preserved by the desaturation procedure. This leads to a result of soundness and completeness of our semantics with respect to SLD-resolution, crucially using both compositionality and the translation into coinductive trees.

\subsection*{Bialgebraic Semantics of Goals} In the second part of this paper we extend our framework to model the saturated semantics of goals instead of single atoms. This broadening of perspective is justified by a second form of compositionality that we want to study. Given atoms $A$ and $B$, one can see the goal $\{A,B\}$ as their conjunction $A \wedge B$: the idea is that a resolution for $A \wedge B$ requires a resolution for $A$ \emph{and} one for $B$. Our aim is to take this logical structure into account, proposing a semantics for $A \wedge B$ that can be equivalently given as the ``conjunction'' (at a higher level) of the semantics $\bb{A}_{p^{\sharp}}$ and $\bb{B}_{p^{\sharp}}$. Formally, we will model the structure given by $\wedge$ as an \emph{algebra} on the space of goals. To properly extend our saturated approach, the coalgebra $p^{\sharp}$ encoding a logic program needs to be compatible with such algebraic structure: the formal ingredient to achieve this will be a \emph{distributive law} $\delta$ involving the type of the algebra of goals and the one of the coalgebra $p^{\sharp}$. This will give raise to an extension of $p^{\sharp}$ to a $\delta$-bialgebra $\pa{p}^{\sharp}$ on the space of goals, via a categorical construction that is commonplace in computer science. For instance, when instantiated to a non-deterministic automaton $t$, it yields the well-known powerset construction $\overline{t}$ on $t$: the states of $t$ are like atoms of a program, and the ones of $\overline{t}$ are collections of states, like goals are collections of atoms.

Thanks to the compatibility of the operational semantics $\pa{p}^{\sharp}$, the induced map $\bb{\cdot}_{\pa{p}^{\sharp}}$ will also be compatible with the algebraic structure of goals, resulting in the compositionality property sketched above,
 \[ \bb{A \wedge B}_{\pa{p}^{\sharp}} = \bb{A}_{\pa{p}^{\sharp}}\ \overline{\wedge}\ \bb{B}_{\pa{p}^{\sharp}}\]
 where on the right side $\bb{\cdot}_{\pa{p}^{\sharp}}$ is applied to goals consisting of just one atom, $A$ or $B$. As we did for $\bb{\cdot}_{p^{\sharp}}$, we will represent the targets of $\bb{\cdot}_{\pa{p}^{\sharp}}$ as trees, which we call \emph{saturated $\vee$-trees}. Like saturated $\wedge\vee$-trees, they encode derivations by (generalized) unification, which now apply to goals instead of single atoms. This operational understanding of $\bb{\cdot}_{\pa{p}^{\sharp}}$ allows for a more concrete grasp on the higher order conjunction $\overline{\wedge}$: it can be seen as the operation of ``gluing together'' (depthwise) saturated $\vee$-trees.

 Beyond this form of compositionality, our approach exhibits another appealing feature, arising by definition of $\delta$. When solving a goal $G$ in a program $\mb{P}$, the operational semantics $\pa{p}^{\sharp}$ attempts to perform unification \emph{simultaneously} on each atom in $G$ with heads in $\mb{P}$, by applying \emph{the same} substitution on all atoms in the goal. If this form of ``synchronous'' resolution is not possible --- for instance, if there is no unique substitution for all atoms in $G$ --- then the computation encoded by $\pa{p}^{\sharp}$ ends up in a failure.

 This behavior is rather different from the one of the standard SLD-resolution algorithm for logic programming, where unification is performed \emph{sequentially} on one atom of the goal at a time. However, we are able to show that the semantics $\bb{\cdot}_{\pa{p}^{\sharp}}$ is \emph{sound and complete} with respect to SLD-resolution, meaning that no expressivity is lost in assuming a synchronous derivation procedure as the one above. This result extends the completeness theorem, previously stated for the semantics $\bb{\cdot}_{p^{\sharp}}$, from atomic to arbitrary goals. It relies on a notion of refutation subtree for saturated $\vee$-trees that is in a sense more natural than the one needed for saturated $\wedge\vee$-trees. Whereas in the latter we had to impose specific constraints to ensure that refutation subtrees only represent sound derivations, there is no such need in saturated $\vee$-trees, because the required synchronisation property is guaranteed by construction.

\subsection*{A Fistful of Trees} The following table summarises how the saturated derivation trees that we introduce compare with the trees appearing in the logic programming literature. 

\begin{table}[h]
\begin{tabular}{|p{2.4cm}|c|c|c|}
\hline
                                                                                             & \multicolumn{3}{c|}{\cellcolor[HTML]{C0C0C0}\textbf{substitution mechanism}}                                                                              \\ \cline{2-4}
\multirow{-2}{*}{}                                                                                   & \cellcolor[HTML]{9B9B9B}\textbf{most general unification} & \cellcolor[HTML]{9B9B9B}\textbf{term-matching} & \cellcolor[HTML]{9B9B9B}\textbf{unification} \\ \hline
\multicolumn{1}{|c|}{\cellcolor[HTML]{9B9B9B}{\color[HTML]{333333} \textbf{\acapo{2.4cm}{nodes are \\ atoms}}}} & \acapo{4.5cm}{$\wedge\vee$-trees \\  (Def.~\ref{DEF:and-or_par_tree_ground}, also called parallel and-or trees~\cite{Gupta94})}    & \acapo{3.1cm}{coinductive trees\\(Def.~\ref{Def:coinductive_trees_Power}, \cite{KomPowerCSL11,KatyaJournal})}         & \acapo{3.2cm}{saturated $\wedge\vee$-trees \\ (Def.~\ref{def:saturatedtree})}             \\ \hline
\multicolumn{1}{|c|}{\cellcolor[HTML]{9B9B9B}{\color[HTML]{333333} \textbf{\acapo{2.4cm}{nodes are \\ goals}}}}        & SLD-trees (e.g.~\cite{Lloyd93})                                                 &                                                & \acapo{3cm}{saturated $\vee$-trees \\ (Def.~\ref{DEF:or_par_tree_general})}                       \\ \hline
\end{tabular}
\end{table}

The computation described by SLD-trees is inherently sequential, whereas all the others in the table exhibit and-or parallelism (\emph{cf.} Section~\ref{ssec:lpbackground}). More specifically, (saturated) $\wedge\vee$-trees and coinductive trees express \emph{independent} and-parallelism: there is no exchange of information between the computations involving different atoms in the goal. Instead, saturated $\vee$-trees encode a dependent form of and-parallelism: at each step \emph{every} atom of the goal is matched with heads of the program, but they have to agree on the same substitution, thus requiring a form of communication between different threads.
Since independent and-parallelism does not cohere well with unification, (saturated) $\wedge\vee$-trees may represent unsound derivations. Instead, they are always sound by construction in coinductive trees (because of the restriction to term-matching) and saturated $\vee$-trees (because the atoms in the goal are processed synchronously, by applying the same substitution).

\subsection*{Related works.}

Our starting point is the key observation that, in coalgebraic logic programming \cite{KomMcCuskerPowerAMAST10,KomPowCALCO11,KomPowerCSL11,KatyaJournal}, the operational semantics fails to be a natural transformation. As an alternative to the lax approach of the above line of research, we propose saturation which, in the case of logic programming, boils down to unification with respect to all substituions, making the whole approach closer to standard semantics, like Herbrand models~\cite{VanEmden1976_LP,clark1980predicate} (where only ground instances are considered) or the $C$-semantics of \cite{FalaschiPalamidessi_LP} (considering also non-ground instances).

As a result, our approach differs sensibly from \cite{KomMcCuskerPowerAMAST10,KomPowCALCO11,KomPowerCSL11,KatyaJournal}: in that series of works, the aim is to give an operational semantics to coinductive logic programs and, at the same time, to exploit and-or parallelism. In our case, the semantics is only meant to model standard (recursive) logic programs and the synchronisation mechanism built-in in saturated $\vee$-trees imposes a form of dependency to and-parallelism.

The two forms of compositionality that we investigate appear in various forms in the standard literature, see e.g. \cite{Apt:AndComp_LP,Lloyd93}. Our interest is to derive them as the result of applying the categorical machinery --- coalgebrae and bialgebrae on presheaves --- that has been fruitfully adopted in the setting of process calculi, as detailed above. For instance, in the open $\pi$-calculus one aims at the same two forms of compositionality: with respect to name substitution --- corresponding to term substitutions in logic programming --- and parallel composition of processes --- corresponding to conjunction of atoms in a goal.

We should also mention other categorical perspectives on (extensions of) logic programming, such as~\cite{Corradini199251,Kinoshita:1996,journals/tcs/AmatoLM09,DBLP:journals/tcs/BonchiM09}. Amongst these, the most relevant for us is \cite{DBLP:journals/tcs/BonchiM09} since it achieves compositionality with respect to substitutions and $\wedge$ by exploiting a form of saturation:
arrows of the index category are both substitutions and $\wedge$-contexts of the shape $G_1 \wedge - \wedge G_2$ (for some goals $G_1, G_2$).
The starting observation in \cite{DBLP:journals/tcs/BonchiM09} is that the construction of relative pushouts \cite{LeiferMilnerCONCUR00} instantiated to such category captures the notion of most general unifiers.

Beyond logic programming, the idea of using saturation to achieve compositionality is even older than \cite{Sang96} (see e.g. \cite{Montanari:92:FI}). As far as we know, \cite{Corradini1999118} is the first work where saturation is explored in terms of coalgebrae. It is interesting to note that, in \cite{DBLP:journals/tcs/CorradiniGH01}, some of the same authors also proposed laxness as a solution for the lack of compositionality of Petri nets.

A third approach, alternative to laxness and saturation, may be possible by taking a special kind of ``powerobject'' functor as done in \cite{DBLP:journals/entcs/Miculan08,DBLP:conf/calco/Staton09} for giving a coalgebraic semantics to fusion and open $\pi$-calculus.
%
%
We have chosen saturated semantics for its generality: it works for any behavioral functor $\B$ and it models a phenomenon that occurs in many different computational models (see e.g. \cite{DBLP:conf/calco/BonchiM09}).

Finally, the approach consisting in saturating and building a bialgebraic model already appeared in \cite{FioreTuri} (amongst others). In that work, a sort of saturated semantics is achieved by transposing along the adjunction between $\prsh{\mf{I}}$ and $\prsh{\mf{F}}$ obtained from the injection of the category $\mf{I}$ of finite sets and injective functions into the category $\mf{F}$ of all functions.
An interesting construction, missing in \cite{FioreTuri}, is the one of the distributive law $\delta$ for saturated semantics: in our work, $\delta$ is built in a canonical way out of the distributive law for the non-saturated semantics.

   \subsection*{Synopsis} After introducing the necessary background in Section~\ref{SEC:BackgroundJournal}, we recall the framework of coalgebraic logic programming of \cite{KomMcCuskerPowerAMAST10,KomPowCALCO11,KomPowerCSL11,KatyaJournal} in Section \ref{SEC:Background}. In Section~\ref{SEC:SemLogProg} we propose saturated semantics, allowing us to achieve the first compositionality property. In Section~\ref{SEC:Desaturation} we compare saturated semantics with the lax approach of \cite{KomPowCALCO11}. This is instrumental for proving, in Section~\ref{SEC:Completeness}, soundness and completeness of saturated semantics with respect to SLD-resolution on atomic goals.

   In the second part of the paper we present the bialgebraic semantics for arbitrary goals. We start in Section~\ref{sec:parground} and \ref{sec:complparsem_ground} by considering the simpler setting of ground logic programs. In particular, Section~\ref{sec:parground} shows the second compositionality property and Section~\ref{sec:complparsem_ground} draws a comparison with the coalgebraic semantics.
   Section~\ref{sec:pargeneral} and Section~\ref{sec:complparsem_gen} generalize the results of the previous two sections to arbitrary logic programs. In particular, we conclude Section~\ref{sec:complparsem_gen} by proving soundness and completeness of the bialgebraic semantics of goals with respect to SLD-resolution, extending the analogous result for atomic goals in Section~\ref{SEC:Completeness}.

The present work extends the conference paper \cite{DBLP:conf/calco/BonchiZ13} with more examples, proofs and the new material presented in Sections \ref{sec:parground}, \ref{sec:complparsem_ground}, \ref{sec:pargeneral} and \ref{sec:complparsem_gen}.


\subsection*{Acknowledgements.} We thank E. Komendantskaya, T. Hirschowitz, U. Montanari, D. Petrisan, J. Power, M. Sammartino and the anonymous referees for the helpful comments. We acknowledge support by project ANR~12IS02001~PACE.

\section{Background}\label{SEC:BackgroundJournal}

In this section we fix some terminology, notation and basic results, mainly concerning category theory and logic programming.

\subsection{Categories and Presheaves}

Given a (small) category $\catC$, $|\catC|$ denotes the category with the same objects as $\catC$ but no other arrow than the identities. With a little abuse of notation, $X \in \catC$ indicates that $X$ is an object of $\catC$ and $\catC [X,Y]$ the set of arrows from $X$ to $Y$. We denote with $X \times Y$ the product of objects $X,Y\in  \catC$ with projections $\pi_1 \: X \times Y \to X$ and $\pi_2 \: X \times Y \to Y$. Given $Z \in \catC$ and arrows $f \: Z \to X$ and $g\: Z \to Y$, we denote with $<f,g> \: Z \to X \times Y$ the arrow given by universal property of $X \times Y$. We use the notation $\efuncat{\C}$ for the category of endofunctors on $\C$ and natural transformations.
 A $\catC$-indexed \emph{presheaf} is any functor $\G\colon \catC \to \Set$. We write $\prsh{\catC}$ for the category of $\catC$-indexed presheaves and natural transformations.

 Throughout this paper we will need to extend functors on $\Set$ to functors on presheaf categories. For this purpose, it will be useful to state the following construction.

 \begin{definition}\label{def:liftpreshgen} Given a category $\C$, the functor $\liftf{(\cdot)}{\C} \: \efuncat{\set} \to \efuncat{\prsh{\C}}$ is defined as follows.
 \begin{itemize}
   \item Given an object $\F \: \Set \to \Set$, $\liftf{\F}{\C} \: \prsh{\C} \to \prsh{\C}$ is defined on $\G \: \C \to \Set$ as $\F \after \G$ and on $\alpha \: \G \To \FH$ as the natural transformation given by $\Big(\F\G(n) \xrightarrow{\F(\alpha_n)} \F \FH(n) \Big)_{n \in \C}$.
   \item Given an arrow $\gamma \: \F \To \B$ of $\efuncat{\set}$, $\liftf{\gamma}{\C} \: \liftf{\F}{\C} \To \liftf{\B}{\C}$ is a natural transformation $\beta$ given by the family $\Big(\F\G \xrightarrow{{\beta}_{\G}} \B \G \Big)_{\G \in \prsh{\C}}$ of natural transformations, where each ${\beta}_{\G}$ is defined by $\Big(\F\G(n) \xrightarrow{\gamma_{\G(n)}} \B \G(n) \Big)_{n \in \C}$.
 \end{itemize}
 We call $\liftf{\F}{\C}$ and $\liftf{\gamma}{\C}$ \emph{extensions} of $\F$ and $\gamma$ respectively.
 \end{definition}

 We will mainly work with categories of presheaves indexed by $\Lw$ --- the (opposite) Lawvere theory on a signature $\Sigma$, defined in Section \ref{ssec:termsAtomsSubstitutions} --- and its discretization $|\Lw|$. To simplify notation, we will follow the convention of writing $\lift{(\cdot)}$ for the extension functor $\liftf{(\cdot)}{\Lw}$, i.e., $\liftf{(\cdot)}{\C}$ where $\C = \Lw$, and $\liftt{(\cdot)}$ for $\liftf{(\cdot)}{|\Lw|}$.

\subsection{Monads and Distributive Laws}

A \emph{monad} in a category $\C$ is a functor $\T\colon\catC \rightarrow \catC$ together with two natural transformations $\eta\colon \id
\To \T$ and $\mu \colon {\T\T} \To
\T$, called respectively unit and multiplication of $\T$, which are required to satisfy the following equations for any $X \in \catC$: $\mu_X\after\eta_{\T X} = \mu_X\after \T\eta_{X} = \id_X$ and $\T\mu_{X}\after \mu_{X X} = \mu_X \after \mu_X$. We shall also make use of the triple notation $(\T,\eta,\mu)$ for monads.
A distributive law of a {\em monad} $(\T,\eta^\T,\mu^\T)$ over a {\em monad} $(\M,\eta^\M,\mu^\M)$ is a
natural transformation $\lambda \colon
\T\M\To \M\T$ making the diagrams below commute:
\begin{equation*}
\vcenter{\xymatrix@R-.5pc{
\T X\ar[d]_{\T(\eta^\M_{X})}\ar@{=}[r] & \T X\ar[d]^{\eta^\M_{\T X}}
& \quad
\T\M^{2}X\ar[d]_{\T\mu^\M_{X}}\ar[r]^-{\lambda_{\M X}} &
   \M\T\M X\ar[r]^-{\M\lambda_{X}} &
   \M^{2}\T X\ar[d]^{\mu^\M_{\T X}} \\
\T\M X\ar[r]_-{\lambda_X} & \M\T X
& \quad
\T\M X\ar[rr]_-{\lambda_X} & & \M\T X\\
\M X\ar[u]^{\eta^\T_{\M X}}\ar@{=}[r] & \M X\ar[u]_{\M\eta^\T_{X}}
& \quad
\T^{2}\M X\ar[u]^{\mu^\T_{\M X}}\ar[r]_{\T\lambda_{\M X}} &
   \T\M\T X\ar[r]_{\lambda_{\T X}} &
\M\T^{2} X\ar[u]_{\M\mu^\T_{X}}  }}
\end{equation*}
A distributive law $\lambda \colon
\T\M\To \M\T$ between monads yields a monad $\M\T$ with unit $\eta^{\M\T} \df \eta^{\M}_{\T}\after\eta^{\T}$ and multiplication $\mu^{\M\T} \df \M\mu^{\T} \after \mu^{\M}_{\T\T} \after \M\lambda_{\T}$. We introduce now two weaker notions of the law. A distributive law of a {\em monad} $(\T,\eta^\T,\mu^\T)$ over a {\em functor } $\M$ is a
natural transformation $\lambda \colon
\T\M \To \M\T$ such that only the two bottommost squares above commute. One step further in generalization, we call a distributive law of a {\em functor } $\T$ over a {\em functor} $\M$ any natural transformation from $\T\M$ to $\M\T$.

 With the next proposition we observe that the extension functor $\liftf{(\cdot)}{\C}$ (Definition \ref{def:liftpreshgen}) preserves the monad structure.

\begin{proposition}\label{prop:liftpreservemonads} If $(\T,\eta,\mu)$ is a monad in $\Set$ then $(\liftf{\T}{\C},\liftf{\eta}{\C},\liftf{\mu}{\C})$ is a monad in $\prsh{\C}$. Moreover, if $\lambda \: \T\M \to \M\T$ is a distributive law of monads in $\Set$ then $\liftf{\lambda}{\C} \: \liftf{\T}{\C}\liftf{\M}{\C} \to \liftf{\M}{\C}\liftf{\T}{\C}$ is a distributive law of monads in $\prsh{\C}$.
\end{proposition}
\begin{proof} The action of $\liftf{(\cdot)}{\C}$ on arrows of $\efuncat{\set}$ is given componentwise. This means that commutativity of the diagrams involving  $\liftf{\eta}{\C}$, $\liftf{\mu}{\C}$ and $\liftf{\lambda}{\C}$ in $\prsh{\C}$ follows by the one of the corresponding diagrams in $\Set$.
\end{proof}

\subsection{Algebrae, Coalgebrae and Bialgebrae}

Given a functor $\F \: \catC \to \catC$, a \emph{$\F$-algebra} on $X \in \catC$ is an arrow $h \: \F(X) \to X$, also written as a pair $(X,h)$. A morphism between $\F$-algebrae $(X,h)$ and $(Y,i)$ is an arrow $f \: X \to Y$ such that $f \circ h = i \circ \F(f)$.

Dually, a \emph{$\B$-coalgebra} on $X \in \catC$ is an arrow $p \: X \to \B (X)$, also written as a pair $(X,p)$. A morphism between $\B$-coalgebrae $(X,p)$ and $(Y,q)$ is an arrow $g \: X \to Y$ such that $q \circ g = \B(g) \circ p$. We fix notation $\coalg{\B}$ for the category of $\B$-coalgebrae and their morphisms. If it exists, the \emph{final $\B$-coalgebra} is the terminal object in $\coalg{\B}$. The \emph{cofree $\B$-coalgebra} on $X \in \C$ is given by $(\Omega,\pi_2 \circ \omega \: \Omega \to \B(\Omega))$, where $(\Omega,\omega)$ is the terminal object in $\coalg{X \times \B(\cdot)}$.

Let $\lambda \: \F\B \To \B\F$ be a distributive law between functors $\F,\B \: \C \to \C$. A {\em $\lambda$-bialgebra} is a triple $(X,h,p)$ where $h \: \F X \to X$ is an $\F$-algebra and $p \: X \to \B X$ is a $\B$-coalgebra subject to the compatibility property given by commutativity of the following diagram:
\[
\xymatrix{
\F X \ar[r]^h \ar[d]_{\F p} & X \ar[r]^p & \B X \\
\F \B X \ar[rr]_{\lambda_X} && \B \F X \ar[u]_{\B h}
}
\]
A \emph{$\lambda$-bialgebra morphism} from $(X,h,p)$ to $(Y,i,q)$ is an arrow $f \: X \to Y$ that is both a $\B$-coalgebra morphism from $(X,p)$ to $(Y,q)$ and an $\F$-algebra morphism from $(X,h)$ to $(Y,i)$. We fix notation $\bialg{\lambda}$ for the category of $\lambda$-bialgebrae and their morphisms. If it exists, the \emph{final $\lambda$-bialgebra} is the terminal object in $\bialg{\lambda}$.

Next we record some useful constructions of bialgebrae out of coalgebrae. When $\F$ is a monad and $\lambda$ is a distributive law of the \emph{monad} $\F$ over the functor $\B$, then any $\B \F$-coalgebra canonically lifts to a $\lambda$-bialgebra as guaranteed by the following proposition (for a proof see e.g. \cite{JSS12_traceSemanticsDeterminization}).
\begin{proposition}\label{prop:bialgfreemonad}
 Given a monad $\T\: \C \to \C$ and a functor $\B\: \C \to \C$, let $\lambda \: \T\B \To \B\T$ be a distributive law of the monad $\T$ over the functor $\B$. Given a $\B \T$-coalgebra $p \: X\to \B \T X$, define the $\B$-coalgebra $\pg{p}\: \T X \to \B \T X$ as
\[\xymatrix{\T X \ar[r]^-{\T p} & \T \B \T X \ar[r]^{\lambda_{\T X}} & \B \T \T X \ar[r]^{\B(\mu^{\T}_X)} & \B \T X}.\]
 Then $(\T X,\mu^{\T}_X,\pg{p})$ forms a $\lambda$-bialgebra. Moreover, this assignment extends to a functor from $\coalg{\B \T}$ to $\bialg{\lambda}$ mapping:
 \begin{itemize}
   \item a $\B \T$-coalgebra $(X,p)$ to the $\lambda$-bialgebra $(\T X,\mu^{\T}_X,\pg{p})$;
   \item a morphism $f \: X \to Y$ of $\B \T$-coalgebrae to a morphism $\T f \: \T X \to \T Y$ of $\lambda$-bialgebrae.
 \end{itemize}
\end{proposition}

We recall also the following folklore result (see e.g. \cite{Klin11_BialgebrasSOSIntro}) on the relation between final coalgebrae and final bialgebrae.
\begin{proposition}\label{prop:finalcoalgBialg}
 Let $\lambda \: \F\B \To \B\F$ be a distributive law between functors $\F,\B\: \C \to \C$ and suppose that a final $\B$-coalgebra $c \: X \to \B X$ exists. Form the $\B$-coalgebra $h \: \F X \xrightarrow{\F c} \F \B X \xrightarrow{\lambda_X} \B \F X$ and let $h' \: \F X \to X$ be the unique $\B$-coalgebra morphism given by finality of $c \: X \to \B X$. Then $(X,h',c)$ is the final $\lambda$-bialgebra.
\end{proposition}

\subsection{Terms, Atoms and Substitutions}\label{ssec:termsAtomsSubstitutions}

We fix a \emph{signature} $\Sigma$ of function symbols, each equipped with a fixed arity, and a countably infinite set $\m{Var} = \{x_1,x_2,x_3,\dots\}$ of variables. We model substitutions and unification of terms over $\Sigma$ and $\m{Var}$ according to the categorical perspective of \cite{Goguen89whatis,Bruni01aninteractive}. To this aim, let the (opposite) \emph{Lawvere Theory} of $\Sigma$ be a category $\Lw$ where objects are natural numbers, with $n \in \Lw$ intuitively representing variables $x_1,x_2,\dots,x_n$ from $\m{Var}$.
For any two $n,m \in \Lw$, the set $\Lw [n,m]$ consists of all $n$-tuples $<t_1,\dots,t_n>$ of terms where only variables among $x_1,\dots,x_m$ occur. The identity on $n \in \Lw$, denoted by $\m{id}_n$, is given by the tuple $<x_1,\dots,x_n>$. The composition of $<t_1^1,\dots,t_n^1> \: n \to m$ and $<t_1^2,\dots,t_m^2> \: m \to m'$ is the tuple $<t_1,\dots,t_n> \: n \to m'$, where $t_i$ is the term $t_i^1$ in which every variable $x_j$ has been replaced with $t_j^2$, for $1 \leq j \leq m$ and $1 \leq i \leq n$.

We call \emph{substitutions} the arrows of $\Lw$ and use Greek letters $\theta,\ \sigma$ and $\tau$ to denote them.  Given $\theta_1 \: n \to m_1$ and $\theta_2 \: n \to m_2$, a \emph{unifier} of $\theta_1$ and $\theta_2$ is a pair of substitutions $\sigma\: m_1 \to m$ and $\tau \: m_2 \to m$, where $m$ is some object of $\Lw$, such that $\sigma \circ \theta_1 = \tau \circ \theta_2$. The \emph{most general unifier} of $\theta_1$ and $\theta_2$ is a unifier with a universal property, i.e. a pushout of the diagram $m_1 \xleftarrow{\theta_1} n \xrightarrow{\theta_2} m_2$.

An \emph{alphabet} $\mc{A}$ consists of a signature $\Sigma$, a set of variables $\m{Var}$ and a set of predicate symbols $P, P_1, P_2,\dots$ each assigned an arity. Given $P$ of arity $n$ and $\Sigma$-terms $t_1,\dots,t_n$, $P(t_1,\dots,t_n)$ is called an \emph{atom}. We use Latin capital letters $A,B,\dots$ for atoms.
Given a substitution $\theta = <t_1,\dots,t_n> \: n \to m$ and an atom $A$ with variables among $x_1,\dots,x_n$, we adopt the standard notation of logic programming in denoting with $A \theta$ the atom obtained by replacing $x_i$ with $t_i$ in $A$, for $1 \leq i \leq n$. The atom $A \theta$ is called a \emph{substitution instance} of $A$. The notation $\{A_1,\dots,A_m\}\theta$ is a shorthand for $\{A_1\theta,\dots,A_m\theta\}$. Given atoms $A_1$ and $A_2$, we say that $A_1$ \emph{unifies} with $A_2$ (equivalently, they are \emph{unifiable}) if they are of the form $A_1 = P(t_1,\dots,t_n)$, $A_2 = P(t_1',\dots,t_n')$ and a unifier $<\sigma,\tau>$ of $<t_1,\dots,t_n>$ and $<t_1',\dots,t_n'>$ exists. Observe that, by definition of unifier, this amounts to saying that $A_1 \sigma = A_2 \tau$. \emph{Term matching} is a particular case of unification, where $\sigma$ is the identity substitution. In this case we say that $<\sigma,\tau>$ is a \emph{term-matcher} of $A_1$ and $A_2$, meaning that $A_1 = A_2 \tau$.

\subsection{Logic Programming}\label{ssec:lpbackground}

A \emph{logic program} $\mb{P}$ consists of a finite set of \emph{clauses} $C$ written as $H \seq B_1,\dots,B_k$. The components $H$ and $B_1,\dots,B_k$ are atoms, where $H$ is called the \emph{head} of $C$ and $B_1,\dots,B_k$ form the \emph{body} of $C$. One can think of $H \seq B_1,\dots,B_k$ as representing the first-order formula $(B_1 \wedge \dots \wedge B_k) \rightarrow H$. We say that $\mb{P}$ is \emph{ground} if only ground atoms (i.e. without variables) occur in its clauses.

The central algorithm of logic programming is SLD-resolution, checking whether a finite collection of atoms $G$ (called a \emph{goal} and usually modeled as a set or a list) is \emph{refutable} in $\mb{P}$. A run of the algorithm on inputs $G$ and $\mb{P}$ gives rise to an \emph{SLD-derivation}, whose steps of computation can be sketched as follows. At the initial step $0$, a set of atoms $G_0$ (the current goal) is initialized as $G$. At each step $i$, an atom $A_i$ is selected in the current goal $G_i$ and one checks whether $A_i$ is unifiable with the head of some clause of the program. If not, the computation terminates with a failure. Otherwise, one such clause $C_i\ =\ H \seq B_1,\dots,B_k$ is selected: by a classical result, since $A_i$ and $H$ unify, they have also a most general unifier $<\sigma_i,\tau_i>$. The goal $G_{i+1}$ for the step $i+1$ is given as
$$\{B_1,\dots,B_k\}\tau_i\ \cup\ (G_{i}\setminus \{A_i\})\sigma_i.$$
Such a computation is called an \emph{SLD-refutation} if it terminates in a finite number (say $n$) of steps with $G_{n} = \emptyset$. In this case one calls \emph{computed answer} the substitution given by composing the first projections $\sigma_n,\dots,\sigma_0$ of the most general unifiers associated with each step of the computation. The goal $G$ is \emph{refutable} in $\mb{P}$ if an SLD-refutation for $G$ in $\mb{P}$ exists.
A \emph{correct answer} is a substitution $\theta$ for which there exist a computed answer $\tau$ and a substituion $\sigma$ such that $\sigma \circ \tau = \theta$.
%
%
We refer to \cite{Lloyd93} for a more detailed introduction to SLD-resolution.

\fznote{For Katya}\begin{convention} In any derivation of $G$ in $\mb{P}$, the standard convention is that the variables occurring in the clause $C_i$ considered at step $i$ do not appear in goals $G_{i-1},\dots,G_0$. This guarantees that the computed answer is a well-defined substitution and may require a dynamic (i.e. at step $i$) renaming of variables appearing in $C_i$. The associated procedure is called \emph{standardizing the variables apart} and we assume it throughout the paper without explicit mention. It also justifies our definition (Section \ref{ssec:termsAtomsSubstitutions}) of the most general unifier as
pushout
of two substitutions \emph{with
different target}, whereas it is also modeled in the literature as the coequalizer of two substitutions \emph{with the same target}, see e.g. \cite{Goguen89whatis}. The different target corresponds to the two substitutions depending on disjoint sets of variables.\end{convention}

Relevant for our exposition are \emph{$\wedge\vee$-trees} (``and-or trees'') \cite{Gupta94}, which represent executions of SLD-resolution exploiting two forms of parallelism: \emph{and-parallelism}, corresponding to simultaneous refutation-search of multiple atoms in a goal, and \emph{or-parallelism}, exploring multiple attempts to refute the same goal. These are also called \emph{and-or parallel derivation tress} in \cite{KatyaJournal}.

\begin{definition}\label{DEF:and-or_par_tree_ground}
Given a logic program $\mb{P}$ and an atom $A$, the (parallel) \emph{$\wedge\vee$-tree} for $A$ in $\mb{P}$ is the possibly infinite tree $\tree$ satisfying the following properties:
\begin{enumerate}
  \item Each node in $\tree$ is either an $\wedge$-node or an $\vee$-node.
  \item Each $\wedge$-node is labeled with one atom and its children are $\vee$-nodes.
  \item The root of $\tree$ is an $\wedge$-node labeled with $A$.
  \item Each $\vee$-node is labeled with $\bullet$ and its children are $\wedge$-nodes.
  \item For every $\wedge$-node $s$ in $\tree$, let $A'$ be its label. For every clause $H\seq B_1,\dots,B_k$ of $\mb{P}$ and most general unifier $<\sigma,\tau>$ of $A'$ and $H$, $s$ has exactly one child $t$, and viceversa. For each atom $B$ in
      $\{B_1,\dots,B_k\}\tau$, $t$ has exactly one child labeled with $B$, and viceversa.
\end{enumerate}
\end{definition}

\noindent As standard for any tree, we have a notion of \emph{depth}: the root is at depth $0$ and depth $i+1$ is given by the children of nodes at depth $i$. In our graphical representation of trees, we draw arcs between a node and its children which we call \emph{edges}. A \emph{subtree} of a tree $\tree$ is a set of nodes of $\tree$ forming a tree $\tree'$, such that the child relation between nodes in $\tree'$ agrees with the one they have in $\tree$.

An SLD-resolution for the singleton goal $\{A\}$ is represented as a particular kind of subtree of the $\wedge\vee$-tree for $A$, called \emph{derivation subtree}. The intuition is that derivation subtrees encode ``deterministic'' computations, that is, no branching given by or-parallelism is allowed. \emph{Refutation subtrees} are those derivation subtrees yielding an SLD-refutation of $\{A\}$: all paths lead to a leaf with no atoms left to be refuted.

\begin{definition}\label{Def:subtree_Power} Let $\tree$ be the $\wedge\vee$-tree for an atom $A$ in a program $\mb{P}$. A subtree $\tree'$ of $\tree$ is a \emph{derivation subtree} if it satisfies the following conditions:
\begin{enumerate}
  \item the root of $\tree'$ is the root of $\tree$;
  \item if an $\wedge$-node of $\tree$ belongs to $\tree'$, then just one of its children belongs to $\tree'$;
  \item if an $\vee$-node of $\tree$  belongs to $\tree'$, then all its children belong to $\tree'$.
\end{enumerate}
A \emph{refutation subtree} (called success subtree in \cite{KomPowerCSL11}) is a finite derivation subtree with only $\vee$-nodes as leaves.
 \end{definition}

\section{Coalgebraic Logic Programming}\label{SEC:Background}

In this section we recall the framework of coalgebraic logic programming, as introduced in \cite{KomMcCuskerPowerAMAST10,KomPowCALCO11,KomPowerCSL11,KatyaJournal}.

\subsection{The Ground Case}\label{ssec:groundcase}

We begin by considering the coalgebraic semantics of ground logic programs \cite{KomMcCuskerPowerAMAST10}. For the sequel we fix an alphabet $\mc{A}$, a set $\At$ of ground atoms and a ground logic program $\mb{P}$. The behavior of $\mb{P}$ is represented by a coalgebra $p \: \At \to \p_f \p_f (\At)$ on $\Set$, where $\p _f$ is the finite powerset functor and $p$ is defined as follows:
\begin{eqnarray*}\label{Eq:def_coalg_p}
  p \: \ A\ \mapsto\ \{\{B_1,\dots,B_k\}\ \mid\ H \seq B_1,\dots,B_k \text{ is a clause of }\mb{P} \text{ and }A = H\}.
\end{eqnarray*}
The idea is that $p$ maps an atom $A \in \At$ to the set of bodies of clauses of $\mb{P}$ whose head $H$ unifies with $A$, i.e. (in the ground case) $A = H$. Therefore $p(A) \in \p_f \p_f (\At)$ can be seen as representing the $\wedge\vee$-tree of $A$ in $\mb{P}$ up to depth $2$, according to Definition \ref{DEF:and-or_par_tree_ground}: each element $\{B_1,\dots,B_k\}$ of $p(A)$ corresponds to a child of the root, whose children are labeled with $B_1,\dots,B_k$. The full tree is recovered as an element of $\mc{C}(\p _f \p _f)(\At)$, where $\mc{C}(\p _f \p _f)$ is the \emph{cofree comonad} on $\p_f \p_f$, standardly provided by the following construction \cite{ak:fixed-point-set-functor,Worrell99}.

\begin{construction}\label{Constr:cofree_ground} The terminal sequence for the functor $\At \times \p_f \p_f (\cdot) \: \Set \to \Set$ consists of sequences of objects $X_{\alpha}$ and arrows $\zeta_{\alpha}\: X_{\alpha+1} \to X_{\alpha}$, defined by induction on $\alpha$ as follows.
\begin{align*}
   X_{\alpha}\ \df \ \left\{
	\begin{array}{ll}
        \At &\ \ \alpha = 0 \\
		\At \times \p_f \p_f (X_{\beta}) &\ \ \alpha = \beta +1
	\end{array}
\right.\ &\
  \zeta_{\alpha}\ \df \ \left\{
	\begin{array}{ll}
        \pi_1 &\ \ \alpha = 0 \\
		\m{id}_{\At} \times \p_f \p_f (\zeta_{\beta}) &\ \ \alpha = \beta +1
	\end{array}
\right.
 \end{align*}
For $\alpha$ a limit ordinal, $X_{\alpha}$ is given as a limit of the sequence and a function
$\zeta_{\alpha}\: X_{\alpha} \to X_{\beta}$ is given for each $\beta \ls \alpha$ by the limiting property of $X_{\alpha}$.

By \cite{Worrell99} it follows that the sequence given above converges to a limit $X_{\gamma}$ such that $\zeta_{\gamma} \: X_{\gamma+1} \to X_{\gamma}$ is an isomorphism forming the final $\At \times \p_f\p_f(\cdot)$-coalgebra $(X_{\gamma},\zeta_{\gamma}^{-1})$. Since $X_{\gamma+1}$ is defined as $\At \times \p_f \p_f (X_{\gamma})$, there is a projection function $\pi_2 \: X_{\gamma +1} \to \p_f \p_f (X_{\gamma})$ which makes $\pi_2 \circ \zeta^{-1}_{\gamma} \:  X_{\gamma} \to \p_f \p_f (X_{\gamma})$ the \emph{cofree $\p_f \p_f$-coalgebra} on $\At$.
This induces the \emph{cofree comonad} $\mc{C}(\p_f \p_f)\: \Set \to \Set$ on $\p_f \p_f$ as a functor mapping $\At$ to $X_{\gamma}$.
\end{construction}

As the elements of the cofree $\p_f$-coalgebra on a set $X$ are standardly presented as finitely branching trees where nodes have elements of $X$ as labels \cite{Worrell99,Rutten00}, those of the cofree $\p_f\p_f$-coalgebra on $X$ can be seen as finitely branching trees with two sorts of nodes occurring at alternating depth, where only one sort has the $X$-labeling.
We now define a $\mc{C}(\p _f \p _f)$-coalgebra $\bb{\cdot}_p\colon \At \to \mc{C}(\p _f \p _f)(\At)$.

\begin{construction}\label{Constr:coalgComonad_ground}
Given a ground program $\mb{P}$, let $p \: \At \to \p_f \p_f (\At)$ be the coalgebra associated with $\mb{P}$. We define a cone $\{p_{\alpha}\: \At \to X_{\alpha}\}_{\alpha \ls \gamma}$ on the terminal sequence of Construction \ref{Constr:cofree_ground} as follows:
\begin{eqnarray*}
   p_{\alpha} &\df& \left\{
	\begin{array}{ll}
        \m{id}_{\At} &\ \ \alpha = 0 \\
		< \m{id}_{\At} , (\p_f \p_f  (p_{\beta}) \circ p) > &\ \ \alpha = \beta +1.
	\end{array}
\right.
\end{eqnarray*}
For $\alpha$ a limit ordinal, $p_{\alpha}\: \At \to X_{\alpha}$ is provided by the limiting property of $X_{\alpha}$. Then in particular $X_{\gamma} = \mc{C}(\p_f \p_f)(\At)$ yields a function $\bb{\cdot}_p\: \At \to \mc{C}(\p_f \p_f)(\At)$.
\end{construction}
Given an atom $A \in \At$, the tree $\bb{A}_p \in \mc{C}(\p _f \p _f)(\At)$ is built by iteratively applying the map $p$, first to $A$, then to each atom in $p(A)$, and so on. For each natural number $m$, $p_{m}$ maps $A$ to its $\wedge\vee$-tree up to depth $m$. As shown in \cite{KomMcCuskerPowerAMAST10}, the limit $\bb{\cdot}_p$ of all such approximations provides the full $\wedge\vee$-tree of $A$.
\begin{example}\label{ex:ground}
Consider the following ground logic program, based on an alphabet consisting of a signature $\{a^0,b^0,c^0\}$ and predicates $\predp (-,-)$, $\predq (-)$.
 \begin{align*}
  \predp (b,c) &\ \seq\ \predq (a),\predq (b),\predq (c) & \predp (b,b) &\ \seq\ \predq (c) \\
  \predp (b,b) &\ \seq\ \predp (b,a),\predp (b,c) &   \predq (c) &\ \seq
\end{align*}
The corresponding coalgebra $p \: \At \to \p_f \p_f (\At)$ and the $\wedge\vee$-tree $\bb{\predp (b,b)}_p \in \mc{C}(\p_f \p_f)(\At)$ are depicted below on the left and on the right, respectively.
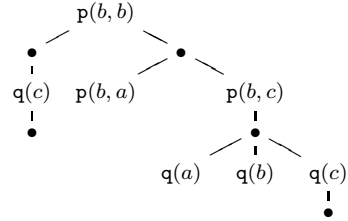
\begin{eqnarray*}
\begin{array}{rcl}
p(\predp (b,c)) & = & \{ \{ \predq (a),\predq (b),\predq (c) \} \} \\
p(\predp (b,b))    & = & \{ \{ \predp (b,a),\predp (b,c)\}\{\predq (c)\} \} \\
p(\predq (c))      & = & \{\{\}\}\\
p(A)      & = & \{ \} \text{ for } A\in \At \setminus  \{\predp (b,c), \predp (b,b), \predq (c) \}
\vspace{-2cm}\end{array}
\qquad \qquad
{\scriptsize
\xymatrix@C=.1cm@R=.1cm{
                 & \ar@{-}[ld] \predp (b,b) \ar@{-}[rd] & & &\\
  \bullet \ar@{-}[d]  &                              & \bullet \ar@{-}[ld] \ar@{-}[rd] & &\\
  \predq (c) \ar@{-}[d] &      \predp (b,a)          &                   & \predp (b,c) \ar@{-}[d] &\\
 \bullet           &                            &                   & \bullet \ar@{-}[ld] \ar@{-}[d] \ar@{-}[rd] & \\
             &                            & \predq (a)        & \predq (b)                   & \predq (c) \ar@{-}[d] \\
             &                            &                   &                              & \bullet}
}
\end{eqnarray*}
\end{example}
\subsection{The General Case} In the sequel we recall the extension of the coalgebraic semantics to arbitrary (i.e. possibly non-ground) logic programs presented in \cite{KomPowCALCO11,KomPowerCSL11}. A motivating observation for their approach is that, in presence of variables, $\wedge\vee$-trees are not guaranteed to represent sound derivations. The problem lies in the interplay between variable dependencies and unification, which makes SLD-derivations for logic programs inherently \emph{sequential} processes \cite{MitchellSeqUnification}.

\begin{example}\label{Ex:unsound_AndOrTree} Consider the signature $\Sigma = \{\tcons^2, \tsucc^1, \tzero^0, \tnil^0\}$ and the predicates $\pList (-)$, $\pNat (-)$. The program $\NatList$, encoding the definition of lists of natural numbers, will be our running example of a non-ground logic program.
 \begin{align*}
  \pList(\tcons (x_1, x_2)) &\ \seq\ \pNat(x_1), \pList(x_2) & \pList(\tnil) &\ \seq  \\
  \pNat(\tsucc(x_1)) &\ \seq\ \pNat(x_1) &   \pNat(\tzero) &\ \seq
\end{align*}
Let $A$ be the atom $\pList ( \tcons (x_1, \tcons (x_2,x_1)))$. It is intuitively clear that there is no substitution of variables making $A$ represent a list of natural numbers: we should replace $x_1$ with a ``number'' (for instance $\tzero$) in its first occurrence and with a ``list'' (for instance $\tnil$) in its second occurrence. Consequently, there is no SLD-refutation for $\{A\}$ in $\NatList$. However, consider the $\wedge\vee$-tree of $A$ in $\NatList$, for which we provide a partial representation as follows.
{\scriptsize
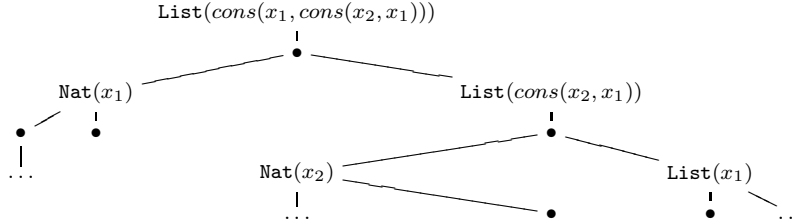
\[\xymatrix@C=.1cm@R=.1cm{
                     &               & \pList ( \tcons (x_1, \tcons (x_2,x_1))) \ar@{-}[d] & & &\\
                     &               &        \bullet \ar@{-}[ld] \ar@{-}[rd]                           & & &\\
                     &\pNat(x_1)  \ar@{-}[ld] \ar@{-}[d]  &                                                   &   \pList ( \tcons (x_2,x_1)) \ar@{-}[d]  & &\\
 \bullet   \ar@{-}[d]           &\bullet    &                                              &        \bullet    \ar@{-}[ld] \ar@{-}[rd]   & \\
 \dots  &    &                  \pNat(x_2) \ar@{-}[rd] \ar@{-}[d]      &                          & \pList(x_1)   \ar@{-}[d] \ar@{-}[rd]  &   \\
          &             &                \dots    &          \bullet   &     \bullet  &  \dots
}
\]
}
The above tree seems to yield an SLD-refutation: $\pList ( \tcons (x_1, \tcons (x_2,x_1)))$ is refuted by proving $\pNat(x_1)$ and
$\pList ( \tcons (x_2,x_1))$.
However, the associated computed answer would be ill-defined, as it is given by substituting $x_2$ with $\tzero$ and $x_1$ both with $\tzero$ and with $\tnil$ (the computed answer of $\pNat(x_1)$ maps $x_1$ to $\tzero$ and the computed answer of $\pList ( \tcons (x_2,x_1))$ maps $x_1$ to $\tnil$).
\end{example}

To obviate to this problem, in \cite{KomPowCALCO11} \emph{coinductive trees} are introduced as a sound variant of $\wedge\vee$-trees, where unification is restricted to term-matching. This constraint is sufficient to guarantee that coinductive trees only represent sound derivations: the key intuition is that a term-matcher is a unifier that leaves untouched the current goal, meaning that the ``previous history'' of the derivation remains uncorrupted.

Before formally defining coinductive trees, it is worth recalling that, in \cite{KomPowCALCO11}, the collection of atoms (based on an alphabet $\mc{A}$) is modeled as a presheaf $\At \: \Lw \to \Set$. The index category is the (opposite) \emph{Lawvere Theory} $\Lw$ of $\Sigma$, as defined above. For each natural number $n \in \Lw$, $\At (n)$ is defined as the set of atoms with variables among $x_1,\dots,x_n$. Given an arrow $\theta \in \Lw[n,m]$, the function $\At(\theta)\: \At(n) \to \At(m)$ is defined by substitution, i.e. $\At(\theta)(A) \df A \theta$. By definition, whenever an atom $A$ belongs to $At(n)$, then it also belongs to $At(n')$, for all $n' \geq n$. However, the occurrences of the same atom in $At(n)$ and $At(n')$ (for $n\neq n'$) are considered distinct: the atoms $A\in At(n)$ and $A\in At(n')$ can be thought of as two states $x_1, \dots, x_n \vdash A$ and $x_1, \dots, x_{n'} \vdash A$ with two different interfaces $x_1, \dots, x_n$ and $x_1, \dots, x_{n'}$.
For this reason, when referring to an atom $A$, it is important to always specify the set $At(n)$ to which it belongs.
\begin{definition}\label{Def:coinductive_trees_Power} Given a logic program $\mb{P}$, a natural number $n$ and an atom $A\in At(n)$, the \emph{$n$-coinductive tree} for $A$ in $\mb{P}$ is the possibly infinite tree $\tree$ satisfying properties 1-4 of Definition \ref{DEF:and-or_par_tree_ground} and property 5 replaced by the following\footnote{Our notion of coinductive tree corresponds to the notion of coinductive forest of breadth $n$ as in \cite[Def.~4.4]{KomPowerCSL11}, the only difference being that we ``glue'' together all trees of the forest into a single tree. It agrees with the definition given in \cite[Def.~4.1]{KatyaJournal} apart for the fact that the parameter $n$ is fixed.}:
 \begin{enumerate}\setcounter{enumi}{4}
  \item For every $\wedge$-node $s$ in $\tree$, let $A' \in \At(n)$ be its label. For every clause $H\seq B_1,\dots,B_k$ of $\mb{P}$ and every term-matcher $<\m{id}_n,\tau>$ of $A'$ and $H$, with $B_1\tau,\dots, B_k\tau \in At(n)$, $s$ has exactly one child $t$, and viceversa. For each atom $B$ in $\{B_1,\dots,B_k\}\tau$, $t$ has exactly one child labeled with $B$, and viceversa.
 \end{enumerate}
\end{definition}
\noindent We recall from \cite{KomPowCALCO11} the categorical formalization of this class of trees. The first step is to generalize the definition of the coalgebra $p$ associated with a program $\mb{P}$. Definition \ref{Def:coinductive_trees_Power} suggests how $p$ should act on an atom $A\in At(n)$, for a fixed $n$:
\begin{align}\label{Eq:KomPow_termMatching}
              A    \mapsto \{ \{B_1,\dots,B_k\}\tau \mid &\ H \seq B_1,\dots,B_k \text{ is a clause of }\mb{P}\text{,} \nonumber \\
                                                           &\  A = H \tau \text{ and } B_1\tau,\dots,B_k\tau \in \At(n)\}.
 \end{align}
For each clause $H \seq B_1,\dots,B_k$, there might be infinitely (but countably) many substitutions $\tau$ such that $A = H \tau$ (see e.g. \cite{KomPowCALCO11}). Thus the object on the right-hand side of \eqref{Eq:KomPow_termMatching} will be associated with the functor $\p_c \p_f \: \Set \to \Set$, where $\p_c$ and $\p_f$ are respectively the countable powerset functor and the finite powerset functor. In order to formalize this as a coalgebra on $\At \: \Lw \to \Set$, let $\lift{\p_c} \: \prsh{\Lw} \to \prsh{\Lw}$ and $\lift{\p_f}\: \prsh{\Lw} \to \prsh{\Lw}$ be extensions of $\p_c$ and $\p_f$ respectively, given according to Definition \ref{def:liftpreshgen}. Then one would like to fix \eqref{Eq:KomPow_termMatching} as the definition of the $n$-component of a natural transformation $p\: \At \to \lift{\p_c}\lift{\p_f} (\At)$. The key problem with this
formulation is that $p$ would \emph{not} be a natural transformation, as shown by the
following example.

\begin{example}\label{Ex:non_compositional} Let $\NatList$ be the same program of Example \ref{Ex:unsound_AndOrTree}. Fix a substitution $\theta = <\tnil> \: 1 \to 0$ and, for each $n \in \Lw$, suppose that $p(n) \: \At(n) \to \lift{\p_c}\lift{\p_f} (\At) (n)$ is defined according to \eqref{Eq:KomPow_termMatching}. Then the following square does not commute.

{
\[\xymatrix{
& \At(1) \ar[d]_{\At(\theta)} \ar[r]^-{p(1)} & \lift{\p_c}\lift{\p_f}(\At)(1) \ar[d]^{\lift{\p_c}\lift{\p_f}(\At)(\theta)} \\
& \At(0) \ar[r]_-{p(0)} & \lift{\p_c}\lift{\p_f}(\At)(0)
}
\]
}
A counterexample is provided by the atom $\pList(x_1) \in \At(1)$. Passing through the bottom-left corner of the square, $\pList(x_1)$ is mapped first to $\pList(\tnil) \in \At(0)$ and then to $\{\emptyset\} \in \lift{\p_c}\lift{\p_f}(\At)(0)$ - intuitively, this yields a refutation of the goal $\{\pList(x_1)\}$ with substitution of $x_1$ with $\tnil$.  Passing through the top-right corner, $\pList(x_1)$ is mapped first to $\emptyset \in \lift{\p_c}\lift{\p_f}(\At)(1)$ and then to $\emptyset \in \lift{\p_c}\lift{\p_f}(\At)(0)$, i.e. the computation ends up in a failure.
\end{example}
In \cite[Sec.4]{KomPowCALCO11} the authors overcome this difficulty by relaxing the naturality requirement. The morphism $p$ is defined as a $\liftLax{\p_c}\liftLax{\p_f}$-coalgebra in the category $\m{Lax}(\Lw,\Poset)$ of locally ordered functors $\F \: \Lw \to \Poset$ and \emph{lax} natural transformations, with each component $p(n)$ given according to \eqref{Eq:KomPow_termMatching} and $\liftLax{\p_c}\liftLax{\p_f}$ the extension of $\lift{\p_c}\lift{\p_f}$ to an endofunctor on $\m{Lax}(\Lw,\Poset)$.

The lax approach fixes the problem, but presents also some drawbacks. Unlike the categories $\Set$ and $\prsh{\Lw}$, $\m{Lax}(\Lw,\Poset)$ is neither complete nor cocomplete, meaning that a cofree comonad on $\liftLax{\p_c}\liftLax{\p_f}$ cannot be retrieved through the standard Constructions \ref{Constr:cofree_ground} and \ref{Constr:coalgComonad_ground} that were used in the ground case. Moreover, the category of $\liftLax{\p_c}\liftLax{\p_f}$-coalgebrae becomes problematic, because coalgebra maps are subject to a commutativity property stricter than the one of lax natural transformations. These two issues force the formalization of non-ground logic program to use quite different (and more sophisticated) categorical tools than the ones employed for the ground case. Finally, as stressed in the Introduction, the laxness of $p$ makes the resulting semantics not compositional.

\section{Saturated Semantics}\label{SEC:SemLogProg}

Motivated by the observations of the previous section, we propose a \emph{saturated approach} to the semantics of logic programs. For this purpose, we consider an adjunction between presheaf categories as depicted on the left.
$$\xymatrix@R=10pt{\\ \prsh{\Lw} \ar@(ur,ul)[rr]^{\U} &\bot & \ar@(dl,dr)[ll]^{\K} \prsh{|\Lw|}} \qquad \qquad \qquad \qquad \xymatrix@R=10pt{|\Lw| \ar@{^{(}->}[r]^{\iota} \ar[dd]_{\F}&  \Lw \ar[ldd]^{\K(\F)}\\ \\
 \set}$$
The left adjoint $\U$ is the forgetful functor, given by precomposition with the inclusion functor
$\incl\: |\Lw| \hookrightarrow \Lw$.
As shown in~\cite[Th.X.1]{mclane}, $\U$ has a right adjoint $\K\: \prsh{|\Lw|}\to \prsh{\Lw}$ sending $\F \: |\Lw| \to \Set$ to its \emph{right Kan extension} along~$\incl$. This is a presheaf $\K(\F) \: \Lw \to \Set$ mapping an object $n$ of $\Lw$ to
\begin{eqnarray*}
  \K(\F)(n) &\df& \prod_{\theta \in \Lw[n,m]} \F (m)
\end{eqnarray*}
where $m$ is any object of $\Lw$.
Intuitively, $\K(\F)(n)$ is a set of tuples indexed by arrows with source $n$ and such that,
at index $\theta\colon n \to m$, there are elements of $\F(m)$. We use $\tuple{x}$ $\tuple{y}, \dots$ to denote such tuples
and we write $\element{\theta}{\tuple{x}}$ to denote the element at index $\theta$ of the tuple $\tuple{x}$.
Alternatively, when it is important to show how the elements depend from the indexes, we use $<x>_{\theta:n\to m}$ (or simply $<x>_{\theta}$) to denote the tuple having at index $\theta$ the element $x$.
With this notation, we can express the behavior of $\K(\F) \: \Lw \to \Set$ on an arrow $\theta \colon n \to m$ as
\begin{equation}\label{eq:deinitionKtheta}
  \K(\F)(\theta) \colon \tuple{x} \mapsto <\element{\sigma\circ\theta}{\tuple{x}}>_{\sigma:m\to m'}.
\end{equation}
The tuple $<\element{\sigma\circ\theta}{\tuple{x}}>_\sigma \in \K(\F)(m)$ can be intuitively read as follows: for each $\sigma \in \Lw[m,m']$, the element indexed by $\sigma$ is the one indexed by $\sigma\circ\theta \in \Lw[n,m']$ in the input tuple $\tuple{x}$.

All this concerns the behavior of $\K$ on the objects of $\prsh{|\Lw|}$.
For an arrow $f\: \F \to \G$ in $\prsh{|\Lw|}$, the natural transformation $\K(f)$ is defined as an indexwise application of $f$ on tuples from $\K(\F)$. For all $n\in \Lw$, $\tuple{x}\in \K(\F)(n)$,
\begin{equation*}\label{eq:deinitionKsuFrecce}
  \K(f)(n) \colon \tuple{x} \mapsto <f(m)(\element{\theta}{\tuple{x}})>_{\theta: n \to m}.
\end{equation*}
For any presheaf $\F \: \Lw \to \Set$, the unit $\eta$ of the adjunction is instantiated to a morphism $\eta_{\F}\: \F \to \K\U(\F)$ given as follows: for all $n\in \Lw$, $X\in \F(n)$,
\begin{equation}\label{eq:deinitionUnit}
  \eta_{\F} (n)\: X  \mapsto < \F(\theta)(X)>_{\theta:n\to m}.
\end{equation}
When taking $\F$ to be $\At$, $\eta_{\At}\: \At \to \K\U(\At)$ maps an atom to its \emph{saturation}: for each $A\in At(n)$,
the tuple $\eta_{\At}(n)(A)$ consists of all substitution instances $\At(\theta)(A) = A \theta$ of $A$, each indexed by the corresponding $\theta \in \Lw[n,m]$.

\smallskip

As shown in Example \ref{Ex:non_compositional}, given a program $\mb{P}$, the family of functions $p$ defined by \eqref{Eq:KomPow_termMatching} fails to be a morphism in $\prsh{\Lw}$. However, it forms a morphism in $\prsh{|\Lw|}$
$$p \colon \U At \to \PP (\U At)$$
where $\liftt{\p _c}$ and $\liftt{\p _f}$ denote the extensions of $\p_c$ and $\p_f$ to $\prsh{|\Lw|}$, given as in Definition \ref{def:liftpreshgen}. The
naturality requirement is trivially satisfied in $\prsh{|\Lw|}$, since $|\Lw|$ is discrete.
The adjunction induces a morphism $p^{\sharp} \colon At \to \K \PP \U (At)$ in $\prsh{\Lw}$, defined as
\begin{eqnarray}
  \At \ \xrightarrow{\eta _{\At}} \ \K\U (\At) \ \xrightarrow{\K(p)} \ \K \PP \U (\At).
\end{eqnarray}
In the sequel, we write $\FS$ for $\K \PP \U$. The idea is to let $\FS$ play the same role as $\p_f\p_f$ in the ground case, with the coalgebra $p^{\sharp} \: At \to \FS (At)$ encoding the program $\mb{P}$. An atom $A \in \At(n)$ is mapped to $<p(m)(A\sigma)>_{\sigma : n \to m}$, that is:
\begin{align}\label{Eq:Sat_p}
  p^{\sharp}(n) \: A  & \mapsto <\{ \{B_1,\dots,B_k\}\tau\ |\ H \seq B_1,\dots,B_k \text{ is a clause of }\mb{P}\text{, } \nonumber \\
                   & \hspace{2cm} A \sigma = H \tau \text{ and } B_1\tau, \dots, B_k\tau \in \At(m)\}>_{\sigma:n\to m}.
\end{align}
Intuitively, $p^{\sharp}(n)$ retrieves all unifiers $<\sigma,\tau>$ of $A$ and heads of $\mb{P}$: first, $A \sigma \in \At(m)$ arises as a component of the saturation of $A$, according to $\eta_{\At}(n)$; then, the substitution $\tau$ is given by term-matching on $A \sigma$, according to $K(p)(m)$.

By naturality of $p^{\sharp}$, we achieve the property of ``commuting with substitutions'' that was precluded by the term-matching approach, as shown by the following rephrasing of Example~\ref{Ex:non_compositional}.
\begin{example}\label{Ex:sat_is_compositional} Consider the same square of Example \ref{Ex:non_compositional}, with $p^{\sharp}$ in place of $p$ and $\FS$ in place of $\lift{\p_c}\lift{\p_f}$. The atom $\pList(x_1) \in \At(1)$ together with the substitution $\theta = <\tnil> \: 1 \to 0$ does not constitute a counterexample to commutativity anymore. Indeed $p^{\sharp}(1)$ maps $\pList(x_1)$ to the tuple $<p(n)(\pList(x_1) \sigma)>_{\sigma \: 1 \to n}$, which is then mapped by $\FS(\At)(\theta)$ to $<p(n)(\pList(x_1)\sigma'\circ \theta)>_{\sigma' \: 0 \to n}$ according to \eqref{eq:deinitionKtheta}. Observe that the latter is just the tuple $<p(n)(\pList(\tnil)\sigma')>_{\sigma' \: 0 \to n}$ obtained by applying first $\At(\theta)$ and then $p^{\sharp}(0)$ to $\pList(x_1)$.
\end{example}
Another benefit of saturated semantics is that $p^{\sharp} \: \At \to \FS(\At)$ lives in a (co)com\-plete category which behaves (pointwise) as $\Set$. This allows us to follow the same steps as in the ground case, constructing a coalgebra for the cofree comonad $\mc{C}(\FS)$ as a straightforward generalization of Constructions \ref{Constr:cofree_ground} and \ref{Constr:coalgComonad_ground}. For this purpose, we first need to verify the following technical lemma.

\begin{proposition}\label{Prop:terminalForKPPU_converges} The functor $\FS$ is accessible and the terminal sequence for $\At \times \FS (\cdot)$ converges to a final $\At \times \FS (\cdot)$-coalgebra. \end{proposition}
\begin{proof} By \cite[Th.7]{Worrell99}, in order to show convergence of the terminal sequence it suffices to prove that $\FS$ is an accessible mono-preserving functor. Since these properties are preserved by composition, we show them separately for each component of $\FS$:
\begin{itemize}
  \item Being adjoint functors between accessible categories, $\K$ and $\U$ are accessible themselves \cite[Prop.2.23]{adamek/rosicky:1994}. Moreover, they are both right adjoints: in particular, $\U$ is right adjoint to the left Kan extension functor along $\iota \: |\catC| \hookrightarrow \catC$. It follows that both preserve limits, whence they preserve monos.
  \item Concerning functors $\liftt{\p_c}\: \prsh{|\Lw|} \to\prsh{|\Lw|}$ and $\liftt{\p_f}\: \prsh{|\Lw|} \to\prsh{|\Lw|}$, it is well-known that $\p_c \: \Set \to \Set$ and $\p_f \: \Set \to \Set$ are both mono-preserving accessible functors on $\Set$. It follows that $\liftt{\p_c}$ and $\liftt{\p_f}$ also have these properties, because (co)limits in presheaf categories are computed objectwise and monos are exactly the objectwise injective morphisms (as shown for instance in \cite[Ch.6]{MaclaneMoerdijk_SheavesGeometryLogic}).
\end{itemize}
\end{proof}

\begin{construction}\label{Constr:cofree_sat} The terminal sequence for the functor $\At \times \FS (\cdot) \: \prsh{\Lw} \to \prsh{\Lw}$ consists of a sequence of objects $X_{\alpha}$ and arrows $\delta_{\alpha}\: X_{\alpha+1} \to X_{\alpha}$, which are defined just as in Construction \ref{Constr:cofree_ground}, with $\FS$ replacing $\p_f \p_f$.
By Proposition \ref{Prop:terminalForKPPU_converges}, this sequence converges to a limit $X_{\gamma}$ such that $X_{\gamma} \cong X_{\gamma+1}$ and $X_{\gamma}$ is the carrier of the cofree $\FS$-coalgebra on $\At$.
\end{construction}

 Since $\FS$ is accessible (Proposition \ref{Prop:terminalForKPPU_converges}), the cofree comonad $\mc{C}(\FS)$ exists and maps $\At$ to $X_{\gamma}$ given as in Construction \ref{Constr:cofree_sat}. Below we provide a $\mc{C}(\FS)$-coalgebra structure $\bb{\cdot}_{p^{\sharp}} \: \At \to \mc{C}(\FS)(\At)$ to $\At$.

\begin{construction}\label{Constr:coalgComonad_sat} The terminal sequence for $\At \times \FS(\cdot)$ induces a cone $\{p_{\alpha}^{\sharp}\: \At \to X_{\alpha}\}_{\alpha \ls \gamma}$ as in Construction \ref{Constr:coalgComonad_ground} with $p^{\sharp}$ and $\FS$ replacing $p$ and $\p_f \p_f$.
This yields a natural transformation $\bb{\cdot}_{p^{\sharp}}\: \At \to X_{\gamma}$, where $X_{\gamma} = \mc{C}(\FS)(\At)$.
\end{construction}
As in the ground case, the coalgebra $\bb{\cdot}_{p^{\sharp}}$ is constructed as an iterative application of $p^{\sharp}$: we call \emph{saturated $\wedge\vee$-tree} the associated tree structure.
\begin{definition}\label{def:saturatedtree}
Given a logic program $\mb{P}$, a natural number $n$ and an atom $A\in At(n)$, the \emph{saturated $\wedge\vee$-tree} for $A$ in $\mb{P}$ is the possibly infinite tree $\tree$ satisfying properties 1-3 of Definition \ref{DEF:and-or_par_tree_ground} and properties 4 and 5 replaced by the following:
 \begin{enumerate}\setcounter{enumi}{3}
  \item Each $\vee$-node is labeled with a substitution $\sigma$ and its children are $\wedge$-nodes.
  \item For every $\wedge$-node $s$ in $\tree$, let $A' \in \At(n')$ be its label. For every clause $H\seq B_1,\dots,B_k$ of $\mb{P}$ and every unifier $<\sigma,\tau>$ of $A'$ and $H$, with $\sigma \: n' \to m'$ and $B_1\tau,\dots, B_k\tau \in At(m')$, $s$ has exactly one child $t$ labeled with $\sigma$, and viceversa. For each atom $B$ in $\{B_1,\dots,B_k\}\tau$, $t$ has exactly one child labeled with $B$, and viceversa.
 \end{enumerate}
\end{definition}

\noindent We have now seen three kinds of tree, exhibiting different substitution mechanisms. In saturated $\wedge\vee$-trees one considers all the unifiers, whereas in $\wedge\vee$-trees and coinductive trees one restricts to most general unifiers and term-matchers respectively. Moreover, in a coinductive tree each $\wedge$-node is labeled with an atom in $\At(n)$ for a fixed $n$, while in a saturated $\wedge\vee$-tree $n$ can dynamically change.

\begin{example}\label{ex:saturatedlist}
Part of the infinite saturated $\wedge\vee$-tree of $\pList ( x_{ 1}) \in \At(1)$ in $\NatList$ is depicted below. Note that not all labels of $\wedge$-nodes belong to $At(1)$, as it would be the case for a coinductive tree: such information is inherited from the label of the parent $\vee$-node, which is now a substitution. For instance, both $\pNat(x_1)$ and $\pList(x_2)$ belong to $At(2)$, since their parent is labeled with $<\tcons(x_{1},x_2)> \: 1 \to 2$ (using the convention that the target of a substitution is the largest index appearing among its variables).
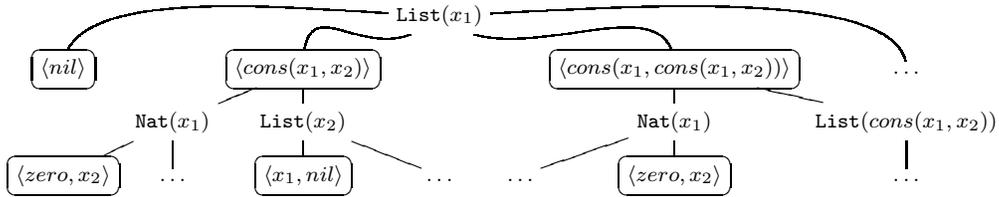
\begin{center}
{\scriptsize
 \[
\xymatrix@C=.1cm@R=.2cm{
 & && \pList ( x_{ 1}) \ar@{-}@(l,u)[llld] \ar@{-}@(dl,u)[ld] \ar@{-}@(dr,u)[drr] \ar@{-}@(r,u)[drrr] \\
 \nodoor{<\tnil>} &     &        \nodoor{<\tcons(x_{1},x_2)>} \ar@{-}[ld] \ar@{-}[d]    &         & & \nodoor{<\tcons(x_1,\tcons(x_1,x_2))>}  \ar@{-}[d] \ar@{-}[rd] & \dots \\
&   \pNat(x_1) \ar@{-}[ld] \ar@{-}[d]    & \pList(x_2)   \ar@{-}[rd] \ar@{-}[d]        &   &   & \pNat(x_1) \ar@{-}[ld] \ar@{-}[d]    & \pList(\tcons(x_1,x_2))   \ar@{-}[d] &\\
 \nodoor{<\tzero,x_2>}             &  \dots   &    \nodoor{<x_1,\tnil>}   & \dots    & \dots  &       \nodoor{<\tzero,x_2>}      &    \dots  \\
}
\]
}
\end{center}\smallskip
\end{example}

\noindent Next we verify that the notion of saturated $\wedge\vee$-tree is indeed the one given by saturated semantics. In the following proposition and in the rest of the paper, with an abuse of notation we use $\bb{A}_{p^{\sharp}}$ to denote the application of $\bb{\cdot}_{p^{\sharp}}(n)$ to $A\in At(n)$ without mentioning the object $n\in \Lw$.

\begin{proposition}[Adequacy] \label{prop:adequacy_sat} For all $n$ and $A \in \At(n)$, the saturated $\wedge\vee$-tree of $A$ in a program $\mb{P}$ is $\bb{A}_{p^{\sharp}}$.
\end{proposition}
\begin{proof} Fix $n \in \Lw$. Just as Construction \ref{Constr:cofree_ground} allows to describe the cofree $\p_f\p_f$-coalgebra $\cof{\p_f\p_f}(\At)$ on $\At$ as the space of trees with two sorts of nodes occurring at alternating depth and one sort labelled by $\At$, observing Construction \ref{Constr:cofree_sat} we can provide a similar description for the elements of $\cof{\FS}(\At)(n)$. Those are also trees with two sorts of nodes occurring at alternating depth. One sort (the one of $\wedge$-nodes), is labeled with elements of $\At(m)$ for some $m \in \Lw$. The root itself is an $\wedge$-node labeled with some atom in $\At(n)$. The $\vee$-nodes children of an $\wedge$-node $B \in \At(m)$ are not given by a plain set, as in the ground case, but instead by a tuple $\tuple{x} \in \FS(\At)(m)$. We can represent $\tuple{x}$ by drawing a child $\vee$-node $t$ labeled with the substitution $\theta \: m \to m'$ for each element $S_t$ of the set $\tuple{x}(\theta) \in \PP \At(m)$. The $\wedge$-node children of $t$ are labeled with the elements of $S_t \in \liftt{\p _f} \At(m')$.


The saturated $\wedge\vee$-tree for an atom $A \in \At(n)$, as in Definition \ref{def:saturatedtree}, is a tree of the above kind and thus an element of $\cof{\FS}(\At)(n)$. The function $A \mapsto T_A$ mapping $A$ in its saturated $\wedge\vee$-tree $T_A$ extends to an $\At \times \FS(\cdot)$-coalgebra morphism from $\At$ --- with coalgebraic structure given by $<\id,p^{\sharp}>$ --- to $\cof{\FS}(\At)$. By construction
$\bb{\cdot}_{p^{\sharp}}$ is also an $\At \times \FS(\cdot)$-coalgebra morphism. Since $\cof{\FS}(\At)$ is the final $\At \times \FS(\cdot)$-coalgebra, these two morphisms must coincide, meaning that $\bb{\cdot}_{p^{\sharp}}$ maps $A$ into its saturated $\wedge\vee$-tree $T_A$.
\end{proof}

For an arrow $\theta \in \Lw[n,m]$, we write $\overline{\theta}$ for $\mc{C}(\FS)(\At)(\theta) \colon \mc{C}(\FS)(\At)(n) \to \mc{C}(\FS)(\At)(m)$.
With this notation, we can state the compositionality result motivating our approach. Its proof is an immediate consequence of the naturality of $\bb{\cdot}_{p^{\sharp}}$.

\begin{theorem}[Compositionality]\label{thm:compo} For all atoms $A \in \At(n)$ and substitutions $\theta \in \Lw[n,m]$,
\begin{equation*}
    \bb{A\theta}_{p^{\sharp}} = \bb{A}_{p^{\sharp}}\overline{\theta}.
\end{equation*}
\end{theorem}

We conclude this section with a concrete description of the behavior of the operator $\overline{\theta}$, for a given substitution $\theta \in \Lw[n,m]$. Let $r$ be the root of a tree $T \in \mc{C}(\FS)(\At)(n)$ and $r'$ the root of $T\overline{\theta}$. Then
\begin{enumerate}
 \item the node $r$ has label $A$ iff $r'$ has label $A\theta$;
 \item the node $r$ has a child $t$ with label $\sigma\circ \theta$ and children $t_1, \dots, t_n$ iff $r'$ has a child $t'$ with label $\sigma$ and children $t_1 \dots t_n$.
\end{enumerate}
Note that the children $t_1,\dots,t_n$ are exactly the same in both trees: $\overline{\theta}$ only modifies the root and the $\vee$-nodes at depth 1 of $\tree$, while it leaves untouched all the others. This peculiar behavior can be better understood by
observing that the definition of $\K(\F)(\theta)$, as in \eqref{eq:deinitionKtheta}, is independent of the presheaf $\F$. As a result, $\overline{\theta}=X_{\gamma}(\theta)$ is independent of all the $X_{\alpha}$s built in Construction \ref{Constr:cofree_sat}.
%
%
%
%
%
%
%
\begin{example}\label{ex:thetatree} Recall from Example \ref{ex:saturatedlist} the saturated $\wedge\vee$-tree $\bb{\pList(x_1)}_{p^{\sharp}}$. For $\theta = <\tcons(x_{1},x_2)>$,
the tree $\bb{\pList(x_1)}_{p^{\sharp}}\overline{\theta}$ is depicted below. 
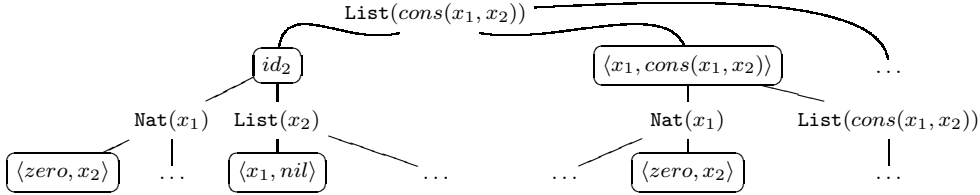
\[
{\scriptsize
\xymatrix@C=.1cm@R=.2cm{
 & && \pList ( \tcons(x_{1},x_2) )  \ar@{-}@(dl,u)[ld] \ar@{-}@(dr,u)[drr] \ar@{-}@(r,u)[drrr] & & &\\
  &     &        \nodoor{id_2} \ar@{-}[ld] \ar@{-}[d]    &         & & \nodoor{<x_1,\tcons(x_1,x_2)>}  \ar@{-}[d] \ar@{-}[rd] & \dots \\
&   \pNat(x_1) \ar@{-}[ld] \ar@{-}[d]    & \pList(x_2)   \ar@{-}[rd] \ar@{-}[d]        &   &   & \pNat(x_1) \ar@{-}[ld] \ar@{-}[d]    & \pList(\tcons(x_1,x_2))   \ar@{-}[d] &\\
 \nodoor{<\tzero,x_2>}             &  \dots   &    \nodoor{<x_1,\tnil>}   & \dots   & \dots    &       \nodoor{<\tzero,x_2>}              &    \dots  \\
&&&&&&&&&&&&&&&&&&&&&&&&&&&&&&&&&&&&&7}
}
\]
\end{example}

\section{Desaturation}\label{SEC:Desaturation}
One of the main features of coinductive trees is to represent (sound) and-or parallel derivations of goals. This leads the authors of \cite{KomPowerCSL11} to a resolution algorithm exploiting the two forms of parallelism. Motivated by these developments, we include coinductive trees in our framework, showing how they can be obtained as a ``desaturation'' of saturated $\wedge\vee$-trees.

For this purpose, the key ingredient is given by the \emph{counit} $\epsilon$ of the adjunction $\U \dashv \K$. Given a presheaf $\F \: |\Lw| \to \Set$, the morphism $\epsilon_{\F} \: \U\K(\F) \to \F$ is defined as follows: for all $n\in \Lw$ and $\tuple{x}\in \U\K(\F)(n)$,
\begin{align}\label{EQ:counit}
  \epsilon_{\F} (n)\: 
 \dot{x}  \mapsto \dot{x}(\m{id}_{n})
\end{align}
where $\dot{x}(\m{id}_{n})$ is the element of the input tuple $\dot{x}$ which is indexed by the identity substitution $\m{id}_{n} \in \Lw[n,n]$. In the logic programming perspective, the intuition is that, while the unit of the adjunction provides the saturation of an atom, the counit reverses the process. It takes the saturation of an atom and gives back the substitution instance given by the identity, that is, the atom itself.

We now want to define a ``desaturation'' map $\desat$ from saturated $\wedge\vee$-trees to coinductive trees, acting as a pointwise application of $\epsilon_{\U\At}$. For this purpose, first we state the construction of the cofree comonad on $\PP$ for later reference.

\begin{construction}\label{Constr:cofree_PP} The terminal sequence for $\U\At \times \PP (\cdot) \: \prsh{|\Lw|} \to \prsh{|\Lw|}$ consists of sequences of objects $Y_{\alpha}$ and arrows $\xi_{\alpha}\: Y_{\alpha+1} \to Y_{\alpha}$, defined by induction on $\alpha$ as follows.
\begin{align*}
   Y_{\alpha}\ \df\ \left\{
	\begin{array}{ll}
        \U\At &\ \ \alpha = 0 \\
		\U\At \times \PP (Y_{\beta}) &\ \ \alpha = \beta +1
	\end{array}
\right.\ &&\
  \xi_{\alpha}\ \df\ \left\{
	\begin{array}{ll}
        \pi_1 &\ \ \alpha = 0 \\
		\m{id}_{\U\At} \times \PP (\xi_{\beta}) &\ \ \alpha = \beta +1
	\end{array}
\right.
\end{align*}
For $\alpha$ a limit ordinal, $Y_{\alpha}$ and $\xi_{\alpha}$ are defined as expected. As stated in the proof of  Proposition \ref{Prop:terminalForKPPU_converges}, $\PP$ is a mono-preserving accessible functors. Then by \cite[Th.7]{Worrell99} we know that the sequence given above converges to a limit $Y_{\chi}$ such that $Y_{\chi} \cong Y_{\chi+1}$ and $Y_{\chi}$ is the value of $\mc{C}(\PP)\: \prsh{|\Lw|} \to \prsh{|\Lw|}$ on $\U\At$, where $\mc{C}(\PP)$ is the cofree comonad on $\PP$ induced by the terminal sequence given above, analogously to Construction \ref{Constr:cofree_ground}.
\end{construction}

The next construction defines the desidered morphism $\desat \: \U \big( \mc{C}(\FS)(\At)\big) \to \mc{C}(\PP)(\U\At)$.
\begin{construction} \label{Constr:Desat} Consider the image of the terminal sequence converging to $\mc{C}(\FS)(\At) = X_{\gamma}$ (Construction \ref{Constr:cofree_sat}) under the forgetful functor
$\U\: \prsh{\Lw} \to \prsh{|\Lw|}$. We define a sequence of natural transformations $\{d_{\alpha}\: \U(X_{\alpha}) \to Y_{\alpha} \}_{\alpha\ls\gamma}$ as follows\footnote{Concerning the successor case, observe that
$\m{id}_{\U \At} \times \big(\PP (d_{\beta}) \circ \epsilon_{\PP \U (X_{\beta})}\big)$
is in fact an arrow from $\U\At \times \U\K\PP\U (X_{\beta})$ to $Y_{\beta+1}$. However, the former is isomorphic to $\U (X_{\beta+1}) = \U \big(\At \times \K\PP\U (X_{\beta})\big)$, because $\U$ is a right adjoint (as observed in Proposition \ref{Prop:terminalForKPPU_converges}) and thence it commutes with products.}:
 \begin{eqnarray*}   \label{eq:desat_cone_sigma}
   d_{\alpha} &\df& \left\{
	\begin{array}{ll}
        \m{id}_{\U \At} & \alpha = 0 \\
		\m{id}_{\U \At} \times \big(\PP (d_{\beta}) \circ \epsilon_{\PP \U (X_{\beta})}\big) & \alpha = \beta +1.
	\end{array}
\right.
 \end{eqnarray*}

\[\xymatrix{
\U (X_{\beta}) \ar@/^/[rrrr]^{d_{\beta}} & & & & Y_{\beta} \\
\U (X_{\beta +1}) \ar@/^/[rrrr]^{d_{\beta +1}} \ar[u]^{\U(\delta_{\beta})} \ar[rrd]|{\m{id}_{\U\At} \times \epsilon_{\PP \U ( X_{\beta})}} & & & & Y_{\beta + 1} \ar[u]_{\xi_{\beta}} \\
                                   & & \U\At \times  \PP \U (X_{\beta}) \ar[rru]|{\m{id}_{\U\At} \times \PP (d_{\beta})} & &
}
\]

For $\alpha  \ls \gamma$ a limit ordinal, an arrow $d_{\alpha}\: \U(X_{\alpha}) \to Y_{\alpha}$ is provided by the limiting property of $Y_{\alpha}$. In order to show that the limit case is well defined, observe that, for every $\beta \ls \alpha$, the above square commutes, that is, $\xi_{\beta} \circ d_{\beta+1} = d_{\beta} \circ \U(\delta_{\beta})$. This can be easily checked by ordinal induction, using the fact that $\epsilon_{\PP \U (X_{\beta})}$ is a natural transformation for each $\beta \ls \alpha$.

We now turn to defining a natural transformation $\desat \: \U \big( \mc{C}(\FS)(\At)\big) \to \mc{C}(\PP)(\U\At)$.
If $\chi \leq \gamma$, then this is provided by $d_{\chi} \: \U(X_{\chi}) \to Y_{\chi}$ together with the limiting property of $\U(X_{\gamma})$ on $\U(X_{\chi})$. In case $\gamma \ls \chi$, observe that, since $X_{\gamma}$ is isomorphic to $X_{\gamma +1}$, then $X_{\gamma}$ is isomorphic to $X_{\zeta}$ for all $\zeta \gr \gamma$, and in particular $X_{\gamma} \cong X_{\chi}$. Then we can suitably extend the sequence to have a natural transformation $d_{\chi} \: \U(X_{\chi}) \to Y_{\chi}$. The morphism $\desat$ is given as the composition of $d_{\chi}$ with the isomorphism between $\U(X_{\gamma})$ and $\U(X_{\chi})$.
\end{construction}

The next theorem states that $\desat$ is a translation from saturated to coinductive trees: given an atom $A \in \At(n)$, it maps $\bb{A}_{p^{\sharp}}$ to the $n$-coinductive tree of $A$. The key intuition is that $n$-coinductive trees can be seen as saturated $\wedge\vee$-trees where the labeling of $\vee$-nodes has been restricted to the identity
substitution $\m{id}_n$, represented as $\bullet$ (see Definition \ref{Def:coinductive_trees_Power}). The operation of pruning all $\vee$-nodes (and their descendants) in $\bb{A}_{p^{\sharp}}$ which are not labeled with $\m{id}_n$ is precisely what is provided by Construction \ref{Constr:Desat}, in virtue of the definition of the counit $\epsilon$ given in~\eqref{EQ:counit}.
\begin{theorem}[Desaturation]\label{Th:Desaturation}
Let $\bb{\cdot}_{p^{\sharp}} \: \At \to \mc{C}(\FS)(\At)$ be defined for a logic program $\mb{P}$ according to Construction \ref{Constr:coalgComonad_sat} and $\desat  \: \U\big(\mc{C}(\FS)(\At)\big) \to \mc{C}(\PP)(\U\At)$ be defined according to Construction \ref{Constr:Desat}. Then for all $n\in|\Lw|$ and $A\in \U\At(n)$, the $n$-coinductive tree of $A$ in $\mb{P}$ is
$\big(\desat \circ \U (\bb{\cdot}_{p^{\sharp}})\big)(n)(A)$.
\end{theorem}
Theorem \ref{Th:Desaturation} provides an alternative formalization for the coinductive tree semantics \cite{KomPowCALCO11}, given by composition of the saturated semantics with desaturation. In fact it represents a different approach to the non-compositionality problem: instead of relaxing naturality to lax naturality, we simply forget about all the arrows of the index category $\Lw$, shifting the framework from $\prsh{\Lw}$ to $\prsh{|\Lw|}$. The substitutions on trees (that are essential, for instance, for the resolution algorithm given in \cite{KomPowerCSL11}) exist at the saturated level, i.e.\ in $\mc{C}(\FS)(\At)$, and they are given precisely as the operator $\overline{\theta}$ described at the end of Section \ref{SEC:SemLogProg}.
%

\noindent\begin{minipage}[]{.69\linewidth}
\noindent \begin{example}
The coinductive tree for $\pList ( \tcons(x_{1},x_2) )$ in $\NatList$ is depicted on the right. It is constructed by desaturating the tree $\bb{\pList(\tcons(x_{1},x_2))}_{p^{\sharp}}$ in Example
\ref{ex:thetatree}, i.e., by pruning all the $\vee$-nodes (and their descendants) that are not labeled with $id_2$.
\end{example}
\end{minipage}%
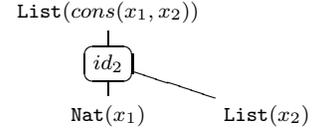
\quad
\begin{minipage}{.3\linewidth}%
{\scriptsize
 \[
\xymatrix@C=.1cm@R=.2cm{
   \pList ( \tcons(x_{1},x_2) )  \ar@{-}[d] & & &\\
               \nodoor{id_2} \ar@{-}[rd] \ar@{-}[d]    &   \\
   \pNat(x_1)     & \pList(x_2)          \\
}
\]
}
\end{minipage}

\section{Soundness and Completeness}\label{SEC:Completeness}


The notion of coinductive tree leads to a semantics that is sound and complete with respect to SLD-resolution \cite[Th.4.8]{KomPowerCSL11}. To this aim, a key role is played by \emph{derivation subtrees} of coinductive trees: they are defined exactly as derivation subtrees of $\wedge\vee$-trees (Definition \ref{Def:subtree_Power}) and represent SLD-derivations where the computation only advances by term-matching.

Similarly, we now want to define a notion of subtree for saturated semantics. This requires care: saturated $\wedge\vee$-trees are associated with unification, which is more liberal than term-matching. In particular, like $\wedge\vee$-trees, they may represent unsound derivation strategies (\emph{cf.} Example \ref{Ex:unsound_AndOrTree}). 
However, in saturated $\wedge\vee$-trees \emph{all} unifiers, not just the most general ones, are taken into account. This gives enough flexibility to shape a sound notion of subtree, based on an implicit synchronization of the substitutions used in different branches.

\begin{definition}\label{Def:subtree_sat} Let $\tree$ be the saturated $\wedge\vee$-tree for an atom $A$ in a program $\mb{P}$. A subtree $\tree'$ of $\tree$ is called a \emph{synched derivation subtree} if it satisfies properties 1-3 of Definition \ref{Def:subtree_Power} and the following condition:
\begin{enumerate}\setcounter{enumi}{3}
  \item all $\vee$-nodes of $\tree'$ at the same depth are labeled with the same substitution.
\end{enumerate}
A \emph{synched refutation subtree} is a finite synched derivation subtree with only $\vee$-nodes as leaves.
Its \emph{answer} is the substitution $\theta_{2k+1} \circ \dots \theta_{3} \circ \theta_1$, where $\theta_i$ is the (unique) substitution labeling the $\vee$-nodes of depth $i$ and $2k+1$ is its maximal depth.
\end{definition}


%
%

The prefix ``synched'' emphasizes the restriction to and-parallelism encoded in Definition \ref{Def:subtree_sat}. Intuitively, we force all subgoals at the same depth to proceed with \emph{the same} substitution. For instance, this rules out the unsound derivation of Example \ref{Ex:unsound_AndOrTree}.
%
For a comparison with Definition \ref{Def:subtree_Power}, note that derivation subtrees can be seen as special instances of synched derivation subtrees where all the substitutions are forced to be identities.

We are now in position to compare the different semantics for logic programming.

\begin{theorem}[Soundness and Completeness]\label{Th:Completeness}  Let $\mb{P}$ be a logic program and $A \in \At(n)$ an atom. The following are equivalent.
\begin{enumerate}
  \item The saturated $\wedge\vee$-tree for $A$ in $\mb{P}$ has a synched refutation subtree with answer~ $\theta$. \label{pt:completenessSat}
  \item For some $m \in \Lw$, the $m$-coinductive tree for $A \theta$ in $\mb{P}$ has a refutation subtree.
  \item There is an SLD-refutation for $\{A\}$ in $\mb{P}$ with correct answer $\theta$.
\end{enumerate}
\end{theorem}

In statement (3), note that $\theta$ is the \emph{correct} (and not the computed) answer: indeed, the unifiers associated with each derivation step in the synched refutation subtree are not necessarily the most general ones. Towards a proof of Theorem \ref{Th:Completeness}, a key role is played by the following proposition. 

\newcommand{\propsubtree}{
Let $\mb{P}$ be a logic program and $A \in \At(n)$ an atom.
If $\bb{A}_{p^{\sharp}}$ has a synched refutation subtree with answer $\theta \colon n\to m$, then $\bb{A}_{p^{\sharp}}\overline{\theta}$ has a synched refutation subtree whose $\vee$-nodes are all labeled with $id_m$.
 }
\begin{proposition}\label{prop:subtree}
\propsubtree
\end{proposition}
\begin{proof} See Appendix \ref{Sec:appendix}.\end{proof}

\begin{figure}[t]
\begin{tabular}{cc}
{\scriptsize
\xymatrix@C=.1cm@R=.2cm{
 & \pList ( \tconsZ (x_{1}, (\tconsZ(x_{1},x_2)))) \ar@{-}[d] & & \\
 &    \nodoor{<x_{1},x_2>} \ar@{-}[d] \ar@{-}[ld]  &  &     \\
  \pNat(x_1) \ar@{-}[d]   & \pList(\tconsZ(x_{1},x_2)) \ar@{-}[d] &  \\
\nodoor{<\tsuccZ(x_1),\tnilZ>}  \ar@{-}[d]     &       \nodoor{<\tsuccZ(x_1),\tnilZ>} \ar@{-}[d] \ar@{-}[rd] &  \\
 \pNat(x_1) \ar@{-}[d]  &       \pNat(\tsuccZ(x_1))   \ar@{-}[d]       &       \pList(\tnilZ) \ar@{-}[d] \\
 \nodoor{<\tzeroZ>}      &       \nodoor{<\tzeroZ>}  \ar@{-}[d]        &         \nodoor{<\tzeroZ>}\\
                           &           \pNat(\tzeroZ)  \ar@{-}[d]              &      \\
                           &                   \nodoor{\m{id}_0}               &
}
}
&
{\scriptsize
\xymatrix@C=.1cm@R=.2cm{
 & \pList ( \tconsZ (\tsuccZ(\tzeroZ), (\tconsZ(\tsuccZ(\tzeroZ),\tnilZ)))) \ar@{-}[d] & \\
 &        \nodoor{\m{id}_0} \ar@{-}[ld] \ar@{-}[d]         &\\
  \pNat(\tsuccZ(\tzeroZ)) \ar@{-}[d]    &  \pList(\tconsZ(\tsuccZ(\tzeroZ),\tnilZ)) \ar@{-}[d] \\
\nodoor{\m{id}_0}  \ar@{-}[d]  &                \nodoor{\m{id}_0} \ar@{-}[rd] \ar@{-}[d]   \\
 \pNat(\tzeroZ) \ar@{-}[d]  &       \pNat(\tsuccZ(\tzeroZ))   \ar@{-}[d]              &       \pList(\tnilZ) \ar@{-}[d]  \\
 \nodoor{\m{id}_0}      &       \nodoor{\m{id}_0}  \ar@{-}[d]        &                 \nodoor{\m{id}_0}   \\
                           &           \pNat(\tzeroZ)  \ar@{-}[d]             &                     &          \\
                           &                   \nodoor{\m{id}_0}             &                       &
}
}
\end{tabular}
\caption{Successful synched derivation subtrees for $\pList ( \tcons (x_{1}, (\tcons(x_{1},x_2))))$ (left) and $\pList ( \tcons (\tsucc(\tzero), (\tcons(\tsucc(\tzero),\tnil))))$ (right) in $\NatList$. The symbols $\tcons$, $\tnil$, $\tsucc$ and $\tzero$ are abbreviated to $\tconsZ$, $\tnilZ$, $\tsuccZ$ and $\tzeroZ$ respectively.}\label{fig:subtree}
\end{figure}

Figure \ref{fig:subtree} provides an example of the construction needed for Proposition \ref{prop:subtree}. Note that the root of the rightmost tree is labeled with an atom of the form $A \theta$, where $\theta$ and $A$ are respectively the answer and the label of the root of the leftmost tree. The rightmost tree is a refutation subtree of the $0$-coinductive tree for $A \theta$ and can be obtained from the leftmost tree via a procedure involving the operator $\overline{\theta}$ discussed at the end of Section \ref{SEC:SemLogProg}.

\medskip

We are now ready to provide a proof of Theorem \ref{Th:Completeness} combining Proposition \ref{prop:subtree}, the desaturation procedure of Theorem \ref{Th:Desaturation} and the compositionality result of Theorem \ref{thm:compo}.

\begin{proof}[Proof of Theorem \ref{Th:Completeness}] The statement $(2 \Leftrightarrow 3)$ is a rephrasing of \cite[Th.4.8]{KomPowerCSL11}, whence we focus on proving $(1 \Leftrightarrow 2)$, where $m$ is always the target object of $\theta$.

$(1 \Rightarrow 2)$.
If $\bb{A}_{p^{\sharp}}$ has a synched refutation subtree with answer $\theta$, then by Proposition \ref{prop:subtree},  $\bb{A}_{p^{\sharp}}\overline{\theta}$ has a synched refutation subtree $\tree$ whose $\vee$-nodes are all labeled with $id_m$. Compositionality (Theorem \ref{thm:compo}) guarantees that $\tree$ is a synched refutation subtree of $\bb{A\theta}_{p^{\sharp}}$.
Since all the $\vee$-nodes of $\tree$ are labeled with $id_m$, $\tree$ is preserved by desaturation. This means that $\tree$ is a refutation subtree of $\desat(\U(\bb{\cdot}_{p^{\sharp}}))(m)(A\theta)$ which, by Theorem \ref{Th:Desaturation}, is the $m$-coinductive tree for $A\theta$ in $\mb{P}$.

$(2 \Leftarrow 1)$.
If the $m$ coinductive tree for $A\theta$ has a refutation subtree $\tree$ then, by Theorem \ref{Th:Desaturation}, this is also the case for  $\desat(\U(\bb{\cdot}_{p^{\sharp}}))(m)(A\theta)$. This means that $\tree$ is a synched derivation subtree of $\bb{A\theta}_{p^{\sharp}}$ whose $\vee$-nodes are all labeled by $id_m$.
By compositionality, $\tree$ is also a subtree of $\bb{A}_{p^{\sharp}}\overline{\theta}$. Let $t$ be the $\vee$-node at the first depth of $\tree$. By construction of the operator $\overline{\theta}$, the root of $\bb{A}_{p^{\sharp}}$ has a child $t'$ labeled with $\theta$ having the same children as $t$. Therefore $\bb{A}_{p^{\sharp}}$ has a synched refutation subtree with answer $\theta$.
\end{proof}

\section{Bialgebraic Semantics of Goals: the Ground Case}\label{sec:parground}



In this section we lay the foundations of a (saturated) semantics that applies directly to goals. As outlined in the introduction, a main motivation for this extension is the study of yet another form of compositionality, now with respect to the algebraic structure of goals.

First, we approach the question in the simpler case of ground logic programs. Such a program $\mb{P}$, that in Section~\ref{ssec:groundcase} has been represented as a $\p_f\p_f$-coalgebra $p \: \At \to \p_f\p_f(\At)$, will now be modeled as a certain bialgebra. The coalgebraic part will be given by the endofunctor $\p_f$: this corresponds to the outer occurrence of $\p_f$ in $p \: \At \to \p_f\p_f(\At)$ and encodes non-determinism, i.e., the possibility that an atom matches multiple heads in a program. The algebraic part, encoding the internal structure of a goal, instead will be given by a monad $\T$ that corresponds to the inner occurrence of $\p_f$ in $p \: \At \to \p_f\p_f(\At)$. Intuitively, $\T(\At)$ represents a goal, i.e. a collection of atoms, to be processed \emph{in parallel}.

To formally define the type of our bialgebrae, we need to provide a distributive law $\lambda \: \T \p_f \To \p_f \T$. Our specification for $\lambda$ stems from the observation that $\p_f(\At)$ models a \emph{disjunction} of atoms, whereas $\T(\At)$ models a \emph{conjunction}. Then $\lambda$ should just distribute conjunction over disjunction, for instance:
\begin{eqnarray} \label{eq:speclambdaground}
(A_1 \vee A_2) \wedge (B_1 \vee B_2) \mapsto (A_1 \wedge B_1) \vee (A_1 \wedge B_2) \vee (A_2 \wedge B_1) \vee (A_2 \wedge B_2).
\end{eqnarray}
Provided this specification for $\lambda$, one could naively think of modelling $\T$ as $\p_f$ itself, so that the conjunction $\T(\At)$ is represented as a finite set of atoms. However, this would pose a problem with the naturality of $\lambda$. Consider again the example above, now formalized with $\T = \p_f$, and suppose to check naturality for $f\: \At \to \At$ defined as $(A_1 \mapsto B_1, A_2 \mapsto B_2)$:
\begin{eqnarray} \label{eq:counterexTground}
\vcenter{\xymatrix{
\T \p_f(\At) \ar[d]_{\T \p_f(f)} \ar[r] ^{\lambda_{\At}}& \p_f \T(\At) \ar[d]^{\p_f \T(f)}\\
\T \p_f(\At) \ar[r]^{\lambda_{\At}} & \p_f \T(\At)
}}
\end{eqnarray}
The function $\lambda_{\At} \circ \T \p_f(f)$ maps $\{\{A_1 , A_2\},\{B_1 , B_2\}\}$ first to $\{\{B_1 , B_2\}\}$ and then to $\{\{B_1\},\{B_2\}\}$. Instead $\p_f \T(f) \circ \lambda_{\At}$ maps $\{\{A_1 , A_2\},\{B_1 , B_2\}\}$ first to $\{\{A_1 , B_1\} , \{A_1 , B_2\},$ $\{A_2 , B_1\} , \{A_2 , B_2\}\}$ and then to $\{\{B_1\} , \{B_1 , B_2\} , \{B_2\}\}$. Therefore our specification \eqref{eq:speclambdaground} for $\lambda$, with $\T$ modeled as $\p_f$, would not yield a natural transformation.

For this reason, we model instead the elements of $\T(\At)$ as \emph{lists} of atoms. Formally, we let $\T$ be the \emph{list} endofunctor $\li \: \set \to \set$ mapping a set $X$ into the set of lists of elements of $X$ and a function $f \: X \to Y$ into its componentwise application to lists in $\li(X)$. Such a functor forms a monad with unit $\eta^{\li}_X \: X \To \li X$ given by $x \mapsto [x]$ (where $[x]$ is the list with a single element $x$) and a multiplication given by flattening: for $l_1,\dots,l_n$ elements of $\li(X)$, $\mu^{\li}_X \: \li \li X \To \li X$ maps $[l_1,\dots,l_n]$ to $l_1 \lconc \dots \lconc l_n$, where $\lconc$ denotes list concatenation. In virtue of its definition, we will make use of the notation $\lconc$ for the function $\mu^{\li}_X$, whenever $X$ is clear from the context.

There is a distributive law $\lambda \: \li \p_f \To \p_f \li$ of the monad $\p_f$ over the monad $\li$ given by assignments $[X_1,\dots,X_n] \mapsto \{[x_1,\dots,x_n] \mid x_i \in X_i\}$ (\emph{cf.} \cite[Ex. 2.4.8]{manes2007monad}). One can readily check that this definition implements the specification given in \eqref{eq:speclambdaground}:
\[ [\{A_1 , A_2\},\{B_1 , B_2\}] \mapsto \{ [A_1 , B_1] , [A_1 , B_2] , [A_2 , B_1] , [A_2 , B_2]\}. \]
Also, the counterexample to commutativity of \eqref{eq:counterexTground} given above is neutralized as follows.
\begin{eqnarray*}
\vcenter{\xymatrix{
[\{A_1 , A_2\},\{B_1 , B_2\}] \ar@{|->}[d]_{\T \p_f(f)} \ar@{|->}[rr]^-{\lambda_{\At}}&& \{[A_1 , B_1],[A_1 , B_2],[A_2 , B_1],[A_2 , B_2]\} \ar@{|->}[d]^{\p_f \T(f)}\\
[\{B_1 , B_2\},\{B_1 , B_2\}] \ar@{|->}[rr]^-{\lambda_{\At}} && \{[B_1 , B_1],[B_1 , B_2],[B_2 , B_1],[B_2 , B_2]\}
}}
\end{eqnarray*}

\begin{remark}Non-commutativity of diagram~\eqref{eq:counterexTground} is essentially due to the idempotence of the powerset constructor. Thus, in order to achieve naturality, multisets instead of lists would also work, as both structures are not idempotent. We chose lists in conformity with implementations of SLD-resolution (for instance PROLOG) where atoms are processed following their order in the goal.
\end{remark}

\begin{convention}\label{conv:ListToPowGround} As motivated above, we will develop our framework modeling goals as lists instead of sets. This requires a mild reformulation of the constructions presented in Section \ref{ssec:groundcase}, with $\li$ taking the role of the inner occurrence of $\p_f$: throughout this and the next section we suppose that a ground logic program $\mb{P}$ is encoded as a coalgebra $p \: \At \to \p_f\li (\At)$ (instead of $p \: \At \to \p_f\p_f(\At)$) and its semantics as a map $\bb{\cdot}_p \: \At \to \cof{\p_f\li}(\At)$. It is immediate to see that this is an harmless variation on the framework for coalgebraic logic programming developed so far and all our results still hold. Accordingly, $\wedge\vee$-trees are now elements of $\cof{\p_f\li}(\At)$ (instead of $\cof{\p_f\p_f}(\At)$ as in Section \ref{SEC:Background}) and therefore the children of an $\vee$-node form a list rather than a set. For that reason, while keeping the standard representation of $\wedge\vee$-trees, we shall implicitly assume a left to right ordering amongst children of $\vee$-nodes (\emph{cf.} Example~\ref{ex:representation}).
\end{convention}

We now have all the formal ingredients to define the extension from coalgebraic to bialgebraic operational semantics.

\begin{construction}\label{constr:paralground} Let $\mb{P}$ be a ground logic program and $p \: \At \to \p_f \li (\At)$ be the coalgebra associated with $\mb{P}$. We define the $\p_f$-coalgebra $\pg{p}\: \li \At \to \p_f \li \At$ in the way prescribed by Proposition \ref{prop:bialgfreemonad} as
\[\xymatrix{\li \At \ar[r]^-{\li p} & \li \p_f \li \At \ar[r]^{\lambda_{\li\At}} & \p_f \li \li \At \ar[r]^{\p_f(\lconc)} & \p_f \li \At}.\]
By Proposition \ref{prop:bialgfreemonad}, $(\li \At, \lconc, \pg{p})$ forms a $\lambda$-bialgebra.
\end{construction}

\begin{remark} Equivalently, the map $\pg{p}\: \li \At \to \p_f \li \At$ may be obtained with the following universal construction. Let $\U^{\li} \dashv \fF^{\li}$ be the canonical adjunction between $\set$ and the category $\EM{\li}$ of Eilenberg-Moore algebrae for the monad $\li$. By using the distributive law $\lambda$ we can define an $\li$-algebra structure on $\p_f \li \At$ as $h \df \p_f(\lconc) \circ \lambda_{\li\At} \: \li \p_f \li \At \to \p_f \li \At$. Let $\widetilde{p} \: (\li\At,\lconc) \to (\p_f\li\At,h)$ be the unique extension of $p \: \At \to \p_f\li\At$ to an algebra morphism in $\EM{\li}$ along the adjunction $\U^{\li} \dashv \fF^{\li}$. The map $\pg{p} \: \li \At \to \p_f\li\At$ of Construction \ref{constr:paralground} is simply $\U^{\li}(\widetilde{p})$. By definition it makes the following diagram commute.
\begin{equation}\label{eq:genpowconstrground}
\vcenter{\xymatrix{
\At \ar[d]_{p} \ar[r]^{\eta^{\li}_{\At}} & \li \At \ar[dl]^{\pg{p}} \\
\p_f \li \At
}}
\end{equation}
As investigated in \cite{BonchiPowersetConstruction10}, the same pattern that here leads to $\pg{p}$ from a given $p$ has many relevant instances. Most notably, the standard powerset construction for non-deterministic automata can be presented in a bialgebraic setting, as a mapping $t \mapsto \overline{t}$, where $t$ encodes a non-deterministic automaton on state space $X$, $\overline{t}$ a deterministic one on state space $\p X$ and a diagram analogous to \eqref{eq:genpowconstrground} yields $\overline{t} \circ \eta^{\p}_X = t$. It is worth noticing that the powerset construction performs a \emph{determinization} of the automaton $t$: non-determinism is internalized, that is, an element of the state space $\p X$ models a \emph{disjunction} of states. Our construction follows a different intuition: the non-determinism given by clauses with the same head is preserved and what is internalized is the possible \emph{$\wedge$-parallel behavior} of a computation: indeed, an element of $\li \At$ models a \emph{conjunction} of atoms.
\end{remark}

\begin{example}\label{ex:parallelizationground}
Let $\pg{p} \: \li\At \to \p_f\li\At$ be constructed out of the coalgebra $p \: \At \to \p_f\li\At$ associated with the logic program of Example \ref{ex:ground}. For a concrete grasp on Construction~\ref{constr:paralground}, we compute $\pg{p}([ \predp (b,b) , \predp (b,c)]) \in \p_f\li \At$.
 First, observe that $p\: \At \to \p_f \li \At$ assigns to the atom $\predp (b,b)$ the set $\{[\predq (c)], [\predp (b,a) , \predp (b,c) ] \}$
 and to $\predp (b,c)$ the set $\{ [ \predq(a), \predq(b), \predq(c) ] \}$.
 It is therefore easy to see that $\li p$ maps $[ \predp (b,b) , \predp (b,c)]$ into the list of sets of lists $$[\{[\predq (c)], [\predp (b,a) , \predp (b,c) ] \} , \, \{ [ \predq(a), \predq(b), \predq(c) ] \} ]\text{.}$$
 This is mapped by $\lambda_{\li\At}$ into the set of lists of lists
$$ \{ [[\predq (c)], [ \predq(a), \predq(b), \predq(c) ] ], \, [ [\predp (b,a) , \predp (b,c) ], [ \predq(a), \predq(b), \predq(c) ]  ] \}$$
and finally, via $\p_f(\lconc)$ into the set of lists
$$ \{ [\predq (c),  \predq(a), \predq(b), \predq(c)  ], \, [ \predp (b,a) , \predp (b,c) ,  \predq(a), \predq(b), \predq(c)  ] \}$$
that is the value $\pg{p}([ \predp (b,b) , \predp (b,c)])$.

The operational reading of such a computation is that $[ \predp (b,b)
  , \predp (b,c)]$ is a goal given by the conjunction of atoms $\predp
(b,b)$ and $\predp (b,c)$.  The action of $\pg{p}$ tells us that, in
order to refute $[ \predp (b,b) , \predp (b,c)]$, one has to refute
either $[\predq (c), \predq(a), \predq(b), \predq(c) ]$ \emph{or} ${[
    \predp (b,a) , \predp (b,c) , \predq(a), \predq(b), }\predq(c)]$
(respectively first and second element of $\pg{p}([ \predp (b,b) ,
  \predp (b,c)])$).

For later use, it is interesting to observe also that $\pg{p}(\elist) = \{\elist\}$ for any program $p$: this is because $\elist \stackrel{\li p}{\mapsto} \elist \stackrel{\lambda_{\li\At}}{\mapsto} \{\elist\} \stackrel{\p_f(\lconc)}{\mapsto} \{\elist\} $.
\end{example}

\begin{remark}\label{rmk:sosparground}
For yet another view on Construction \ref{constr:paralground}, we remark that the map $\pg{p}$ can be presented in terms of the following SOS-rules, where $l \to l'$ means $l' \in \pg{p}(l)$.
\begin{eqnarray*}
\inference*[l1]{l \in p(A)}{[A] \to l}
\qquad
\inference*[l2]{l_1 \to l_1' & l_2 \to l_2'}{ \ l_1 \lconc l_2 \to l_1' \lconc l_2'\ } \qquad
\inference*[l3 ]{}{\elist \to \elist}
\end{eqnarray*}
Rule $(l1)$ corresponds to commutativity of \eqref{eq:genpowconstrground}. Rule $(l2)$ states that $\pg{p}$ preserves the algebraic structure of $\li \At$ given by list concatenation $\lconc$. Operationally, this means that one step of (parallel) resolution for $l_1 \lconc l_2$ requires both a step for $l_1$ and a step for $l_2$.
Finally, rule $(l3)$ encodes the trivial behavior of $\pg{p}$ on the empty list, as observed in Example \ref{ex:parallelizationground}.
\end{remark}

We now provide a cofree construction for the semantics $\bb{\cdot}_{\pg{p}}$ arising from $\pg{p}$.

\begin{construction}\label{constr:parcofreeground} Let
  $\cof{\p_f}(\li\At)$ denote the cofree $\p_f$-coalgebra on $\li\At$,
  obtained like in Construction \ref{Constr:cofree_ground}. It enjoys a final $\li\At \times\p_f (\cdot)$-coalgebra structure $c \: \cof{\p_f}(\li\At) \xrightarrow{\cong} \li\At \times \p_f(\cof{\p_f}(\li\At))$. We now want to lift its universal property to the setting of bialgebrae, showing that $\cof{\p_f}(\li\At)$ forms a final $\lambda'$-bialgebra, where $\lambda'\: \li(\li\At \times \p_f (\cdot)) \To \li\At \times  \p_f \li (\cdot)$ is a distributive law of the monad $\li$ over the functor $\li\At \times \p_f (\cdot)$ defined by
      \[\xymatrix{\li(\li\At \times \p_f (X)) \ar[rr]^{<\li\pi_1,\li \pi_2>} && \li\li\At \times \li\p_f (X) \ar[rr]^{\lconc \times \lambda_X} && \li\At \times \p_f \li (X)} \text{.}\]
For this purpose, we apply the construction of Proposition \ref{prop:finalcoalgBialg}. First, using finality of $c \: \cof{\p_f}(\li\At) \xrightarrow{\cong} \li\At \times  \p_f(\cof{\p_f}(\li\At))$ we define an $\li$-algebra $\lconcc\: \li(\cof{\p_f}(\li\At)) \to \cof{\p_f}(\li\At)$ as follows:
\[\xymatrix{
\li(\cof{\p_f}(\li\At)) \ar[d]_{\li c} \ar@{-->}[rrr]^{\lconcc} &&& \cof{\p_f}(\li\At) \ar[dd]^{c}\\
\li(\li\At \times \p_f(\cof{\p_f}(\li\At))) \ar[d]_{\lambda'_{\cof{\p_f}(\li\At)}} &&& \\
\li\At \times \p_f(\li(\cof{\p_f}(\li\At))) \ar[rrr]^{\id_{\li\At} \times \p_f(\lconcc) } &&& \li\At \times \p_f(\cof{\p_f}(\li\At))
}\]
This algebraic structure yields the final $\lambda'$-bialgebra $\big(\cof{\p_f}(\li\At), \lconcc, c)$ as guaranteed by Proposition \ref{prop:finalcoalgBialg}.
Now we turn to the definition of the semantics $\bb{\cdot}_{\pg{p}} \: \li\At \to \cof{\p_f}(\li\At)$. First, observe that the $\lambda$-bialgebra $(\li\At,\lconc,\pg{p})$ canonically extends to a $\lambda'$-bialgebra by letting $<\id_{\li\At},\pg{p}>$ be the $\li\At \times \p_f(\cdot)$-coalgebra structure on $\li\At$. We can then use finality of $\big(\cof{\p_f}(\li\At), \lconcc, c)$ to obtain the unique $\lambda'$-bialgebra morphism $\bb{\cdot}_{\pg{p}} \: \li\At \to \cof{\p_f}(\li\At)$ making the following diagram commute:
\begin{equation}\label{diag:bialgebraSemGround}
\vcenter{\xymatrix@R=20pt{
\li\li\At \ar[d]_{\lconc} \ar[rrr]^{\li \bb{\cdot}_{\pg{p}}} &&& \li(\cof{\p_f}(\li\At)) \ar[d]^{\lconcc} \\
\li\At \ar[d]_{<\id_{\li\At},\pg{p}>} \ar@{-->}[rrr]^{\bb{\cdot}_{\pg{p}}} &&& \cof{\p_f}(\li\At) \ar[d]^{c} \\
\li\At \times \p_f(\li\At) \ar[rrr]^{\id_{\li\At} \times \p_f(\bb{\cdot}_{\pg{p}})} &&& \li\At \times \p_f(\cof{\p_f}(\li\At)).
}}
\end{equation}
One can check that the above construction of the map $\bb{\cdot}_{\pg{p}}$ can be equivalently obtained by building a cone with vertex $\li \At$ on the terminal sequence generating $\cof{\p_f}(\li\At)$, analogously to Construction \ref{Constr:coalgComonad_ground}.
\end{construction}

Since $\bb{\cdot}_{\pg{p}}$ is given as a morphism of bialgebrae instead of plain coalgebrae, it will preserve the algebraic structure of $\li\At$. This allow us to state the motivating compositionality property of bialgebraic semantics where, intuitively, list concatenation $\lconc$ models the conjunction $\wedge$ described in the introduction.

\begin{theorem}[$\wedge$-Compositionality] \label{thm:ANDcompo} Given lists $l_1,\dots,l_k$ of atoms in $\At$:
\begin{equation*}
    \bb{l_1 \lconc \dots \lconc l_k}_{\pg{p}} = \bb{l_1}_{\pg{p}}\lconcc\dots \lconcc \bb{l_k}_{\pg{p}}.
\end{equation*}
Where $l_1 \lconc \dots \lconc l_k$ is notation for $\lconc ([l_1,\dots,l_k]) = \mu^{\li}_{\At} ([l_1,\dots,l_k])$ and $\bb{l_1}_{\pg{p}}\lconcc\dots \lconcc \bb{l_k}_{\pg{p}}$ for $\lconcc([\bb{l_1}_{\pg{p}},\dots,\bb{l_k}_{\pg{p}}])$. \end{theorem}
\begin{proof} The statement is given by the following derivation:
\begin{align*}
\bb{l_1 \lconc \dots \lconc l_k}_{\pg{p}} =\ & \bb{\cdot}_{\pg{p}} \circ \lconc ([l_1,\dots,l_k]) \\
=\ & \lconcc\circ \li \bb{\cdot}_{\pg{p}}([l_1,\dots,l_k]) \\
=\ & \lconcc([\bb{l_1}_{\pg{p}},\dots,\bb{l_k}_{\pg{p}}]) \\
=\ & \bb{l_1}_{\pg{p}}\lconcc\dots \lconcc \bb{l_k}_{\pg{p}}.
\end{align*}
The first and the last equality are just given by unfolding notation. The second equality amounts to commutativity of the top square in diagram \eqref{diag:bialgebraSemGround} and the third one is given by definition of $\li$ on the function $\bb{\cdot}_{\pg{p}}$.
\end{proof}

\begin{example}\label{ex:comp_ground}
Recall from Example \ref{ex:parallelizationground}, that $\pg{p}([\predp (b,b)])= \{[ \predq(c) ], [\predp (b,a), \predp (b,c) ] \}$ and
$\pg{p}([\predp (b,c)])=\{ [\predq (a), \predq (b) , \predq (c) ] \}$. Below we represent the values $\bb{\predp (b,b)}_{\pg{p}}$ (on the left) and $\bb{\predp (b,c)}_{\pg{p}}$ (on the right) as trees.
{\scriptsize
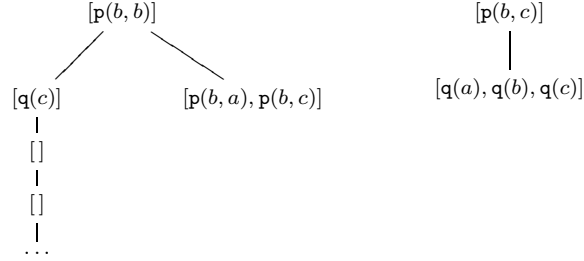
\[
\xymatrix@C=.1cm@R=.2cm{
                 & \ar@{-}[ldd] [ \predp (b,b) ] \ar@{-}[rdd] & & &\\ \\
  [ \predq(c) ] \ar@{-}[d]  &                              & [\predp (b,a), \predp (b,c) ] \\
  \elist \ar@{-}[d] \\
  \elist \ar@{-}[d] \\
  \dots
  }
\qquad
\xymatrix@C=.1cm@R=.5cm{
     [ \predp (b,c) ] \ar@{-}[d] \\
     [\predq (a), \predq (b) , \predq (c) ]
     }
\]
}
The idea is that edges represent the transitions of the rule system of Remark \ref{rmk:sosparground}. The node $[\predq(c)]$ has a child labeled with the empty list, since $p(\predq (c) )=\{\elist\}$ (\emph{cf.} Example \ref{ex:ground}) and thus, by virtue of rule $(l1)$, $\pg{p}( [\predq (c)] )=\{\elist\}$.
Instead, $[\predq (a), \predq (b) , \predq (c) ]$ has no children because $\pg{p}( [\predq (a), \predq (b) , \predq (c) ] )$ is empty. That follows by the fact that $\pg{p}(\predq (a) )$ --- and, in fact, also $\pg{p}(\predq (b) )$ --- is empty and thus we cannot trigger rule $(l2)$ to let $[\predq (a), \predq (b) , \predq (c) ]$ make a transition. Intuitively, this means that $[\predq (a), \predq (b) , \predq (c) ]$ cannot be refuted because not all of its atoms have a refutation. A similar consideration holds for $[\predp (b,a), \predp (b,c) ]$.

\smallskip

\noindent In Example \ref{ex:parallelizationground}, we have shown that $$\pg{p}( [\predp (b,b) , \predp (b,c)]  ) = \{ [\predq (c),  \predq(a), \predq(b), \predq(c)  ], \, [ \predp (b,a) , \predp (b,c) ,  \predq(a), \predq(b), \predq(c)  ] \}$$
and, by similar arguments to those above, we can show that both $\pg{p}( [\predq (c),  \predq(a), \predq(b), \predq(c)  ] )$ and
$\pg{p} ( [ \predp (b,a) , \predp (b,c) ,  \predq(a), \predq(b), \predq(c)  ] )$ are empty. Therefore, we can depict $\bb{  \predp (b,b) , \predp (b,c) }_{\pg{p}}$ as the tree below.
{\scriptsize
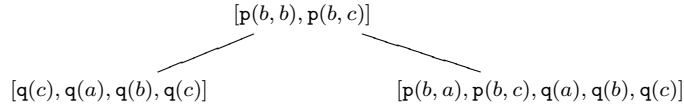
 \[
\xymatrix@C=.1cm@R=.5cm{
                 & \ar@{-}[ld] [ \predp (b,b) , \predp (b,c) ] \ar@{-}[rd] & & &\\
  [\predq (c),\predq (a),\predq (b),\predq (c)]  &                              & [\predp (b,a), \predp (b,c),\predq (a),\predq (b),\predq (c)]
}
\]
}
By virtue of Theorem \ref{thm:ANDcompo}, such a tree can be computed also by concatenating via $\lconcc$ the tree $\bb{[\predp (b,b)]}_{\pg{p}}$ with the tree $\bb{[\predp (b,c)]}_{\pg{p}}$ depicted above. The operation $\lconcc$ on trees $T_1$ and $T_2$ can also be described as follows.
\begin{enumerate}
 \item If the root of $T_1$ has label $l_1$ and the root of $T_2$ has label $l_2$, then the root of $T_1 \lconcc T_2$ has label $l_1 \lconc l_2$;
 \item If $T_1$ has a child $T_1'$ and $T_2$ has a child $T_2'$, then  $T_1 \lconcc T_2$ has a child $T_1' \lconcc T_2'$.
\end{enumerate}
Observe that such trees are rather different from the $\wedge\vee$-trees introduced in Definition \ref{DEF:and-or_par_tree_ground} (\emph{cf.} also Example \ref{ex:ground}): all nodes are of the same kind (there is no more distinction between $\wedge$-nodes and $\vee$-nodes) and are labeled by lists of atoms (rather than just atoms). In the next section, we will formally introduce such trees under the name of (parallel) $\vee$-trees and show that they provide a sound and complete semantics for ground logic programs.
\end{example}

\section{Soundness and Completeness of Bialgebraic Ground Semantics}\label{sec:complparsem_ground}

 In this section we investigate the relation between the semantics $\bb{\cdot}_p$ and $\bb{\cdot}_{\pg{p}}$, stating a relative soundness and completeness result.

In Section \ref{ssec:groundcase} we observed that $\bb{\cdot}_{p}$ --- of type $\At \to \cof{\p_f\li}(\At)$, following Convention~\ref{conv:ListToPowGround} --- maps an atom $A$ into its $\wedge\vee$-tree (Definition \ref{DEF:and-or_par_tree_ground}). As a first step of our analysis, we provide an analogous operational understanding for the value $\bb{l}_{\pg{p}} \in \cof{\p_f}(\li\At)$ associated with a goal $l \in \li\At$. The resulting notion of \emph{$\vee$-tree} will correspond to the one intuitively given in Example \ref{ex:comp_ground}.

\begin{definition}\label{DEF:or_par_tree_ground}
Given a ground logic program $\mb{P}$ and a list $l \in \li\At$ of atoms, the \emph{(parallel) $\vee$-tree} for $l$ in $\mb{P}$ is the possibly infinite tree $T$ satisfying the following properties:
\begin{enumerate}
  \item Each node in $T$ is labeled with a list of atoms and the root is labeled with $l$.
  \item Let $s$ be a node in $T$ with label $l' = [A_1,\dots,A_k]$. For every list
      $[C_1,\dots,C_k]$ of clauses of $\mb{P}$ such that $H_i = A_i$ for each $C_i = H^i \seq B_1^i,\dots,B_j^i$, $s$ has exactly one child $t$, and viceversa. The node $t$ is labeled with  the list $l_1 \lconc \dots \lconc l_k$, where $l_i = [B_1^i,\dots,B_j^i]$ is the body of clause $C_i$.
\end{enumerate}
\end{definition}

\noindent Differently from $\wedge\vee$-trees, where two kinds of nodes yield a distinction between or- and and-parallelism, $\vee$-trees have only one kind of nodes, intuitively corresponding to or-parallelism. The and-parallelism, which in $\wedge\vee$-trees is given by the branching of and-nodes labeled with an atom, is encoded in $\vee$-trees by the labeling of nodes with lists of atoms. The children of a node labeled with $l$ yield the result of simultaneously matching \emph{each} atom in $l$ with heads in the program (\emph{cf.} rule $(l2)$ in Remark \ref{rmk:sosparground}).\footnote{This synchronous form of computation is what makes $\vee$-trees essentially different from SLD-trees~\cite{Lloyd93}, in which also nodes are labeled with lists, but the parent-child relation describes the unification of a \emph{single} atom in the parent node.}

Analogously to Proposition \ref{prop:adequacy_sat}, it is immediate to check the following observation.

\begin{proposition}[Adequacy]
Given a list of atoms $l \in \li\At$ and a program $\mb{P}$, $\bb{l}_{\pg{p}} \in \cof{\p_f}(\li\At)$ is the $\vee$-tree for $l$ in $\mb{P}$.
\end{proposition}

We are now in position to provide a translation between the two notions of tree associated respectively with the semantics $\bb{\cdot}_p$ and $\bb{\cdot}_{\pg{p}}$.

\begin{construction}\label{rmk:ANDORtoORtrees} There is a canonical representation of $\wedge\vee$-trees as $\vee$-trees given as follows. First recall that the domain $\cof{\p_f \li}(\At)$ of $\wedge\vee$-trees is the final $ \At \times \p_f\li(\cdot)$-coalgebra, say with structure given by $u \: \cof{\p_f \li}(\At) \xrightarrow{\cong} \At \times \p_f \li \big(\cof{\p_f \li}(\At)\big)$ (\emph{cf.} Construction~\ref{Constr:cofree_ground}). Let $d\ \df (\eta^{\li}_{\At} \times \id) \after u$ be the extension of $u$ to an $\li\At \times \p_f\li(\cdot)$-coalgebra. We now define an $\li\At \times \p_f(\cdot)$-coalgebra structure $\pgg{d}$ on $\li (\cof{\p_f \li}(\At))$ in the way prescribed by Proposition~\ref{prop:bialgfreemonad}:
 \[\xymatrix@R=.3cm{
\li (\cof{\p_f \li}(\At)) \ar[dd]_-{\li d} \ar[rr]^{\pgg{d}} && \li\At \times \p_f \li \big(\cof{\p_f \li}(\At)\big)\\
&\rotatebox[origin=c]{90}{\df} & \\
 \li \Big(\li\At \times \p_f\li \big(\cof{\p_f \li}(\At)\big)\Big) \ar[rr]^{\lambda'_{\cof{\p_f \li}(\At)}} && \li\At \times \p_f \li \li \big(\cof{\p_f \li}(\At)\big) \ar[uu]_{\id_{\li\At} \times \p_f(\lconc)}
}\]
  Then we use finality of $\cof{\p_f}(\li\At)$ (the domain of $\vee$-trees) to obtain our representation map $r \df \bb{\cdot}_{\pgg{d}} \after \eta^{\li}_{\cof{\p_f \li}(\At)}$ as in the following commutative diagram.
\begin{equation}\label{diag:representationmapground}
\vcenter{\xymatrix{
\cof{\p_f \li}(\At) \ar[d]_{u} \ar@/_5pc/[dd]_{d} \ar[r]^{\eta^{\li}_{\cof{\p_f \li}(\At)}} & \li (\cof{\p_f \li}(\At)) \ar[ddl]^{\pgg{d}} \ar@{-->}[rr]^{\bb{\cdot}_{\pgg{d}}} && \cof{\p_f}(\li\At) \ar[dd]^{c} \\
\At \times \p_f \li \big(\cof{\p_f \li}(\At)\big) \ar[d]_{ \eta^{\li}_{\At}\times \id }\\
\li\At \times \p_f \li \big(\cof{\p_f \li}(\At)\big) \ar[rrr]^{\id_{\li\At}\times \p_f(\bb{\cdot}_{\pgg{d}})} &&& \li\At \times \p_f \cof{\p_f}(\li\At)
}}
\end{equation}
\end{construction}

The representation map $r$ is a well-behaved translation between the semantics given by $\bb{\cdot}_p$ and the one given by $\bb{\cdot}_{\pg{p}}$, as shown by the following property.
\begin{proposition}\label{prop:reprIsWellBehaved} For all $A\in \At$, $r(\bb{A}_p) = \bb{[A]}_{\pg{p}}$. \end{proposition}
\begin{proof} By construction $\bb{\cdot}_p \: \At \to \cof{\p_f \li}(\At)$ is a $\p_f\li$-coalgebra morphism and thus clearly also an $\li\At \times \p_f\li(\cdot)$-coalgebra morphism. By Proposition \ref{prop:bialgfreemonad} the canonical mapping of an $\li\At \times \p_f\li(\cdot)$-coalgebra to a  $\lambda'$-bialgebra
\[
\xymatrix@R=.5cm{
\At \ar[dd]^{<\eta^{\li}_{\At},p>} &&&& \li\li\At \ar[d]^{\lconc}\\
&& \mapsto && \li\At \ar[d]^{\pgg{<\eta^{\li}_{\At},p>}}\\
\li\At \times \p_f\li(\At)&&&& \li\At \times \p_f\li\At\\
\cof{\p_f \li}(\At) \ar[dd]^{d} &&&& \li\li\cof{\p_f \li}(\At) \ar[d]^{\lconc}\\
&& \mapsto && \li\cof{\p_f \li}(\At) \ar[d]^{\pgg{d}}\\
\li\At \times \p_f\li(\cof{\p_f \li}(\At)) &&&& \li\At \times \p_f\li\cof{\p_f \li}(\At)
}
\]
is functorial, meaning that $\li \bb{\cdot}_p \: \li \At \to \li (\cof{\p_f \li}(\At))$ is a $\lambda'$-bialgebra morphism and thus in particular a morphism of $\li\At \times \p_f(\cdot)$-coalgebrae. By construction $\bb{\cdot}_{\pg{p}}$ and $\bb{\cdot}_{\pgg{d}}$ are also $\li\At \times \p_f(\cdot)$-coalgebra morphisms. Therefore $\bb{\cdot}_{\pgg{d}} \circ \li \bb{\cdot}_p = \bb{\cdot}_{\pg{p}}$ by finality of $\cof{\p_f}(\li\At)$. Precomposing both sides with $\eta^{\li}_{\At}$ yields the statement of the proposition:
 \[\bb{\cdot}_{\pg{p}}\circ\eta^{\li}_{\At} = \bb{\cdot}_{\pgg{d}} \circ \li \bb{\cdot}_p \circ \eta^{\li}_{\At} = \bb{\cdot}_{\pgg{d}} \circ \eta^{\li}_{\cof{\p_f \li}(\At)} \circ \bb{\cdot}_p = r \circ \bb{\cdot}_p.\]
\end{proof}

In order to study soundness and completeness of $\bb{\cdot}_{\pg{p}}$, we give a notion of refutation for $\vee$-trees.

\begin{definition}\label{Def:ParGround_subtree} Let $T \in \cof{\p_f}(\li\At)$ be an $\vee$-tree. A \emph{derivation subtree} of $T$ is a sequence of
nodes $s_1, s_2, \dots $ such that $s_1$ is the root of $T$ and $s_{i+1}$ is a child of $s_i$.
A \emph{refutation subtree} is a finite derivation subtree $s_1, \dots, s_n $ where the last node $s_n$ is labeled with the empty list.
 \end{definition}

In fact, derivation subtrees of $\vee$-trees have no branching: they are just paths starting from the root. This is coherent with the previously introduced notions of subtree (\emph{cf. } Definition \ref{Def:subtree_Power}): there the only branching allowed was given by and-parallelism, which in the case of $\vee$-tress has been internalized inside the node labels. For refutation subtrees, the intuition is that they represents paths of computation where all atoms in the initial goal $[A_1,\dots,A_n]$ have been refuted, whence eventually the current goal becomes the empty list.

\begin{theorem}[Soundness and Completeness] Let $A_1,\dots,A_n$ be atoms in $\At$. Then $\bb{[A_1,\dots,A_n]}_{\pg{p}}$ has a refutation subtree if and only if $\bb{A_i}_{p}$ has a refutation subtree for each $A_i$.
\end{theorem}
\begin{proof} The statement is given by the following derivation:
\begin{align*}
\bb{A_i}_{p} \text{ has a refutation subtree for each }A_i  \text{ iff } & r(\bb{A_i}_{p}) \text{ has a refutation subtree for each }A_i \\
 \text{ iff } & \bb{[A_i]}_{\pg{p}} \text{ has a refutation subtree for each }A_i \\
\text{ iff }  & \bb{[A_1]}_{\pg{p}}\lconcc \dots \lconcc \bb{[A_n]}_{\pg{p}} \text{ has a refutation subtree } \\
\text{ iff } & \bb{[A_1,\dots,A_n]}_{\pg{p}} \text{ has a refutation subtree}
\end{align*}
The first and the third equivalence are given by checking that the property of having a refutation subtree is preserved and reflected both by the representation map $r \: \cof{\p_f\li}(\At) \to \cof{\p_f}(\li\At)$ (\emph{cf.} Construction \ref{rmk:ANDORtoORtrees}) and by the concatenation operation $\lconcc$ on $\vee$-trees. The second equivalence is given by Proposition \ref{prop:reprIsWellBehaved} and the last one by Theorem \ref{thm:ANDcompo}.
\end{proof}

\begin{example}\label{ex:representation}
Consider the ground logic program in Example \ref{ex:ground} and compare the $\wedge\vee$-tree $\bb{\predp (b,b)}_{p}$ with
the $\vee$-tree $\bb{[\predp (b,b)]}_{\pg{p}}$ --- that is, $r(\bb{\predp (b,b)}_{p})$ --- depicted below on the left and right, respectively.

{\scriptsize
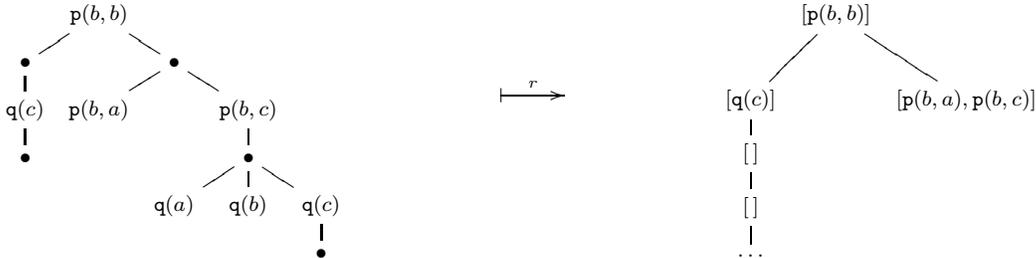
\[
\xymatrix@C=.1cm@R=.2cm{
                 & \ar@{-}[ld] \predp (b,b) \ar@{-}[rd] & & &\\
  \bullet \ar@{-}[d]  &                              & \bullet \ar@{-}[ld] \ar@{-}[rd] & &\\
  \predq (c) \ar@{-}[d] &      \predp (b,a)          &                   & \predp (b,c) \ar@{-}[d] &\\
 \bullet           &                            &                   & \bullet \ar@{-}[ld] \ar@{-}[d] \ar@{-}[rd] & \\
             &                            & \predq (a)        & \predq (b)                   & \predq (c) \ar@{-}[d] \\
             &                            &                   &                              & \bullet
 }
\qquad \qquad \qquad \xymatrix{\\ \ar@{|->}[r]^{r}&}\qquad \qquad \qquad
\xymatrix@C=.1cm@R=.2cm{
                 & \ar@{-}[ldd] [ \predp (b,b) ] \ar@{-}[rdd] & & &\\ \\
  [ \predq(c) ] \ar@{-}[d]  &                              & [\predp (b,a), \predp (b,c) ] \\
  \elist \ar@{-}[d] \\
  \elist \ar@{-}[d] \\
  \dots}
\]}
Following Convention \ref{conv:ListToPowGround}, children of $\vee$-nodes form a list which is implicit in the representation above: for instance, in the tree on the left $\predp (b,a)$ and $\predp (b,c)$ are displayed according to their order in $[\predp (b,a),\predp (b,c)] \in p(\predp (b,b))$.

This convention is instrumental in illustrating the behaviour of the representation map $r$: all the children $t_1,\dots,t_k$ of an $\vee$-nodes (labeled with $\bullet$) of an $\wedge\vee$-trees $T$ are grouped in $r(T)$ into a single node whose label lists all the atoms labeling $t_1,\dots,t_k$.
For instance, the rightmost child of $\bb{ \predp (b,b)}_p$
becomes in $\bb{[ \predp (b,b)]}_{\pg{p}} $ a node labeled with $[\predp (b,a), \predp (b,c)]$.
It is also worth to note that the whole subtree reachable from $\bb{ \predp (b,c)}_p$ is pruned: since $\pg{p}(\predp (b,a))$ is the empty set, then also $\pg{p}([\predp (b,a), \predp (b,c) ])$ is empty. Intuitively, $\predp (b,c)$ should be proved in conjunction with $\predp (b,a)$ which has no proof and therefore the (parallel) resolution of $[\predp (b,a), \predp (b,c) ]$ can stop immediately.

Instead, the unique refutation subtree of $\bb{ \predp (b,b)}_p$ (that is $\predp (b,b) - \bullet - \predq (c) - \bullet$) is preserved in $r(\bb{ \predp (b,b)}_p)$. Indeed,
$[ \predp (b,b)] - [\predq (c)] - \elist$ is a refutation subtree of $\bb{ \predp (b,b)}_{\pg{p}}$.
\end{example}

\section{Bialgebraic Semantics of Goals: the General Case}\label{sec:pargeneral}

In the sequel we generalize the bialgebraic semantics for goals from ground to arbitrary logic programs. Our approach is to extend the saturated semantics $\bb{\cdot}_{p^{\sharp}}$ in $\prsh{\Lw}$ introduced in Section \ref{SEC:SemLogProg}. This will yield compositionality both with respect to the substitutions in the index category $\Lw$ (\emph{cf.} Theorem \ref{thm:compo}) and with respect to the concatenation of goals (extending the result of Theorem \ref{thm:ANDcompo}).

\begin{convention} As we did for ground programs (\emph{cf.} Convention \ref{conv:ListToPowGround}), also for arbitrary logic programs we intend to model goals as lists of atoms. This requires a mild reformulation of our framework for saturated semantics. For this purpose, we introduce the extension $\liftt{\li} \: \prsh{|\Lw|} \to \prsh{|\Lw|}$ of $\li \: \Set \to \Set$ as in Definition \ref{def:liftpreshgen}. Since $\li$ is a monad, then $\liftt{\li}$ is also a monad by Proposition \ref{prop:liftpreservemonads}. A logic program $\mb{P}$ is now encoded in $\prsh{|\Lw|}$ as a coalgebra $p \: \U\At \to \liftt{\p_c}\liftt{\li}\U (\At)$ (instead of $p \: \U\At \to \liftt{\p_c}\liftt{\p_f}\U (\At)$) and its saturation in $\prsh{\Lw}$ as a colgebra $p^{\sharp} \: \At \to \K \liftt{\p_c}\liftt{\li}\U (\At)$. The saturated semantics of $\mb{P}$ (Construction \ref{Constr:coalgComonad_sat}) is given by $\bb{\cdot}_{p^{\sharp}}\: \At \to \cof{\K \liftt{\p_c}\liftt{\li}\U}(\At)$.
All results stated in Section \ref{SEC:SemLogProg}, \ref{SEC:Desaturation} and \ref{SEC:Completeness} continue to hold in this setting.
\end{convention}

Our first task is to generalize Construction \ref{constr:paralground}. The extension $\pa{p}^{\sharp}$ of a saturated logic program $p^{\sharp} \: \At \to \K \liftt{\p_c}\liftt{\li}\U (\At)$ to a bialgebra will depend on a distributive law $\delta$, yet to be defined. The domain of $\pa{p}^{\sharp}$ will be the presheaf $\lift{\li}\At \in \prsh{\Lw}$ of goals, where $\lift{\li} \: \prsh{\Lw} \to \prsh{\Lw}$ extends $\li \: \Set \to \Set$ as in Definition~\ref{def:liftpreshgen}. Since $\li$ is a monad, by Proposition \ref{prop:liftpreservemonads} $\lift{\li}$ is also a monad with unit and multiplication given componentwise by the ones of $\li$.

For ground logic programs, the definition of a bialgebra involves the construction of a $\p_f$-coalgebra out of a $\p_f\li$-coalgebra. For arbitrary logic programs there is a type mismatch, because in $p^{\sharp} \: \At \to \K \liftt{\p_c}\liftt{\li}\U (\At)$ the functor $\lift{\li}$ does not apply to $\At$. However, this can be easily overcome by observing the following general property of the adjunction $\U \dashv \K$.

\begin{proposition}\label{prop:liftingcommuteswithU} Let $\F \: \Set \to \Set$ be a functor and $\liftt{\F}\: \prsh{|\Lw|} \to \prsh{|\Lw|}$, $\lift{\F}\: \prsh{\Lw} \to \prsh{\Lw}$ its extensions to presheaf categories, given as in Definition \ref{def:liftpreshgen}. Then $\liftt{F}\U = \U\lift{\F}$.
\end{proposition}
\begin{proof} Fix $\G \in \prsh{\Lw}$. By definition of $\U$, $\lift{\F}$ and $\liftt{\F}$ the following diagram commutes, where $\iota \: |\Lw| \hookrightarrow \Lw$ is the inclusion functor.
\[\xymatrix{
|\Lw| \ar@/^1.5pc/[rr]^{\liftt{\F}\U\G} \ar[d]_{\iota} \ar[r]^{\U\G} & \Set \ar[r]^{\F} & \Set \\
\Lw \ar[ur]_{\G} \ar@/_1pc/[urr]_{\lift{\F}\G}
}\]
This gives the following derivation: $\liftt{\F}\U\G = \F \U \G = \F \G \iota = \lift{\F} \G \iota = \U \lift{\F}\G$. The reasoning on arrows of $\prsh{\Lw}$ is analogous.
\end{proof}

By Proposition \ref{prop:liftingcommuteswithU}, we can consider $p^{\sharp} \: \At \to \K \liftt{\p_c}\liftt{\li}\U (\At)$ as having the type $p^{\sharp} \: \At \to \K \U \lift{\p_c}\lift{\li} (\At)$. Thus the bialgebra constructed out of $p^{\sharp}$ will be formed by a $\K \U \lift{\p_c}$-coalgebra $\pa{p}^{\sharp} \: \lift{\li}\At \to \K \U \lift{\p_c}\lift{\li} (\At)$. In order to define it by applying Proposition \ref{prop:bialgfreemonad}, the next step is to define $\delta$ as a distributive law of type $\lift{\li} (\K \U \lift{\p_c}) \To (\K \U \lift{\p_c})\lift{\li}$.

For this purpose, our strategy will be to construct $\delta$ as the combination of different distributive laws. First, recall the distributive law $\lambda \: \li \p_f \To \p_f \li$ of monads introduced in Section \ref{sec:parground}: we override notation by calling $\lambda$ the distributive law of type $\li \p_c \To \p_c \li$ defined as the one involving $\p_f$. By Proposition \ref{prop:liftpreservemonads}, $\lambda \: \li \p_c \To \p_c \li$ extends to a distributive law of monads $\lift{\lambda} \: \lift{\li} \lift{\p_c} \To \lift{\p_c} \lift{\li}$ in $\prsh{\Lw}$. Next, we introduce two other distributive laws: one for $\lift{\p_c}$ over $\K\U$ and the other for $\lift{\li}$ over $\K\U$. In fact, because $\K\U$ is a monad arising from the adjunction $\U \dashv \K$ and extensions commute with $\U$ (Proposition \ref{prop:liftingcommuteswithU}), we can let
\begin{eqnarray*}
\varphi \: & \lift{\li} (\K\U)\ \To\  (\K\U)\lift{\li} \\
\psi \: & \lift{\p_c} (\K\U)\  \To\  (\K\U)\lift{\p_c}
\end{eqnarray*}
be defined in a canonical way, using the following general result.

\newcommand{\propDistrLawFromGenLiftings}{
 Let $\catC$ and $\catD$ be categories with an adjunction $\U \dashv \K$ for $ \U\colon \catC \to \catD$ and $\K\colon \catD \to \catC$.
 Let $\lift{\T}\colon \catC \to \catC$ and $\liftt{\T}\colon \catD \to \catD$ be two monads such that $\U \lift{\T} =\liftt{\T} \U$.
 Then, there is a distributive law of monads $\lambda \colon \lift{\T} (\K \U) \Rightarrow (\K \U) \lift{\T}$ defined for all $X\in \catC$ by
} \newcommand{\propDistrLawFromGenLiftingsBis}{
 where $(\cdot)^{\flat}_{X,Z}\: \catC[X,\K Z] \to \catD[\U X, Z]$ and $(\cdot)^{\sharp}_{X,Z}\colon \catD[\U X, Z] \to \catC[X,  \K Z]$ are the components of the canonical bijection given by the adjunction $\U \dashv \K$.
 }
\begin{proposition}\label{prop:distrlawfromgenliftings}
\propDistrLawFromGenLiftings
 \begin{equation}\label{eq:candistrlaw}
 (\liftt{\T}(\id_{\K\U X})^{\flat})^{\sharp}
 \end{equation}
 \propDistrLawFromGenLiftingsBis
\end{proposition}
\begin{proof} See Appendix \ref{Sec:appendix}.
\end{proof}

In our case, $\C = \prsh{\Lw}$, $\catD = \prsh{|\Lw|}$ and the required property of commuting with $\U$ is given for both pairs of monads $\lift{\li},\liftt{\li}$ and $\lift{\p_c},\liftt{\p_c}$ by Proposition \ref{prop:liftingcommuteswithU}. In the sequel we provide the explicit calculation of the distributive law $\varphi \: \lift{\li} \K\U \To  \K\U\lift{\li}$ according to \eqref{eq:candistrlaw}. For this purpose, fix a presheaf $\G \in \prsh{\Lw}$ and $n \in \Lw$. By definition $\varphi_{\G} \: \prsh{\Lw} \to \prsh{\Lw}$ is given on $n \in \Lw$ as the following function (where $\U\lift{\li} = \liftt{\li}\U$ by Proposition \ref{prop:liftingcommuteswithU}):
 \begin{eqnarray*}
  \xymatrix{
 (\liftt{\li}(\id_{\K\U \G})^{\flat})^{\sharp}(n) \: \lift{\li} \K\U \G(n) \ar[rr]^-{\eta_{\lift{\li} \K\U\G}(n)} && \K \U \lift{\li} \K \U \G(n) \ar[rrr]^{\K(\liftt{\li}(\id_{\K\U \G})^{\flat})(n)} &&& \K\U\lift{\li}\G(n)\text{.}
 }
\end{eqnarray*}
We now compute $\varphi_{\G}(n)$ on a list of tuples $[\tuple{x}_1,\dots,\tuple{x}_k]$. This is first mapped onto the value
\begin{eqnarray*}
   \xymatrix@R=5pt{
 [\tuple{x}_1,\dots,\tuple{x}_k]\ \ar@{|->}[rr]^-{\!\eta_{\lift{\li} \K\U\G}(n)\!} && \langle \lift{\li} \K\U \G(\theta) [\tuple{x}_1,\dots,\tuple{x}_k]\rangle_{\theta \: n \to m}\ =\ \langle [<\tuple{x}_1(\sigma \after \theta)>_{\sigma},\dots,<\tuple{x}_k(\sigma \after \theta)>_{\sigma}]\rangle_{\theta}
 }
 \end{eqnarray*}
 where the unit $\eta$ is computed as in \eqref{eq:deinitionUnit} and
 the equality follows by definition of $\K (\U\G)$ on arrows $\theta
 \in \Lw[n,m]$ --- \emph{cf.} \eqref{eq:deinitionKtheta}.  Then we apply $(\id_{\K\U \G})^{\flat}$ componentwise in the list elements
of the tuple:
  \begin{eqnarray*}
   \xymatrix@R=3pt{
 \langle [<\tuple{x}_1(\sigma \after
\theta)>_{\sigma},\dots,<\tuple{x}_k(\sigma \after
\theta)>_{\sigma}]\rangle_{\theta}\  \ar@{|->}[rr]^-{{\scriptstyle
\K(\liftt{\li}(\id_{\K\U \G})^{\flat})(n)}} &&\ \langle
[<\tuple{x}_1(\sigma \after
\theta)>_{\sigma}(\id_n),\dots,<\tuple{x}_k(\sigma \after
\theta)>_{\sigma}(\id_n)]\rangle_{\theta} \ar@{}[d]|{=}\\
  && \langle
[\tuple{x}_1(\theta),\dots,\tuple{x}_k(\theta)]\rangle_{\theta}.
 }
\end{eqnarray*}
By definition of $(\cdot)^{\flat}$, $(\id_{\K\U \G})^{\flat} (n) = \big(\epsilon_{\U\G} \after \U(\id_{\K\U \G}\big)(n) \: \U\K\U \G(n) \to \U \G (n)$ maps a tuple $\tuple{x}$ into its element $\tuple{x}({\id_n})$ --- \emph{cf.} the definition \eqref{EQ:counit} of the counit $\epsilon$. The equality holds because $<\tuple{x}_i(\sigma \after \theta)>_{\sigma}(\id_n) = \tuple{x}_i(\id_n \after \theta) = \tuple{x}_i(\theta)$.

The calculation leading to the definition of $\psi \: \lift{\p_c} \K\U \To  \K\U\lift{\p_c}$ is analogous. In conclusion we obtain the following definitions for distributive laws of monads $\varphi$ and $\psi$:
\begin{eqnarray*}
 \varphi_{\G}(n) \: & \lift{\li} \K\U \G(n) & \to \hspace{0.5cm} \K\U\lift{\li}\G(n) \\
 & [\tuple{x}_1,\dots,\tuple{x}_k]  & \mapsto\ \langle [\tuple{x}_1(\theta),\dots,\tuple{x}_k(\theta)] {\rangle}_{\theta \: n \to m}\\
  \psi_{\G}(n) \: & \lift{\p_c} \K\U \G(n) & \to \hspace{0.5cm} \K\U\lift{\p_c}\G(n) \\
 & \{\tuple{x}_i\}_{i \in I} & \mapsto\ \langle \{\tuple{x}_i(\theta)\}_{i\in I} {\rangle}_{\theta \: n \to m}
\end{eqnarray*}
where $I$ is a countable set of indices. Note that the existence of distributive laws $\varphi$ and $\psi$ implies in particular that $\K\U\lift{\li}$ and $\K\U\lift{\p_c}$ are monads.

We have now all ingredients to define the distributive law $\delta \: \lift{\li} (\K \U \lift{\p_c}) \To (\K \U \lift{\p_c})\lift{\li}$:
\begin{equation}\label{eq:distrlawParGen}
\xymatrix{\delta \: & \lift{\li} \K \U \lift{\p_c} \ar[rr]^{\varphi_{\lift{\p_c}}} && \K \U \lift{\li} \lift{\p_c}  \ar[rr]^{\K\U\lift{\lambda}} && \K \U \lift{\p_c}\lift{\li}}
\end{equation}
Concretely, given $\G \in \prsh{\Lw}$ and $n \in \Lw$, the function $\delta_{\G}(n)$ is defined by
\[\xymatrix@C=10pt{
[<X_1>_{\theta}, \dots , <X_k>_{\theta}] \ar@{|->}[rrr]^<<<<<<<<{\varphi_{\lift{\p_c}\G}(n)} &&& \langle [X_1,\dots,X_k] \rangle_{\theta} \ar@{|->}[rrr]^<<<<<<<<{\K\U\lift{\lambda}_{\G}(n)} &&& \langle \{[x_1,\dots,x_k] \mid x_i \in X_i \}\rangle_{\theta}
}
\]

\begin{proposition} $\delta \: \lift{\li} (\K \U \lift{\p_c}) \To (\K \U \lift{\p_c})\lift{\li}$ is a distributive law of the monad $\lift{\li}$ over the monad $\K \U \lift{\p_c}$.
\end{proposition}
\begin{proof}
In \cite{Cheng_IteratedLaws} it is proven that the natural transformation $\delta$ defined as in \eqref{eq:distrlawParGen} (or, equivalently, the natural transformation $\psi_{\lift{\li}} \after \lift{\p_c}\varphi \: \lift{\p_c}\lift{\li} \K \U \To \K \U \lift{\p_c}\lift{\li}$) is a distributive law yielding the monad $\K \U \lift{\p_c}\lift{\li}$ if one can prove that the three distributive laws $\lift{\lambda}$, $\varphi$ and $\psi$ satisfy a compatibility condition called Yang-Baxter equation. This is given by commutativity of the following diagram, which can be easily verified by definition of $\lift{\lambda}$, $\varphi$ and $\psi$.
\[\xymatrix@R=10pt{
& \lift{\li}\K\U\lift{\p_c} \ar[r]^{\varphi \lift{\p_c}} & \K\U\lift{\li}\lift{\p_c} \ar[dr]^{\K\U\lift{\lambda}} &\\
\lift{\li}\lift{\p_c}\K\U \ar[ur]^{\lift{\li}\psi} \ar[dr]_{\lift{\lambda}\K\U} & & & \K\U\lift{\p_c}\lift{\li}\\
& \lift{\p_c}\lift{\li}\K\U \ar[r]^{\lift{\p_c \varphi}} & \lift{\p_c} \K\U\lift{\li} \ar[ur]_{\psi \lift{\li}} &\\
}\]
\end{proof}

\begin{convention} Throughout the rest of the paper, we do not need to manipulate further the components of the functor $\K\U\lift{\p_c} \: \prsh{\Lw} \to \prsh{\Lw}$ and thus we adopt the shorter notation $\KUP$ for it.
\end{convention}

We are now in position to extend Construction \ref{constr:paralground} to arbitrary logic programs.

\begin{construction} Let $\mb{P}$ be a logic program and $p^{\sharp} \: \At \to \KUP\lift{\li} (\At)$ be the associated coalgebra in $\prsh{\Lw}$. We define the $\KUP$-coalgebra $\pa{p}^{\sharp}\: \li \At \to \KUP\lift{\li} (\At)$ in the way prescribed by Proposition \ref{prop:bialgfreemonad} as
\[\xymatrix{\lift{\li} \At \ar[r]^-{\lift{\li} p^{\sharp}} & \lift{\li} \KUP\lift{\li} \At \ar[r]^{\delta_{\lift{\li}\At}} & \KUP \lift{\li} \lift{\li} \At \ar[rr]^{\KUP(\lconc)} && \KUP \lift{\li} \At}\]
where $\lconc$ is $\mu^{\lift{\li}}_{\At}$, defined componentwise by list concatenation $\mu^{\li}$ since $\lift{\li}$ is the extension of $\li$.
By Proposition \ref{prop:bialgfreemonad}, $(\li \At, \lconc, \pa{p}^{\sharp})$ forms a $\delta$-bialgebra.
\end{construction}

In order to have a more concrete intuition, we spell out the details of the above construction.
%
Fixed $n\in \Lw$, $\lift{\li} p^{\sharp}(n)$ maps a list $[A_1, \dots A_k] \in \li\At(n)$ into the list of tuples
$$[<X_1>_{\theta}, \dots <X_k>_{\theta}]$$
where each $X_i = p^{\sharp}(n)(A_i)(\theta)$ is a set of lists of atoms. 
This is mapped by $\delta_{\lift{\li}\At}$ into the tuple of sets (of lists of lists)
$$<\{[l_1, \dots, l_k] | l_i\in X_i \}  >_{\theta}$$
and, finally, by $\KUP(\lconc)$ into the tuple of sets (of list)
$$<\{l_1\lconc \dots \lconc l_k | l_i\in X_i \}  >_{\theta}\text{.}$$

Alternatively, $\pa{p}^{\sharp}$ can be inductively defined from $p^{\sharp}$ by the following rules, where $l\tr{\theta}l'$ stands for $l'\in \pa{p}^{\sharp}(l)(\theta)$.
$$\inference*[l1]{l' \in p^{\sharp}(A)(\theta) }{[A] \tr{\theta}l'}
\quad \quad
\inference*[l2]{l_1 \tr{\theta}l_1' \quad l_2\tr{\theta}l_2' }{l_1 \lconc l_2 \tr{\theta} l_1' \lconc l_2'}
\quad \quad
\inference*[l3 ]{}{\elist\tr{\theta}\elist}
$$
The rule system extends the one provided for the ground case (\emph{cf.} Remark \ref{rmk:sosparground}) by labeling transitions with the substitution applied on the goal side. Observe that rule $(l2)$ is the same as the one for parallel composition in CSP \cite{roscoe1998theory}: the composite system can evolve only if its parallel components are able to synchronise on some common label $\theta$.

\begin{example}\label{ex:parallelizationGeneral}
Consider the logic program $\NatList$ in Example \ref{Ex:non_compositional} and the atoms $\pNat(x_1)$ and $\pList(\tcons (x_1, x_2))$, both in $\At(2)$.
The morphism 
$p^{\sharp} \: \At \to \K \U \lift{\p_c}\lift{\li} (\At)$ maps these atoms into the tuples defined for all $\theta\in \Lw[2,m]$ as
$$\begin{array}{rcl}
   p^{\sharp}(\pNat(x_1))(\theta) & = & \ \left\{
	\begin{array}{ll}
        \{\elist\}        & \text{if } x_1 \theta   =\tzero \\
	\{[\pNat(t)]\}  & \text{if } x_1 \theta =\tsucc(t)\\	
	\emptyset   & \text{otherwise}
	\end{array}
\right.\\
p^{\sharp}(\pList (\tcons (x_1, x_2)))(\theta) & = & \{[ \pNat (x_1 )\theta, \pList (x_2)\theta]\}
 \end{array}
$$
for some $\Sigma$-term $t$. By application of rule $(l1)$, such tuples are the same of
$\pa{p}^{\sharp}( [\pNat(x_1)] )$ and $\pa{p}^{\sharp}( [\pList (\tcons (x_1, x_2))] )$ and, by mean of rule
$(l2)$, it is easy to compute the value of $\pa{p}^{\sharp}$ on the list of atoms
$[\pNat(x_1), \pList (\tcons (x_1, x_2))]$. 
$$\pa{p}^{\sharp}([\pNat(x_1), \pList (\tcons (x_1, x_2))]) (\theta) = \ \left\{
	\begin{array}{ll}
        \{[ \pNat (\tzero ), \pList (x_2)\theta]\}        & \text{if } x_1 \theta   =\tzero \\
	\{[\pNat(t),\pNat (\tsucc(t) ), \pList (x_2)\theta]\}  & \text{if } x_1 \theta =\tsucc(t)\\	
	\emptyset   & \text{otherwise}
	\end{array}
\right.\\$$
Intuitively, $\pa{p}^{\sharp}$ forces all the atoms of a list to synchronize by choosing a common substitution.
For instance, $[\pList (\tcons (x_1, x_2))]$ can make a transition with any substitution $\theta$ but, when in parallel with
$[\pNat(x_1)]$, it cannot evolve (and thus cannot be refuted) for those substitutions that do not allow $\pNat(x_1)$ to evolve --- i.e., those $\theta$ belonging to the third case above.
\end{example}

We now generalize Construction \ref{constr:parcofreeground} to define the cofree semantics $\bb{\cdot}_{\pa{p}^{\sharp}}$ arising from $\pa{p}^{\sharp}$.

\begin{construction}
The cofree $\KUP$-coalgebra $\cof{\KUP}(\lift{\li}\At)$ on $\lift{\li}\At$, defined following the same steps of Construction \ref{Constr:cofree_sat}, forms the final $\lift{\li}\At\times\KUP(\cdot)$-coalgebra $c \: \cof{\KUP}(\lift{\li}\At) \xrightarrow{\cong} \lift{\li}\At\times\KUP(\cof{\KUP}(\lift{\li}\At))$.
 We now build its canonical extension to a final $\delta'$-bialgebra, where $\delta'\: \lift{\li}(\lift{\li}\At\times\KUP (\cdot)) \To \lift{\li}\At\times\KUP \lift{\li} (\cdot)$ is a distributive law of the monad $\lift{\li}$ over the functor $\lift{\li}\At\times\KUP (\cdot)$ defined in terms of $\delta$:
      \[\delta'_{\G} \: \xymatrix{\lift{\li}(\lift{\li}\At\times\KUP ({\G})) \ar[rr]^{<\lift{\li}\pi_1,\lift{\li} \pi_2>} && \lift{\li}\lift{\li}\At\times\lift{\li}\KUP ({\G}) \ar[rr]^{\lconc \times \delta_{\G}} && \lift{\li}\At\times\KUP \lift{\li} ({\G})} .\]
For this purpose, we construct a $\lift{\li}$-algebra $\lconcc\: \lift{\li}(\cof{\KUP}(\lift{\li}\At)) \to \cof{\KUP}(\lift{\li}\At)$, using finality of $\cof{\KUP}(\lift{\li}\At)$:
\[\xymatrix{
\lift{\li}(\cof{\KUP}(\lift{\li}\At)) \ar[d]_{\lift{\li} c} \ar@{-->}[rrr]^{\lconcc} &&& \cof{\KUP}(\lift{\li}\At) \ar[dd]^{c}\\
\lift{\li}(\lift{\li}\At\times\KUP(\cof{\KUP}(\lift{\li}\At))) \ar[d]_{\delta'_{\cof{\KUP}(\lift{\li}\At)}} &&& \\
\lift{\li}\At\times\KUP(\lift{\li}(\cof{\KUP}(\lift{\li}\At))) \ar[rrr]^{\id_{\lift{\li}\At} \times \KUP(\lconcc)} &&& \lift{\li}\At \times \KUP(\cof{\KUP}(\lift{\li}\At))
}\]
Proposition \ref{prop:finalcoalgBialg} guarantees that $\big(\cof{\KUP}(\lift{\li}\At), \lconcc, c)$ is the final $\delta'$-bialgebra.

We now turn to the definition of the semantics $\bb{\cdot}_{\pa{p}^{\sharp}} \: \lift{\li}\At \to \cof{\KUP}(\lift{\li}\At)$. We let it be the unique $\delta'$-bialgebra morphism from $\lift{\li}\At$ to $\cof{\KUP}(\lift{\li}\At)$, given by finality of $\big(\cof{\KUP}(\lift{\li}\At), \lconcc, c)$:
\begin{equation*}\label{diag:bialgebraSemGen}
\vcenter{\xymatrix{
\lift{\li}\lift{\li}\At \ar[d]_{\lconc} \ar[rrr]^{\lift{\li} \bb{\cdot}_{\pa{p}^{\sharp}}} &&& \lift{\li}(\cof{\KUP}(\lift{\li}\At)) \ar[d]^{\lconcc} \\
\lift{\li}\At \ar[d]_{<\id_{\lift{\li}\At},\pa{p}^{\sharp}>} \ar@{-->}[rrr]^{\bb{\cdot}_{\pa{p}^{\sharp}}} &&& \cof{\KUP}(\lift{\li}\At) \ar[d]^{c} \\
\lift{\li}\At \times \KUP(\lift{\li}\At) \ar[rrr]^{\id_{\lift{\li}\At}\times\KUP(\bb{\cdot}_{\pa{p}^{\sharp}}) } &&& \lift{\li}\At\times\KUP(\cof{\KUP}(\lift{\li}\At)).
}}
\end{equation*}
\end{construction}

The next result states that bialgebraic semantics exhibits two forms of compositionality: it respects both the substitutions in $\Lw$ (by saturation) and the internal structure of goals. In order to formulate such theorem, given a substitution $\theta \in \Lw[n,m]$, we use notation:
\begin{eqnarray*}
{\theta}^{l}\ & \df\ &\lift{\li}(\At)(\theta) \: \lift{\li}(\At)(n) \to \lift{\li}(\At)(m)\\ \overline{\theta}\ & \df\ & \cof{\KUP}(\li\At)(\theta) \colon \cof{\KUP}(\li\At)(n) \to \cof{\KUP}(\li\At)(m).
\end{eqnarray*}

\begin{theorem}[Two-Fold Compositionality] \label{thm:ANDcompoGen} Let $l_1,\dots,l_k$ be list of atoms in $\At(n)$ and $\theta \in \Lw[n, m]$ a substitution. The following two equalities hold:
\begin{eqnarray*}
    \bb{{\theta}^{l}\, l_1 \lconc \dots \lconc {\theta}^{l} \, l_k}_{\pa{p}^{\sharp}} = \overline{\theta}\bb{l_1}_{\pa{p}^{\sharp}}\lconcc\dots \lconcc \overline{\theta}\bb{l_k}_{\pa{p}^{\sharp}} =     \overline{\theta}(\bb{l_1}_{\pa{p}^{\sharp}}\lconcc\dots \lconcc \bb{l_k}_{\pa{p}^{\sharp}}).
\end{eqnarray*}
Where $l_1 \lconc \dots \lconc l_k$ is notation for $\lconc([l_1,\dots,l_k]) = \mu^{\lift{\li}}_{\At}(n) ([l_1,\dots,l_k])$ and $\bb{l_1}_{\pa{p}^{\sharp}}\lconcc\dots \lconcc \bb{l_k}_{\pa{p}^{\sharp}}$ for $\lconcc(n)([\bb{l_1}_{\pa{p}^{\sharp}},\dots,\bb{l_k}_{\pa{p}^{\sharp}}])$. \end{theorem}
\begin{proof} The proof is entirely analogous to the one for the ground case, see Theorem \ref{thm:ANDcompo}. Commutativity with substitutions is given by naturality of $\bb{\cdot}_{\pa{p}^{\sharp}}$, $\lconc$ and $\lconcc$ in $\prsh{\Lw}$.
\end{proof}

\begin{figure}[t]
{\scriptsize
 \[
\xymatrix@C=.005cm@R=.7cm{
 & && [ \pNat(x_{1}) ]   \ar@{-}@(l,u)[lld]|{<\tzeroZ,x_2>} \ar@{-}[d]|{<\tsuccZ(x_1),x_2>} \ar@{-}@(r,u)[dr]^{\dots} & & &\\
       & \elist   \ar@{-}[d]    &     &   [\pNat(x_1)] \ar@{-}[d]|{<\tzeroZ,x_2>} \ar@{-}[rd]|(0.6){<\tzeroZ,\tnilZ>} \ar@{-}[rrd]^{\dots}   & \dots   \\
                     & \elist \ar@{-}[d]          &                    &  \elist \ar@{-}[d]   & \elist \ar@{-}[d] & \dots \\
                     &     \dots                   &                       &   \dots       & \dots %
}
\qquad
\xymatrix@C=.005cm@R=.7cm{
   && [ \pList ( \tconsZ(x_{1},x_2) ) ]  \ar@{-}[d]|{<\tzeroZ,x_2>} \ar@{-}@(dr,u)[dr]|{<\tsuccZ(x_1),x_2>} \ar@{-}@(r,u)[drr]^{\dots} \ar@{-}@(l,u)[dll]|{<x_1,x_2>} & & &\\
  [\pNat(x_1),\pList(x_2)] \ar@{-}[d]|{<\tzeroZ,\tnilZ>} \ar@{-}[dr]^{\dots}  &      & [\pNat(\tzeroZ), \pList(x_{2})] \ar@{-}[ld]_{\dots}   \ar@{-}[d]|{<x_1,\tnilZ>}         &    [\pNat(\tsuccZ(x_1)),\pList(x_2)] \ar@{-}[dr]|(0.6){<\tzeroZ,\tnilZ>} \ar@{-}[d]|{<x_1,\tnilZ>} \ar@{-}[drr]^{\dots}   & \dots   \\
  \elist \ar@{-}[d]  & \dots                       & \elist \ar@{-}[d]                               &   [\pNat(x_1)] \ar@{-}[d]^{\dots}   & [\pNat(\tzeroZ)] \ar@{-}[d] & \dots  \\
 \elist \ar@{-}[d] &                               &  \elist \ar@{-}[d]                                     &   \dots        &  \elist \ar@{-}[d] \\
  \dots        &                              &   \dots                                              &        &   \dots
}
\]
}
\caption{Part of the bialgebraic semantics of $[\pNat(x_1)]$ (left) and $[\pList(\tcons (x_1, x_2))]$ (right) in $\At(2)$. We use the convention that unlabeled edges stand for edges with any substitution.}\label{fig:parallezidesemantics}
\end{figure}

\begin{figure}[t]
{\scriptsize
 \[
\xymatrix@C=.005cm@R=.7cm{
   && [ \pNat(x_1), \pList ( \tconsZ(x_{1},x_2) ) ]  \ar@{-}[d]|{<\tzeroZ,x_2>} \ar@{-}@(dr,u)[dr]|{<\tsuccZ(x_1),x_2>} \ar@{-}@(r,u)[drr]^{\dots}  & & &\\
   &      & [\pNat(\tzeroZ), \pList(x_{2})] \ar@{-}[ld]_{\dots}   \ar@{-}[d]|{<x_1,\tnilZ>}         &    [\pNat(x_1),\pNat(\tsuccZ(x_1)),\pList(x_2)] \ar@{-}[dr]|(0.6){<\tzeroZ,\tnilZ>}  \ar@{-}[drr]^{\dots}   & \dots   \\
    & \dots                       & \elist \ar@{-}[d]                               &     & [\pNat(\tzeroZ)] \ar@{-}[d] & \dots  \\
  &                               &  \elist \ar@{-}[d]                                     &        &  \elist \ar@{-}[d] \\
  &                              &   \dots                                              &        &   \dots
}
\]
}
\caption{Part of the bialgebraic semantics of $[\pNat(x_1), \pList(\tcons (x_1, x_2))]$. }\label{fig:saturatedparallelcomposition}
\end{figure}

\begin{example}\label{ex:compo_general}
In Example \ref{ex:parallelizationGeneral} we have computed the values of $\pa{p}^{\sharp}$ for the lists $[\pNat(x_1)]$ and $[\pList(\tcons(x_1,x_2))]$. Figure \ref{fig:parallezidesemantics} shows (a finite part of) their bialgebraic semantics
$\bb{[\pNat(x_1)]}_{\pa{p}^{\sharp}}$ and $\bb{ [\pList(\tcons(x_1,x_2))] }_{\pa{p}^{\sharp}}$. These are depicted as $\vee$-trees (Definition \ref{DEF:or_par_tree_ground}) where edges are labeled with substitutions. Analogously to the ground case, one can think of the edges as (labeled) transitions generated by the rule presentation of $\pa{p}$ given above.

It is instructive to note that, while $[\pNat(x_1)]$ has one $<\tzero, x_2>$-child, $[\pNat(x_1), \pList(x_2)]$ cannot have a child with such substitution: in order to progress $[\pNat(x_1), \pList(x_2)]$ needs a substitution which makes progress at the same time both $\pNat(x_1)$ and $\pList(x_2)$ like, for instance, $<\tzero, \tnil>$.

\medskip

In Example \ref{ex:parallelizationGeneral} we discussed the value $\pa{p}^{\sharp}([\pNat(x_1), \pList(\tcons(x_1,x_2)) ])$. Figure \ref{fig:saturatedparallelcomposition} shows the bialgebraic semantics $\bb{[\pNat(x_1), \pList(\tcons(x_1,x_2)) ]}_{\pa{p}^{\sharp}}$ for such list. By virtue of Theorem \ref{thm:ANDcompoGen}, it can be equivalently obtained by concatenating via $\lconcc$ the trees  $\bb{[\pNat(x_1)]}_{\pa{p}^{\sharp}}$ and $\bb{ [\pList(\tcons(x_1,x_2))] }_{\pa{p}^{\sharp}}$ in Figure \ref{fig:parallezidesemantics}.
Similarly to the ground case, the operation of concatenating two trees $T_1, T_2$ can be described as follows.
\begin{enumerate}
 \item If the root of $T_1$ has label $l_1$ and the root of $T_2$ has label $l_2$, then the root of $T_1 \lconcc T_2$ has label $l_1 \lconc l_2$;
 \item if $T_1$ has a $\theta$-child $T_1'$ and $T_2$ has a $\theta$-child $T_2'$, then  $T_1 \lconcc T_2$ has a $\theta$-child $T_1' \lconcc T_2'$.
\end{enumerate}
For instance, while $[\pList(\tcons(x_1,x_2))]$ has one $<x_1,x_2>$-child, $[\pNat(x_1),\pList(\tcons(x_1,x_2))]$ has no $<x_1,x_2>$-children because $[\pNat(x_1)]$ has no $<x_1,x_2>$-children. Instead it has one $<\tzero, x_2>$-child labeled with $\elist \lconc [\pNat(\tzero), \pList(x_2)]$ and one $<\tsucc(x_1), x_2>$-child labeled with $[\pNat(x_1)]\lconc[\pNat(\tsucc(x_1)), \pList(x_2)]$. The latter node has no $<x_1,\tnil>$-children because $[\pNat(x_1)]$ has no $<x_1,\tnil>$-children.
\end{example}

\section{Soundness and Completeness of Bialgebraic Semantics}\label{sec:complparsem_gen}

In this section we study the relationship between the bialgebraic semantics $\bb{\cdot}_{\pa{p}^{\sharp}}$ and the other approaches investigated so far. First, analogously to the ground case, we provide an explicit description of the elements of $\cof{\p_f}(\li\At)$ that $\bb{\cdot}_{\pa{p}^{\sharp}}$ associates with goals. This formalizes the notion of tree given in Example \ref{ex:compo_general}.

\begin{definition}\label{DEF:or_par_tree_general}
Given a logic program $\mb{P}$, $n \in \Lw$ and a list $l \in \li\At(n)$ of atoms, the \emph{(parallel) saturated $\vee$-tree} for $l$ in $\mb{P}$ is the possibly infinite tree $\tree$ satisfying the following properties:
\begin{enumerate}
  \item Each node $s$ in $\tree$ is labeled with a list of atoms $l_s \in \li\At(m)$ for some $m\in \Lw$ and the root is labeled with $l$. For any child $t$ of $s$, say labeled with a list $l_t \in \li\At(z)$, the edge from $s$ to $t$ is labeled with a substitution $\sigma \: m \to z$.
  \item Let $s$ be a node in $\tree$ with label $l' = [A_1,\dots,A_k] \in \li\At(m)$.  For all substitutions $\sigma, \tau_1,\dots,\tau_k$ and list $[C_1,\dots,C_k]$ of clauses of $\mb{P}$ such that, for each $C_i = H^i \seq B_1^i,\dots,B_j^i$, $<\sigma,\tau_i>$ is a unifier of $A_i$ and $H^i$, $s$ has exactly one child $t$, and viceversa. Furthermore, the edge connecting $s$ to $t$ is labeled with $\sigma$ and the node $t$ is labeled with  the list $l_1 \lconc \dots \lconc l_k$, where $l_i = [\tau_i B_1^i,\dots, \tau_i  B_j^i]$ is given by applying $\tau_i$ to the body of clause $C_i$.
\end{enumerate}
\end{definition}

\noindent Saturated $\vee$-trees extend the $\vee$-trees introduced in Definition \ref{DEF:or_par_tree_ground} and we can formulate the same adequacy result.

\begin{proposition}[Adequacy]
Given a list of atoms $l \in \li\At(n)$ and a program $\mb{P}$, $\bb{l}_{\pa{p}^{\sharp}} \in \cof{\p_f}(\li\At)(n)$ is the saturated $\vee$-tree for $l$ in $\mb{P}$.
\end{proposition}

It is worth clarifying that the substitutions labeling edges in a saturated $\vee$-tree play the same role as  the substitutions labeling or-nodes in saturated $\wedge\vee$-trees. For any node in $\bb{l}_{\pa{p}^{\sharp}}$, say labeled with $l_1$, an outgoing edge with label $\theta$ connecting to a child $l_2$ means that $l_2$ is an element of $\pa{p}^{\sharp}(l')({\theta})$ --- \emph{cf.} the rule presentation of $\pa{p}^{\sharp}$ in Section \ref{sec:pargeneral}.

In analogy to Construction \ref{rmk:ANDORtoORtrees}, we now set a translation $\reprsat \: \cof{\KUP\lift{\li}}(\At) \to \cof{\KUP}(\lift{\li}\At)$ between saturated $\vee\wedge$-trees (Definition \ref{def:saturatedtree}) and saturated $\vee$-trees.

\begin{construction} Let $u \: \cof{\KUP \lift{\li}}(\At) \xrightarrow{\cong} \KUP \lift{\li} \At \times  \big(\cof{\KUP \lift{\li}}(\At)\big)$ be the final $\At \times \KUP\lift{\li}(\cdot)$-coalgebra structure on  $\cof{\KUP \lift{\li}}(\At)$ and define $d\ \df \eta^{\lift{\li}}_{\At} \times \id \after u$ as the extension of $u$ to a $\lift{\li}\At \times \KUP\lift{\li}(\cdot)$-coalgebra. Next we provide a $\lift{\li}\At \times \KUP(\cdot)$-coalgebra structure $\paa{d}$ on $\lift{\li} (\cof{\KUP \lift{\li}}(\At))$ in the way prescribed by Proposition~\ref{prop:bialgfreemonad}:
 \[\xymatrix@C=1cm@R=.2cm{
\lift{\li} (\cof{\KUP \lift{\li}}(\At)) \ar[dd]_-{\lift{\li} d} \ar[rr]^{\paa{d}} && \lift{\li}\At \times \KUP \lift{\li} \big(\cof{\KUP \lift{\li}}(\At)\big)\\
&\rotatebox[origin=c]{90}{\df} & \\
 \lift{\li} \Big(\lift{\li}\At \times \KUP\lift{\li} \big(\cof{\KUP \lift{\li}}(\At)\big)\Big) \ar[rr]^{\lambda'_{\cof{\KUP \lift{\li}}(\At)}} &&
 \lift{\li}\At \times \KUP \lift{\li} \lift{\li} \big(\cof{\KUP \lift{\li}}(\At)\big) \ar[uu]_{\id_{\lift{\li}\At} \times \KUP(\lconc)}
}\]
  By finality of $\cof{\KUP}(\lift{\li}\At)$ we obtain our representation map $\reprsat \df \bb{\cdot}_{\paa{d}} \after \eta^{\lift{\li}}_{\cof{\KUP \lift{\li}}(\At)}$.
\begin{equation}\label{diag:representationmapground}
\vcenter{\xymatrix{
\cof{\KUP \lift{\li}}(\At) \ar[d]_{u} \ar@/_5pc/[dd]_{d} \ar[r]^{\eta^{\lift{\li}}_{\cof{\KUP \lift{\li}}(\At)}} & \lift{\li} (\cof{\KUP \lift{\li}}(\At)) \ar[ddl]^{\paa{d}} \ar@{-->}[rr]^{\bb{\cdot}_{\paa{d}}} && \cof{\KUP}(\lift{\li}\At) \ar[dd]^{c} \\
\At \times \KUP \lift{\li} \big(\cof{\KUP \lift{\li}}(\At)\big) \ar[d]_{\eta^{\lift{\li}}_{\At} \times \id}\\
\lift{\li} \At \times \KUP \lift{\li} \big(\cof{\KUP \lift{\li}}(\At)\big) \ar[rrr]^{\id_{\lift{\li}\At} \times \KUP(\bb{\cdot}_{\paa{d}})} &&& \lift{\li}\At \times \KUP \cof{\KUP}(\lift{\li}\At)
}}
\end{equation}
\end{construction}

\begin{proposition}\label{prop:correctnessTranslParGen} For any atom $A \in \At(n)$, $\reprsat (\bb{A}_{p^{\sharp}}) = \bb{A}_{\pa{p}^{\sharp}}$.
\end{proposition}
\begin{proof} The proof is entirely analogous to the one provided for the ground case, see Proposition~\ref{prop:reprIsWellBehaved}.
\end{proof}

We now focus on the notion of refutation subtree associated with saturated $\vee$-trees. As outlined in the introduction, here lies one of the main motivations for bialgebraic semantics. When defining refutation subtrees for saturated semantics (Definition \ref{Def:subtree_sat}), we had to require that they were \emph{synched}. The corresponding operational intuition is that, at each derivation step, the proof-search for the atoms in the current goal can only advance by applying to all of them the same substitution. If we did not impose such condition, we would take into account derivation subtrees yielding an unsound refutation, where the same variable is substituted for different values, as shown in Example~\ref{Ex:unsound_AndOrTree}.

The deep reason for requiring such constraint is that, to be sound, the explicit and-parallelism exhibited by saturated $\wedge\vee$-trees has to respect some form of dependency between the substitutions applied on different branches. Coinductive trees (Definition~\ref{Def:coinductive_trees_Power}) achieve it by construction, because all substitutions applied in the goal have to be identities. For saturated $\vee$-trees, this property is also given by construction, but in a more general way: at each step the same substitution (not necessarily the identity) is applied on all the atoms of the goal. This synchronicity property is already encoded in the operational semantics $\pa{p}^{\sharp}$, as immediately observable in its rule presentation, and arises by definition of $\delta$.

By these considerations, subtrees of saturated $\vee$-trees are always synched (in the sense of Definition~\ref{Def:subtree_sat}) and we can define a sound notion of derivation as in the ground case, without the need of additional constraints.

 \begin{definition}\label{Def:ParGen_subtree} Let $\tree$ be a saturated $\vee$-tree in $\cof{\KUP}(\li\At)(n)$ for some $n \in \Lw$. A \emph{derivation subtree} of $\tree$ is a sequence of
nodes $s_1, s_2, \dots $ such that $s_1$ is the root of $\tree$ and $s_{i+1}$ is a child of $s_i$.
A \emph{refutation subtree} is a finite derivation subtree $s_1, \dots, s_k $ where the last element $s_k$ is labeled with the empty list. Its \emph{answer} is the substitution $\theta_{k} \circ \dots \theta_2 \circ \theta_1$, where $\theta_i$ is the substitution labeling the edge between $s_i$ and $s_{i+1}$.
 \end{definition}

The following statement about refutation subtrees will be useful later.
\newcommand{\proplemmasubtreesamesubts}{
Let $\mb{P}$ be a logic program and $l \in \li\At(n)$ a list of atoms.
If $\bb{l}_{\pa{p}^{\sharp}}$ has a refutation subtree with answer $\theta \colon n\to m$, then it has also a refutation subtree $\tree = s_1, s_2, \dots $ with the same answer where the edge connecting $s_1$ and $s_2$ is labeled with $\theta$ and all the other edges in $\tree$ are labeled with $\id_m$.
 }
\begin{proposition}\label{prop:lemmasubtreesamesubts}
\proplemmasubtreesamesubts
\end{proposition}
\begin{proof}
See Appendix \ref{Sec:appendix}.
\end{proof}

\begin{example}\label{ex:parrefutationgluinggen}
The saturated $\vee$-trees for $[\pNat(x_1)]$ and $[\pList(\tcons(x_1,x_2))]$ in Figure \ref{fig:parallezidesemantics} contain several refutation subtrees: for instance,
 \begin{equation}\label{subtree2}
\xymatrix{[\pList(\tcons(x_1,x_2))] \ar@{-}[rr]^-{<\tzero, x_2>} && [\pNat(\tzero), \pList(\tcons(x_2))] \ar@{-}[rr]^-{<x_1,\tnil>} && \elist}
 \end{equation}
 is a refutation subtree of $[\pList(\tcons(x_1,x_2))]$ with answer $ <x_1,\tnil> \circ <\tzero, x_2> =<\tzero, \tnil>$. Observe that $[\pNat(x_1)]$ has a refutation subtree with the same answer: indeed for any substitution $\theta\colon 2 \to m$,
 \begin{equation}\label{subtree1}
\xymatrix{[\pNat(x_1)] \ar@{-}[rr]^-{<\tzero, x_2>} && \elist \ar@{-}[rr]^-{\theta} && \elist}
 \end{equation}
is a refutation subtree with answer $\theta \circ <\tzero, x_2>$. This is because rule $(l3)$ yields $\elist\tr{\theta}\elist$, which is graphically represented in our figures by an unlabeled edge.
%

It is interesting to note that concatenating \eqref{subtree1} (where $\theta = <x_1,\tnil>$) with \eqref{subtree2} via $\lconcc$ yields a refutation subtree for $[\pNat(x_1), \pList(\tcons(x_1,x_2))]$ with the same answer: this is
 \begin{equation}\label{subtree3}
\xymatrix{[\pNat(x_1),\pList(\tcons(x_1,x_2))] \ar@{-}[rr]^-{<\tzero, x_2>} && [\pNat(\tzero), \pList(\tcons(x_2))] \ar@{-}[r]^-{<x_1,\tnil>} & \elist}
 \end{equation}
depicted at the center of Figure \ref{fig:saturatedparallelcomposition}. More generally, if two trees $T_1$ and $T_2$ have refutation subtrees with the same sequence of substitutions $\theta_{k} \circ \dots \theta_2 \circ \theta_1$, then also $T_1\lconcc T_2$ has a refutation subtree with the same sequence.

\medskip

For an example of the behaviour of $r_{sat}$, consider part of the saturated $\wedge \vee$-tree for $\pList ( \tcons(x_{1},x_2) )$ depicted below (Example \ref{ex:thetatree} discusses a different part of the same tree). This is mapped by $r_{sat}$ into the saturated $\vee$-tree of $[\pList ( \tcons(x_{1},x_2) )]$ shown on the right of Figure \ref{fig:parallezidesemantics}.

{\scriptsize
 \[
\xymatrix@C=.01cm@R=.5cm{
 & && \pList ( \tconsZ(x_{1},x_2) )  \ar@{-}@(dl,u)[ld] \ar@{-}@(dr,u)[drr] \ar@{-}@(r,u)[drrrrrr] & & &\\
  &     &        \nodoor{<x_1,x_2>} \ar@{-}[ld] \ar@{-}[rd]    &  & & \nodoor{<\tzeroZ,x_2>}  \ar@{-}[ld] \ar@{-}[rd] & &&& \nodoor{<\tsuccZ(x_1),x_2>}  \ar@{-}[ld] \ar@{-}[rd] \\
&   \pNat(x_1) \ar@{-}[ldd] \ar@{-}[dd] \ar@{-}[rd]    & & \pList(x_2)   \ar@{-}[dd] \ar@{-}[ldd] \ar@{-}[ld]      &   \pNat(\tzeroZ)  \ar@{-}[dd] \ar@{-}[rd] &    & \pList(x_2) \ar@{-}[ld]  \ar@{-}[dd]  &   & \pNat(\tsuccZ(x_1)) \ar@{-}[ldd] \ar@{-}[dd] \ar@{-}[rd] &  & \pList(x_2) \ar@{-}[ld]  \ar@{-}[ldd] \ar@{-}[dd] \\
 & & \dots & & & \dots & & & & \dots & & &
 \\
 \nodoor{<\tzeroZ,x_2>}    &  \nodoor{<\tzeroZ,\tnilZ>}            &   \nodoor{<\tzeroZ,\tnilZ>}  & \nodoor{<x_1,\tnilZ>}     &       \nodoor{<x_1,\tnilZ>}              &    & \nodoor{<x_1,\tnilZ>} & \nodoor{<x_1,\tnilZ>} \ar@{-}[d] & \nodoor{<\tzeroZ,\tnilZ>} \ar@{-}[d] & \nodoor{<x_1,\tnilZ>} & \nodoor{<\tzeroZ,\tnilZ>} && \\
&&&& &&& \pNat(x_1) \ar@{-}[d] & \pNat(\tzeroZ)\ar@{-}[d] \ar@{-}[rd] \\
&&&& &&&         \dots         &  \nodoor{<x_1,x_2>} & \dots }
\]
}

The representation map $r_{sat}$ behaves similarly to the one given for the ground case (\emph{cf.} Example \ref{ex:representation}), the main difference being that, in saturated $\wedge\vee$-trees, $\vee$-nodes are now labeled with substitutions: the effect of $r_{sat}$ is to move such substitutions to the edges of the target saturated $\vee$-trees. For instance, the label $<x_1,x_2>$ on the $\vee$-node (on the left above) is moved in Figure \ref{fig:parallezidesemantics} to the edge connecting the root $[\pList(\tcons(x_1,x_2))]$ with the node labeled with $[\pNat(x_1),\pList(x_2)]$. Observe that this node has one $<\tzero, \tnil>$-child and no children associated with $<\tzero, x_2>$ or $<x_1, \tnil>$: instead, those two substitutions label one child of $\pNat(x_1)$ and one of $\pList(x_2)$, respectively (on the left above). Intuitively, the children of $[\pNat(x_1),\pList(x_2)]$ are given by considering only the children of $\pNat(x_1)$ and those of $\pList(x_2)$ labeled with the same substitution.

It is worth to observe that, for every synched refutation subtree $T'$ (Definition \ref{Def:subtree_sat}) of a saturated $\wedge\vee$-tree $T$ there is a refutation subtree in $r_{sat}(T)$ with the same answer. The effect on $T'$ of applying $r_{sat}$ to $T$ can be described by the following procedure: for every depth in $T'$, (a) all the $\wedge$-nodes are grouped into a single node whose label lists all the labels of these $\wedge$-nodes; (b) the $\vee$-nodes become an edge whose label is the common substitution labeling all these $\vee$-nodes. For an example, we depict below a synched derivation subtree on the left and the corresponding subtree on the right. (In the saturated $\wedge\vee$-tree above there are other three synched refutation subtrees: the reader can find the corresponding refutation subtrees in the $\vee$-tree on the right of Figure \ref{fig:parallezidesemantics}.)
{\scriptsize
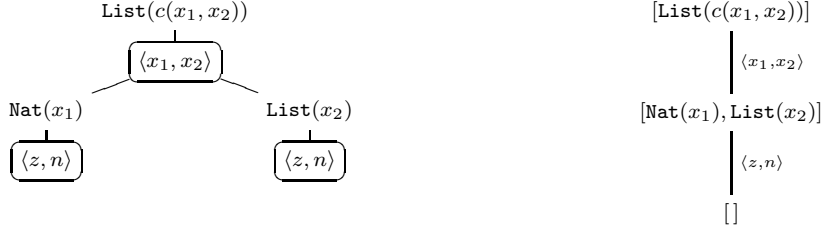
 \[
\xymatrix@C=.01cm@R=.15cm{
 & \pList ( \tconsZ(x_{1},x_2) )  \ar@{-}[d]  \\
&  \nodoor{<x_1,x_2>} \ar@{-}[dl] \ar@{-}[dr] \\
   \pNat(x_1)  \ar@{-}[d] & & \pList(x_2) \ar@{-}[d]  \\
  \nodoor{<\tzeroZ,\tnilZ>}            &    &   \nodoor{<\tzeroZ,\tnilZ>}    \\
}
\qquad \qquad \qquad \qquad \qquad \qquad
\xymatrix@C=.01cm@R=.30cm{
 [\pList ( \tconsZ(x_{1},x_2) )]  \ar@{-}[dd]^{<x_1,x_2>}  \\
 \\
   [\pNat(x_1)  , \pList(x_2)] \ar@{-}[dd]^{<\tzeroZ,\tnilZ>}  \\
    \\
    \elist }
\]
}
Such transformation is neither injective nor surjective. Against surjectivity, consider the refutation subtree \eqref{subtree1}: it does not correspond to any synched derivation subtree because it contains two occurrences of $\elist$, while the above procedure always transforms synched derivation subtrees into refutation subtrees with exactly one occurrence of $\elist$. Nevertheless, using the construction of Proposition \ref{prop:lemmasubtreesamesubts} one can show that for every refutation subtree in $r_{sat}(T)$ there exists a synched refutation subtree in $T$ with the same answer. For instance, \eqref{subtree1} corresponds to the synched refutation subtree with root $\pNat(x_1)$ and the only other node a child of $\pNat(x_1)$ labeled with $\theta \circ <\tzero, x_2>$.
\end{example}

Generalizing the observations of Example \ref{ex:parrefutationgluinggen} above, we have that a saturated $\wedge\vee$-tree $T$ has a synched derivation subtree with answer $\theta$ if and only if $r_{sat}(T)$ has a refutation subtree with answer $\theta$. In case $\tree = \bb{A}_{p^{\sharp}}$ for some atom $A \in \At(n)$, then $\reprsat (\tree) = \bb{[A]}_{\pa{p}^{\sharp}}$ by Proposition \ref{prop:correctnessTranslParGen}. This yields the following statement.

\begin{proposition}\label{prop:translationpreservessubtreesGen} Fix an atom $A \in \At(n)$ and a program $\mb{P}$. The following are equivalent.
\begin{enumerate}[label={(\Roman*)}]
  \item The saturated $\vee$-tree for $[A]  \in \lift{\li}\At(n)$ in $\mb{P}$ has a refutation subtree with answer~ $\theta$. \label{pt:lemmacompletenessPar}
    \item The saturated $\wedge\vee$-tree for $A$ in $\mb{P}$ has a synched refutation subtree with answer~ $\theta$. \label{pt:lemmacompletenessSat}
\end{enumerate}
\end{proposition}

\begin{corollary}[Soundness and Completeness I]\label{Th:CompletenessPar}  Let $\mb{P}$ be a logic program and $A \in \At(n)$ an atom. The following statement is equivalent to any of the three of Theorem \ref{Th:Completeness}.
\begin{enumerate}
\setcounter{enumi}{3}
  \item The saturated $\vee$-tree for $[A]  \in \lift{\li}\At(n)$ in $\mb{P}$ has a refutation subtree with answer~ $\theta$. \label{pt:completenessPar}
\end{enumerate}
\end{corollary}
\begin{proof} It suffices to prove the equivalence between \eqref{pt:completenessPar} and statement \eqref{pt:completenessSat} of Theorem \ref{Th:Completeness}, which is given by Proposition \ref{prop:translationpreservessubtreesGen}.
\end{proof}

In fact, since both bialgebraic semantics and SLD-resolution are defined on arbitrary goals, we can state the following stronger result.

\begin{theorem}[Soundness and Completeness II] \label{th:CompletenessParSLD} Let $\mb{P}$ be a logic program and $G = [A_1,\dots,A_k]  \in \lift{\li}\At(n)$ be a goal. The following are equivalent.
\begin{enumerate}
  \item The saturated $\vee$-tree for $G$ in $\mb{P}$ has a refutation subtree with answer~ $\theta$.
  \item There is an SLD-refutation for $G$ in $\mb{P}$ with correct answer $\theta$.
\end{enumerate}
\end{theorem}
\begin{proof} Fix a program $\mb{P}$ and the goal $G = [A_1,\dots,A_k]$. The statement is given by the following reasoning.
\begin{flalign*}
 &G \text{ has an SLD-refutation }\hspace{-12 pt} & & \\
 &\text{ with correct answer } \theta  & \Leftrightarrow\ & \text{ each } [A_i] \text{ has an SLD-refutation with correct answer } \theta \\
 & \hspace{.8cm}\text{(Corollary \ref{Th:CompletenessPar}) } & \Leftrightarrow\ & \text{ each } \bb{[A_i]}_{\pa{p}^{\sharp}}\text{ has a refutation subtree with answer } \theta\\
 & & \Leftrightarrow\ &
  \bb{[A_1]}_{\pa{p}^{\sharp}} \lconcc \dots \lconcc \bb{[A_k]}_{\pa{p}^{\sharp}}\text{ has a refutation subtree with answer }\theta \\
  & \hspace{.8cm}\text{(Theorem \ref{thm:ANDcompoGen}) } & \Leftrightarrow\ &
  \bb{G}_{\pa{p}^{\sharp}}\text{ has a refutation subtree with answer }\theta.
\end{flalign*}
The first equivalence is a basic fact implied by the definition of SLD-resolution. The~third equivalence comes from the observation that, like in the ground case, $\lconcc$ preserves and reflects the property of yielding a refutation. To see this, suppose that $\bb{[A_k]}_{\pa{p}^{\sharp}},\dots,\bb{[A_k]}_{\pa{p}^{\sharp}}$ all have refutation subtrees with answer $\theta \: n \to m$. By Proposition \ref{prop:lemmasubtreesamesubts}, for each $i$ we can pick a refutation subtree $\tree_i$ of $\bb{[A_i]}_{\pa{p}^{\sharp}}$ with answer $\theta$, such that $\theta$ is the substitution labeling the edge connecting depth $0$ and $1$ and all the other edges are labeled with $\id_m$. Therefore, at each depth, $\tree_1,\dots,\tree_k$ all have the same substitution labeling the corresponding edge. This means that the operation $\lconcc$ applied to $\bb{[A_1]}_{\pa{p}^{\sharp}}, \dots, \bb{[A_k]}_{\pa{p}^{\sharp}}$ has the effect of ``gluing together'' $\tree_1,\dots,\tree_k$ into a refutation subtree $\tree$ of $\bb{[A_1]}_{\pa{p}^{\sharp}} \lconcc \dots \lconcc \bb{[A_k]}_{\pa{p}^{\sharp}}$ with answer $\theta$, just as in Example \ref{ex:parrefutationgluinggen}, where \eqref{subtree1} and \eqref{subtree2} were glued to form \eqref{subtree3}. Starting instead with a refutation subtree in $\bb{[A_1]}_{\pa{p}^{\sharp}} \lconcc \dots \lconcc \bb{[A_k]}_{\pa{p}^{\sharp}}$, one clearly has a decomposition in the converse direction.
\end{proof}

\section{Conclusions}

The first part of this work proposed a coalgebraic semantics for logic programming, extending the framework introduced in \cite{KomMcCuskerPowerAMAST10} for the case of ground logic programs. Our approach has been formulated in terms of coalgebrae on presheaves, whose nice categorical properties made harmless to reuse the very same constructions as in the ground case. A critical point of this generalization was to achieve compositionality with respect to substitutions, which we obtained by employing \emph{saturation} techniques.
We emphasized how these can be explained in terms of substitution mechanisms: while the operational semantics $p$ proposed in \cite{KomPowCALCO11} is associated with term-matching, its saturation $p^{\sharp}$ corresponds to unification.
The map $p^{\sharp}$ gave rise to the notion of \emph{saturated $\wedge\vee$-tree}, as the model of computation represented in our semantics. We observed that coinductive trees, introduced in \cite{KomPowCALCO11}, can be seen as a desaturated version of saturated $\wedge\vee$-trees, and we compared the two notions with a translation. Eventually, we tailored a notion of subtree (of a saturated $\wedge\vee$-tree), called \emph{synched derivation subtree}, representing a sound derivation of a goal in a program. This led to a result of soundness and completeness of our semantics with respect to SLD-resolution.

In the second part of the paper, we extended our framework to model the saturated semantics of goals with bialgebrae on presheaves. The main feature of this approach was yet another form of compositionality: the semantics of a goal $G$ can be equivalently expressed as the ``pasting'' of the semantics of the single atoms composing $G$. This property arose naturally via universal categorical constructions based on monads and distributive laws. The corresponding operational description was given the name of \emph{saturated $\vee$-trees}. The synchronisation of different branches of a derivation subtree, which was imposed on saturated $\wedge\vee$-trees, is now given by construction: in saturated $\vee$-trees the parallel resolution of each atom in the goal always proceeds with the same substitution. On the base of these observations, we extended the soundness and completeness result for saturated semantics to the SLD-resolution of arbitrary (and not just atomic) goals.

Saturated $\vee$-trees carry more information than traditional denotational models like Herbrand or $C$-models \cite{FalaschiPalamidessi_LP}. The latter can be obtained by saturated $\vee$-trees as follows: a substitution $\theta$ is an answer of the saturated $\vee$-tree of a goal $G$ if and only if $G \theta$ belongs to the minimal $C$-model. For future work, we would like to find the right categorical machinery to transform saturated $\vee$-trees into $C$-models. The approach should be close to the one used in \cite{DBLP:conf/cmcs/BonchiM0Z14} for the semantics of automata with $\epsilon$-transitions: first, the branching structure of $\vee$-trees is flattened into sets of sequences of substitutions (similarly to passing from bisimilarity to trace equivalence); second, the substitutions in a sequence are composed to form a single substitution (similarly to composing a sequence of words to form a single word).

Moreover, we find of interest to investigate \emph{infinite} computations and the semantics of coinductive logic programming \cite{GuptaBMSM07}. These have been fruitfully explored within the approach based on coinductive trees \cite{KomPowerCSL11}. We expect our analysis of the notion of synchronisation for derivation subtrees to bring further insights on the question.
%


\bibliographystyle{abbrv}
\bibliography{catBib}

\newpage \appendix 

\section{Proofs}\label{Sec:appendix}

In this appendix we collect the proofs of some of the statements in the main text, whose details are not of direct relevance for our exposition.

\begin{trivlist}
\item \textbf{Proposition~\ref{prop:subtree}}.
\propsubtree
\end{trivlist}
\begin{proof}
First, we fix the two following properties, holding for all the $\wedge$-nodes of $\bb{A}_{p^{\sharp}}$.
\begin{itemize}[label=(\ddag)]
 \item[(\dag)] Let $\theta, \theta'$ be two arrows in $\Lw$ such that $\theta' \circ \theta$ is defined. If an $\wedge$-node $s$ has a child $t$ such that (a) the label of $t$ is $\theta$ and (b) $t$ has children labeled with $B_1, \dots, B_n$,
 then $s$ has also another child $t'$ such that (a) the label of $t'$ is $\theta'\circ \theta$ and (b) $t'$ has children labeled with $B_1\theta', \dots, B_n \theta'$.
 \item[(\ddag)] Let $\theta, \theta', \sigma, \sigma'$ be four arrows in $\Lw$ such that $\sigma \circ \theta = \sigma'\circ \theta'$. If an $\wedge$-node labeled with $A'$ has a child $t$ such that (a) the label of $t$ is $\theta$ and (b) $t$ has children labeled with $B_1, \dots, B_n$, then each node labeled with $A'\theta '$ has a child $t'$ such that (a) the label of $t'$ is $\sigma'$ and (b) $t'$ has children labeled with $B_1 \sigma, \dots, B_n \sigma$.
\end{itemize}
Assume that $\bb{A}_{p^{\sharp}}$ has a synched refutation subtree $\tree$ whose $\vee$-nodes are labeled with $\theta_1$, $\theta_3$, \dots, $\theta_{2k+1}$ (where $\theta_i$ is the substitution labeling the $\vee$-nodes of depth $i$). We prove that $\bb{A}_{p^{\sharp}}$ has another synched refutation subtree $\tree'$ whose first $\vee$-node is labeled with $\theta = \theta_{2k+1}\circ \theta_{2k-1}\circ  \dots \circ \theta_1$ and all the other $\vee$-nodes are labeled with identities.

By assumption, the root $r$ has a child (in $\tree$) that is labeled with $\theta_1$. Assume that its children are labeled with $B_1^2 \dots B_{n_2}^2$.
By $(\dag)$, $r$ has another child $t'$ (in $\bb{A}_{p^{\sharp}}$), that (a) is labeled with $\theta$ and (b) has children labeled with $B_1^2 \sigma_3  \dots B_n^2 \sigma_3$ where $\sigma_3 = \theta_{k+1}\circ \theta_{k-1}\circ  \dots \circ \theta_3$. These children form depth $2$ of $\tree'$ (the root $r$ and $t$ form, respectively, depth $0$ and $1$).

We now build the other depths. For an even $i\leq 2k$, let $\sigma_{i+1}$ denote $\theta_{2k+1}\circ \theta_{2k-1}\circ  \dots \circ \theta_{i+1}$ and let $B_1^i, \dots, B_{n_i}^i$ be the labels of the $\wedge$-nodes of $\tree$ at depth $i$. The depth $i$ of $\tree'$ is given by $\wedge$-nodes labeled with $B_1^i \sigma_{i+1}, \dots, B_{n_i}^i\sigma_{i+1}$; the depth $i+1$ by $\vee$-nodes all labeled with $id_m$.
It is easy to see that $\tree'$ is a subtree of $\bb{A}_{p^{\sharp}}$: by assumption the nodes labeled with $B_1^i, \dots, B_{n_i}^i$ have children in $\tree$ all labeled with $\theta_{i+1}$; since $\sigma_{i+3} \circ \theta_{i+1} = id_m \circ \sigma_{i+1}$, by property $(\ddag)$, the nodes labeled with $B_1^i \sigma_{i+1}, \dots, B_{n_i}^i \sigma_{i+1}$ have children (in $\bb{A}_{p^{\sharp}}$) that (a) are labeled with $id_m$ and (b) have children with labels $B_1^{i+2} \sigma_{i+3}, \dots, B_{n_{i+2}}^{i+2}\sigma_{i+3}$.

\medskip

Once we have built $\tree'$, we can easily conclude. Recall that $t'$ (the first $\vee$-node of $\tree'$) is labeled with $\theta$. Following the construction at the end of Section \ref{SEC:SemLogProg}, the root of $\bb{A}_{p^{\sharp}}\overline{\theta}$ has a child that is labeled with $id_m$ and that has the same children as $t'$. Therefore $\bb{A}_{p^{\sharp}}\overline{\theta}$ has a synched refutation subtree with answer $id_m$.
\end{proof}

\medskip

Next we provide a proof of the following statement.

\begin{trivlist}
\item \textbf{Proposition~\ref{prop:distrlawfromgenliftings}}.
 \propDistrLawFromGenLiftings
 \begin{equation*}
 (\liftt{\T}(\id_{\K\U X})^{\flat})^{\sharp}
 \end{equation*}
 \propDistrLawFromGenLiftingsBis
\end{trivlist}

For this purpose, we first need to state some auxiliary preliminaries on monads and distributive laws. Given a monad $(\M , \eta^{\M}, \mu^{\M})$ on a category $\C$, we use $\Kl(\M)$ to denote its Kleisli category and $\J \colon \catC \to \Kl(\M)$ to denote the canonical functor acting as identity on objects and mapping $f\in \catC[X,Y]$ into $\eta_Y \circ f \in \Kl(\K \U)[X,Y]$. The \emph{lifting} of a monad $(\T,\eta, \mu)$ in $\C$ is a monad $(\T',\eta', \mu')$ in $\Kl(\M)$ such that $\J \T = \T'  \J$ and $\eta'$, $\mu'$ are given on $\J X \in \Kl(\M)$ (i.e. $X \in \C$) respectively as $\J(\eta_{X})$ and $\J(\mu_{X})$.

The following ``folklore'' result gives an alternative description of
distributive laws in terms of liftings to Kleisli categories, see e.g. \cite[$\S$4]{tanaka2005pseudo}, \cite{Mulry93}.
\begin{proposition}
\label{LiftProp}
Let $(\M,\eta^\M,\mu^\M)$ be a monad on a category $\C$. For every monad $(\T,\eta^\T,\mu^\T)$ on $\C$, there is a bijective correspondence between liftings of $(\T,\eta^\T,\mu^\T)$ to $\Kl(\M)$ and distributive laws of the monad $\T$ over the monad $\M$.
\end{proposition}

We are now ready to supply the proof of Proposition \ref{prop:distrlawfromgenliftings}.

\begin{proof}[Proof of Proposition \ref{prop:distrlawfromgenliftings}] By Proposition \ref{LiftProp}, it suffices that we define a monad lifting $\lift{\T}$ to $\Kl(\K \U)$. For this purpose, we will use the canonical comparison functor $\FH \colon \Kl(\K\U) \to \catD$ associated with $\Kl(\K\U)$ (see e.g. \cite[$\S$VI.5]{mclane}), enjoying the following property:
\begin{equation}\label{eq:kleisli}
\FH \J =\U.
\end{equation}
Given $X,Y \in |\Kl(\K \U)|$ and $f \in \Kl(\K \U)[X,Y]$, the functor $\FH$ is defined as
$$\begin{array}{lrcl}
  \FH\colon & \Kl(\K\U) & \to     & \catD \\
          & X         & \mapsto & \U X \\
          & X\stackrel{f}{\to}Y & \mapsto & \U X \stackrel{f^{\flat}}{\to} \U Y
\end{array}$$
where $f^{\flat}$ is given by observing that $f$ in $\Kl(\K\U)$ is a morphism $f\colon X \to \K \U Y$ in $\catC$ and using the bijective correspondence $(\cdot)^{\flat}_{X,Z}\colon \catC[X,\K Z] \to \catD[\U X, Z]$.

We have now the ingredients to introduce the monad $\lifts{\T}$ on $\Kl(\K \U)$ that we will later show to be a lifting of $\lift{\T}$.
On objects $\J X$ of $\Kl (\K \U)$, the functor $\lifts{\T} \colon \Kl(\K \U) \to \Kl(\K \U)$ is defined as $\J \lift{\T} X$ (note that all objects in $\Kl (\K \U)$ are of the shape $\J X$ for some $X\in \catC$). For an arrow  $\J X \stackrel{f}{\to} \J Y$ in $\Kl (\K \U)$, (i.e., $X \stackrel{f}{\to} \K \U Y$ in $\catC$), we take $\U X \stackrel{f^{\flat}}{\to} \U Y$ in $\catD$ and apply $\liftt{\T}$ to obtain $\liftt{\T} \U X \stackrel{\liftt{\T} f^{\flat}}{\to} \liftt{\T} \U Y$ which, by assumption, is
$ \U \lift{\T} X \stackrel{\liftt{\T} f^{\flat}}{\to} \U \lift{\T} Y$.
Using the bijective correspondence $(\cdot)^{\sharp}_{X,Z}\colon \catD[UX, Z] \to \catC[ X,  \K Z]$, we obtain $\lift{\T} X \stackrel{(\liftt{\T} f^{\flat})^{\sharp}}{\to} \K \U \lift{\T} Y$ in $\catC$, that is an arrow $\J \lift{\T} X \stackrel{(\liftt{\T} f^{\flat})^{\sharp}}{\to} \J \lift{\T} Y$ in $\Kl(\K \U)$.
We define $\lifts{\T}f$ as $(\liftt{\T} f^{\flat})^{\sharp}$.
The unit of the monad $\eta^{\lifts{\T}}_{X}$ is defined as $X = \J X \stackrel{\J\eta_X^{\lift{\T}}}{\to}\J \lift{\T} X = \lifts{\T} X$.
The multiplication $\mu^{\lifts{\T}}_{X}$ as $\lifts{\T} \lifts{\T} X = \J \lift{\T} \lift{\T} X \stackrel{\J\mu^{\lift{\T}}_X}{\to}\J \lift{\T} X = \lifts{\T}X$.
One can readily check that $(\lifts{\T}, \eta^{\lifts{\T}}, \mu^{\lifts{\T}})$ is a monad.

In order to verify that $(\lifts{\T}, \eta^{\lifts{\T}}, \mu^{\lifts{\T}})$ is a monad lifting of $(\lift{\T}, \eta^{\lift{\T}}, \mu^{\lift{\T}})$,
it only remains to check that $\lifts{\T}\J = \J \lift{\T}$, since the unit and multiplication are simply defined by applying $\J$ to $\eta^{\lift{\T}}$ and $\mu^{\lift{\T}}$.
For objects, it follows from the definition of $\lifts{\T}$. For arrows $f\in \catC[X,Y]$,
\begin{align*}
\lifts{\T}\J X \stackrel{\lifts{\T}\J f}{\to} \lifts{\T}\J Y &= & \J \lift{\T} X \stackrel{(\liftt{\T}(\J f)^{\flat})^{\sharp}}{\to} \J \lift{\T} Y \tag{definition of $\lifts{\T}$}\\
  & = & \J \lift{\T} X \stackrel{(\liftt{\T} \FH \J f)^{\sharp}}{\to} \J \lift{\T} Y & \tag{definition of $\FH$} \\
& = & \J \lift{\T} X \stackrel{(\liftt{\T} \U f)^{\sharp}}{\to} \J \lift{\T} Y & \tag{by \ref{eq:kleisli}} \\
& = & \J \lift{\T} X \stackrel{(\U \lift{\T} f )^{\sharp}}{\to} \J \lift{\T} Y & \tag{by assumption} \\
& = & \J \lift{\T} X \stackrel{(\FH \J \lift{\T} f )^{\sharp}}{\to} \J \lift{\T} Y & \tag{by \ref{eq:kleisli}} \\
& = & \J \lift{\T} X \stackrel{((\J \lift{\T} f)^{\flat} )^{\sharp}}{\to} \J \lift{\T} Y & \tag{definition of $\FH$} \\
& = & \J \lift{\T} X \stackrel{\J \lift{\T} f}{\to} \J \lift{\T} Y. & \tag{$(\cdot)^{\sharp}$ and $(\cdot)^{\flat}$ form a bijection}
\end{align*}

By Proposition \ref{LiftProp}, we thus obtain a distributive law of monads $\lambda \colon \lift{\T} (\K \U) \Rightarrow (\K \U) \lift{\T}$.
Following the correspondence in \cite{tanaka2005pseudo}, this is effectively constructed for all $X\in \catC$ as follows. Let $\id_{\K \U X}$ be the identity on $\catC$ and $\iota_X\colon \J \K U X \to \J X$ be the corresponding arrow in $\Kl(\K \U)$. Then $\lifts{\T}(\iota_X)\colon \J \lift{\T} \K \U X \to \J \lift{\T} X$ in $\Kl(\K \U)$ corresponds to an arrow $\lifts{\T}(\iota_X)\colon  \lift{\T} \K \U X \to \K \U \lift{\T} X$ in $\catC$, which is how $\lambda$ is defined on $X$. Unfolding the definition of $\lifts{\T}$, this means that $\lambda_X = (\liftt{\T}(\id_{\K \U X})^{\flat})^{\sharp}$.
\end{proof}


\begin{trivlist}
\item \textbf{Proposition~\ref{prop:lemmasubtreesamesubts}}.
\proplemmasubtreesamesubts
\end{trivlist}
\begin{proof} The proof follows closely the one of Proposition \ref{prop:subtree}. First, observe that properties analogous to $(\dag)$ and $(\ddag)$ in the proof of Proposition \ref{prop:subtree} hold for $\vee$-trees, since both saturated $\wedge\vee$- and $\vee$-trees are based on unification. We express them using the rule presentation of~$\pa{p}^{\sharp}$:
$$\inference*[l4 ]{l_1 \tr{\sigma}l_2 }{l_1 \tr{\sigma'\after \sigma}l_2 \sigma'}
\quad \quad
\inference*[l5]{{l_1 \tr{\tau}l_2} & {\sigma \after \tau = \sigma' \after \tau'}}{\ \ \ \ \ \ \ l_1\tau' \tr{\sigma'} l_2 \sigma\ \ \ \ \ \ \ }
$$
where $l\theta$ is the result of applying $\theta$ to each atom in $l$. Rule $(l4)$ corresponds to $(\dag)$ and $(l5)$ to $(\ddag)$.

Now suppose that $\tree' = t_1,t_2,\dots,t_k$ is a refutation subtree of $\bb{l}_{\pa{p}^{\sharp}}$, say with answer $\theta = \theta_k \after \dots \after \theta_2 \circ \theta_1$ and where $t_i$ labeled with a list $l_i$. We build $\tree$ inductively as follows. Depth $0$ in $\tree$ is given by the root $t_1$ labeled with $l_1 = l$. By construction $\theta_1$ labels the edge between $t_1$ and $t_2$ in $\tree'$. Thus, by rule $(l4)$, in $\bb{l}_{\pa{p}^{\sharp}}$ the root $t_1$ has also an edge $\theta$ targeting a child $s_2$ with label $\theta_k \after \dots \after \theta_2 l_2$. We let $s_2$ be the node of depth $1$ in $\tree$. Inductively, suppose that we built $\tree$ up to depth $i \ls k$. We know that in $\tree'$ the node $t_i$ is connected to $t_{i+1}$ by an edge labeled with $\theta_i$. Also, by inductive hypothesis the node $s_i$ in $\tree$ is labeled with $\theta_k \after \dots \after \theta_{i} l_{i}$. Since $(\theta_k \after \dots \after \theta_{i+1}) \after \theta_{i} = \id_m \after ( \theta_{k} \after \dots \after \theta_{i})$, then by rule $(l5)$ the node $s_i$ is connected to a node $s_{i+1}$ in $\bb{l}_{\pa{p}^{\sharp}}$ labeled with $\theta_k \after \dots \after \theta_{i+1} l_{i+1}$ via an edge labeled with $\id_m$. We let $s_{i+1}$ be the node of depth $i+1$ in $\tree$. It is clear by construction that $\tree$ is a refutation subtree of $\bb{l}_{\pa{p}^{\sharp}}$ with the required properties.
\end{proof}

\vspace{-50 pt}
\end{document}